\documentclass[9pt]{article}

\usepackage{amsmath,amsfonts,amssymb,amsthm,hyperref,enumerate,multicol,eufrak,bm}
\usepackage{calligra,indentfirst,epsfig,lettrine}
\usepackage{caption}
\usepackage{relsize}
\usepackage{tabls}
\usepackage{amsmath}
\usepackage{eqnarray}
\usepackage{cancel}
\usepackage{array}
\usepackage{longtable}
\calligra

\DeclareFontFamily{U}{mathc}{}
\DeclareFontShape{U}{mathc}{m}{it}%
{<->s*[1.03] mathc10}{}

\usepackage{comment}
\newcommand{\Mod}[1]{\ (\mathrm{mod}\ #1)}

\newcommand{\qbin}[3]{\genfrac{[}{]}{0pt}{}{#1}{#2}_{#3}}
\makeatletter
\newcommand*{\rom}[1]{\expandafter\@slowromancap\romannumeral #1@}
\makeatother
\usepackage{mathtools}
\usepackage{graphicx}
\usepackage{ulem}
\usepackage{subfigure}
\usepackage{xcolor}
\usepackage{color, soul}
\usepackage{mathrsfs}
\DeclareMathAlphabet{\mathpzc}{OT1}{pzc}{m}{it}
\DeclarePairedDelimiter\ceil{\lceil}{\rceil}
\DeclarePairedDelimiter\floor{\lfloor}{\rfloor}
\usepackage{upgreek}

\allowdisplaybreaks
\numberwithin{equation}{section}

\usepackage{comment}
\numberwithin{equation}{section}

\newtheorem{thm}{Theorem}[section]
\newcommand{\etal}{\textit{et al.}}

\newtheorem{prop}{Proposition}[section]
\newtheorem{cor}{Corollary}[section]
\newtheorem{lemma}{Lemma}[section]
\newtheorem{example}{Example}[section]
\newtheorem{remark}{Remark}[section]

\begin{document}

\title{On Eisenstein additive codes over chain rings and\\linear codes over mixed alphabets}

\author{Leijo Jose{\footnote{Email address:~\urlstyle{same}\href{mailto:leijoj@iiitd.ac.in}{leijoj@iiitd.ac.in}}} \thanks{Research support by the National Board for Higher Mathematics (NBHM), India,  under Grant no. 0203/13(46)/2021-R\&D-II/13176,   is gratefully acknowledged.}~ and ~Anuradha Sharma{\footnote{Corresponding Author, Email address:~\urlstyle{same}\href{mailto:anuradha@iiitd.ac.in}{anuradha@iiitd.ac.in}}~\thanks{Research support by the Department of Science and Technology, India, under the Grant no. DST/INT/RUS/RSF/P-41/2021 with TPN 65025, is gratefully acknowledged.} } \\
 {Department of  Mathematics, IIIT-Delhi}\\{New Delhi 110020, India}}
\date{}
\maketitle
\date{}
	\maketitle
\begin{abstract}\label{Abstract}
Let  $\mathcal{R}_e=GR(p^e,r)[y]/\langle g(y),p^{e-1}y^t\rangle$ be a finite commutative chain ring, where  $p$ is a prime number, $GR(p^e,r)$ is the Galois ring of characteristic $p^e$ and rank $r,$ $t$ and $k$ are  positive integers satisfying $1\leq t\leq k$ when $e \geq 2,$ while $t=k$ when $e=1,$ and $g(y)=y^k+p(g_{k-1}y^{k-1}+\cdots+g_1y+g_0)\in GR(p^e,r)[y]$ is an Eisenstein polynomial with $g_0$ as a unit in $GR(p^e,r).$ In this paper, we first establish a duality-preserving 1-1 correspondence between additive codes  over $\mathcal{R}_e$ and $\mathbb{Z}_{p^e}\mathbb{Z}_{p^{e-1}}$-linear  codes, where the character-theoretic dual codes of additive codes  over $\mathcal{R}_e$ correspond to the Euclidean dual codes of $\mathbb{Z}_{p^e}\mathbb{Z}_{p^{e-1}}$-linear  codes, and vice versa. This correspondence gives rise to a method for constructing additive codes over $\mathcal{R}_e$ and their character-theoretic dual codes, as  unlike additive codes over $\mathcal{R}_e,$ $\mathbb{Z}_{p^e}\mathbb{Z}_{p^{e-1}}$-linear  codes can be completely described in terms of generator matrices. We also list additive codes over the chain ring $\mathbb{Z}_4[y]/\langle y^2-2,2y \rangle$ achieving the Plotkin's bound for homogeneous weights, which suggests that additive codes over $\mathcal{R}_e$ is a promising class of error-correcting codes to find optimal codes with respect to the homogeneous metric.  We further provide a method to construct and enumerate all Euclidean self-orthogonal  and self-dual $\mathtt{R}_e\mathtt{R}_{e-1}$-linear codes of an arbitrary block-length,  where $\mathtt{R}_e$ is a finite commutative chain ring of odd characteristic with the maximal ideal $\langle \gamma\rangle $ of nilpotency index $e\geq 2$ and $\mathtt{R}_{e-1} = \mathtt{R}_e/\langle \gamma^{e-1}\rangle$ is the chain ring with the maximal ideal of nilpotency index $e-1.$ By employing this method, we obtain enumeration formulae for all Euclidean self-orthogonal  and self-dual $\mathtt{R}_e\mathtt{R}_{e-1}$-linear codes of an arbitrary block-length.   This also gives rise to a construction method and enumeration formulae for all self-orthogonal and self-dual additive codes over $\mathcal{R}_e,$ where $p$ is an odd prime.  We also obtain an enumeration formula for all complementary-dual additive  codes (or ACD codes in short) over $\mathcal{R}_e.$  Besides this, we  translate the notion of monomial equivalence between additive codes over $\mathcal{R}_e$ to a suitable notion of equivalence between  $\mathbb{Z}_{p^e}\mathbb{Z}_{p^{e-1}}$-linear  codes. With the help of this observation and our enumeration formulae,  we classify all self-orthogonal and self-dual additive codes of lengths $2$ and $3$ over the chain ring $\mathbb{Z}_9[y]/\langle y^2-3,3y\rangle$ up to monomial equivalence by classifying all Euclidean self-orthogonal and self-dual $\mathbb{Z}_9\mathbb{Z}_3$-linear codes of block-lengths $(2,2)$ and $(3,3),$ respectively. We also classify all ACD codes of length $2$ over  $\mathbb{Z}_4[y]/\langle y^2-2,2y\rangle$ and $\mathbb{Z}_9[y]/\langle y^2-3,3y\rangle$ up to monomial equivalence by classifying all Euclidean complementary-dual $\mathbb{Z}_4\mathbb{Z}_2$-linear codes (or Euclidean $\mathbb{Z}_4\mathbb{Z}_2$-LCD codes in short) and all Euclidean complementary-dual $\mathbb{Z}_9\mathbb{Z}_3$-linear codes (or Euclidean $\mathbb{Z}_9\mathbb{Z}_3$-LCD codes in short) of block-length $(2,2),$ respectively.\end{abstract}
 \vspace{-4mm}
\section{Introduction}
\vspace{-2mm}
Additive codes  over the finite field $\mathbb{F}_4$ of order $4$ are introduced and studied by Calderbank \etal~\cite{Calderbank}. They also studied  their dual codes with respect to the trace inner product and provided a method to construct quantum error-correcting  codes from self-orthogonal additive codes over $\mathbb{F}_4.$ Later, Bierbrauer and Edel \cite{Bierbrauer} generalized the theory of additive codes over $\mathbb{F}_4$ to additive  codes (\textit{i.e.,} $\mathbb{F}_q$-linear codes) over $\mathbb{F}_{q^2},$ where $\mathbb{F}_q$ and $\mathbb{F}_{q^2}$ are finite fields of orders $q$ and $q^2,$ respectively. They also constructed families of quantum error-correcting codes with good parameters. Danielsen and Parker \cite{Danielsen} studied self-dual additive codes over $\mathbb{F}_4$ with respect to the Hermitian trace inner product and classified all self-dual additive codes over $\mathbb{F}_4$ of lengths up to $12.$  They also classified all extremal Type II codes of length $14$ over $\mathbb{F}_4$ and showed that  Type I and Type II codes over $\mathbb{F}_4$ with the trivial automorphism group have smallest lengths $9$ and $12,$ respectively. Later, Huffman \cite{Huffman2013} studied additive codes (\textit{i.e.,} $\mathbb{F}_q$-linear codes) over $\mathbb{F}_{q^t},$ where $t\geq2$ is an integer and $\mathbb{F}_{q^t}$ is the finite field of order $q^t.$ By placing ordinary and Hermitian trace inner products,   he also studied their dual codes and  derived the MacWilliams identity, the Gleason-Pierce theorem, the Gleason polynomials,  mass formulae for self-orthogonal and self-dual additive codes over $\mathbb{F}_{q^t}$, the balance principle and the Singleton bound for these codes. He also classified all self-dual additive codes over $\mathbb{F}_{q^t}$ with respect to both ordinary and Hermitian trace inner products in certain specific cases. Since then, much research has been devoted to studying additive codes over various alphabets \cite{Dougherty2022,Dougherty2,Dougherty2023,Huffman2007,Huffman2010}. 

On the other hand, the study of linear codes over finite commutative chain rings  gained a lot of attention after the seminal work   \cite{Hammons}, where  many important binary non-linear codes are  obtained as  Gray images of linear codes over the ring $\mathbb{Z}_4$ of integers modulo $4.$ This motivated several researchers to study linear codes over finite commutative chain rings   \cite{Betty,Dougherty,Norton,Wood,Yadav}. Recently, Mahmoudi and Samei \cite{Samei} studied additive codes of an arbitrary length over the Galois ring $GR(p^e,r)$ of characteristic $p^e$ and rank $r,$  where $p$ is a prime  number,  $e\geq 1$ and $r\geq2$ are  integers, and $\mathbb{Z}_{p^e}$ is the ring of integers modulo $p^e.$ They also obtained a generator matrix for these codes and enumerated all additive codes of an arbitrary length  over $GR(p^e,r).$ They also derived lower and upper bounds on their minimum Hamming weights  and obtained some classes of maximum distance with respect to rank (MDR) additive codes over $GR(p^e,r).$ They also  studied permutation equivalent and decomposable additive codes over $GR(p^e,r).$ Apart from this, they derived the MacWilliams identity and Delsarte's Theorem for additive codes over $GR(p^e,r).$ Cao \etal~\cite{Cao} studied additive cyclic codes of length $N$ over the Galois ring $GR(p^e,m),$ where both $p$ and $m$ are  prime numbers,  $\gcd(N,p)=1$ and $e\geq2$ is an integer. By providing a canonical form decomposition for these codes, they enumerated all additive cyclic codes of length $N$ over $GR(p^e,m).$  They also  obtained a canonical form decomposition of their dual codes with respect to  the ordinary trace inner product. 


Now, let $S,R$ be two finite commutative chain rings  such that $R$ is the Galois extension of $S$ of degree $d \geq 2,$ and let $N$ be a positive integer with $\gcd(N, |S|)=1,$ (throughout this paper, $|\cdot|$ denotes the cardinality function).  When $d$ is a prime number,  Mart\'{i}nez-Moro \textit{et al.} \cite{Edgar} defined  (Galois) additive  cyclic codes of  length $N$ over $R$ as $S$-submodules of $R^N$ that are invariant under the cyclic shift operator. They studied additive  cyclic codes of  length $N$ over $R$ and their dual codes with respect to the ordinary trace inner product on $R^N.$  In a recent work, Jose and Sharma \cite{Jose} studied these codes and their dual codes in the general case when $d \geq 2$ is any integer, not necessarily a prime number. They also provided  canonical form decompositions of additive cyclic codes of length $N$ over $R$ and their dual codes, and characterized all self-orthogonal and self-dual additive cyclic codes and  complementary-dual  additive cyclic codes (or ACD cyclic codes in short) of length $N$ over $R$. They also derived necessary and sufficient conditions under which  a self-dual additive cyclic code of length $N$ over $R$ exists, and enumerated all self-orthogonal and self-dual additive cyclic codes of length $N$ over $R$ by considering the following two cases: (I) both $|R|$ and $r$ are odd, and (II) $q$ is even and $S$ = $\mathbb{F}_q[u]/\langle u^e\rangle$. They also counted  all  ACD cyclic codes of length $N$ over $R.$  

On the other hand, when $S$ is the ring of integers modulo a prime power and $R$ is an Eisenstein extension of $S$ (\textit{i.e.,} in the Eisenstein additivity case),  Mart\'{i}nez-Moro \textit{et al.} \cite{Edgar} defined  the dual codes of (Eisenstein) additive cyclic codes of length $N$ over $R$ with respect to a duality defined via the annihilators of characters (see Wood \cite{Wood}).  To the best of our knowledge,  algebraic structures of  Eisenstein additive codes over $R$ and their character-theoretic dual codes are not thoroughly investigated. In  this paper, we will consider the  Eisenstein additivity case and thoroughly investigate algebraic structures of Eisenstein additive codes over finite commutative chain rings  and their character-theoretic dual codes by establishing a duality-preserving 1-1 correspondence between these codes and linear codes over mixed alphabets of rings of integers modulo prime powers. This correspondence leads to a method to investigate and construct  Eisenstein additive codes over chain rings and their special classes consisting of self-orthogonal, self-dual, and ACD codes, as linear codes over mixed alphabets of chain rings have richer algebraic structures as compared to that of Eisenstein additive codes over chain rings.

Linear codes over mixed alphabets of $\mathbb{Z}_2$ and $\mathbb{Z}_{4}$ are called $\mathbb{Z}_4\mathbb{Z}_2$-linear codes, which  are first introduced and studied by Rif\`{a} and Pujol \cite{Rifa} as abelian translation-invariant propelinear codes. Borges \etal~\cite{Borges2010} thoroughly studied $\mathbb{Z}_4\mathbb{Z}_2$-linear codes and obtained their fundamental parameters. They also obtained standard forms of their generator and parity-check matrices  by defining an appropriate notion of duality for these codes. Additionally, they studied the binary Gray images of $\mathbb{Z}_4\mathbb{Z}_2$-linear codes. Later,  Aydogdu and Siap \cite{Siap2013} generalized $\mathbb{Z}_4\mathbb{Z}_2$-linear codes to $\mathbb{Z}_{2^e}\mathbb{Z}_{2}$-linear codes,  determined standard forms of their generator and parity-check matrices  and derived two upper bounds on their minimum Lee distances. In another work \cite{Siap1}, they further generalized $\mathbb{Z}_{2^e}\mathbb{Z}_{2}$-linear codes to $\mathbb{Z}_{p^e}\mathbb{Z}_{p^s}$-linear codes, where $p$ is a prime number and $e,s$ are positive integers satisfying $e > s.$ They also determined standard forms of their generator and parity-check matrices  and derived  upper bounds on their minimum Lee distances. Borges \etal~\cite{Borges} studied linear and cyclic  codes over mixed alphabets of finite commutative chain rings and  their dual codes. Recently, Jose and Sharma \cite{Jose2} obtained enumeration formulae for all Euclidean and Hermitian complementary-dual linear codes  over mixed alphabets of chain rings, and classified these codes up to monomial equivalence in certain special cases. The reader may refer to \cite{Siap2015,Bajalan2, Bajalan} for other recent works on linear codes over mixed alphabets.

Throughout this paper, let $\mathcal{R}_e=GR(p^e,r)[y]/\langle g(y),p^{e-1}y^t\rangle$ be a finite commutative chain ring, where $p$ is a prime number, $GR(p^e,r)$ is the Galois ring of characteristic $p^e$ and rank $r,$ $t$ and $k$ are positive  integers satisfying $1\leq t\leq k$ when $e \geq 2,$ while $t=k$ when $e=1,$ and $g(y)=y^k+p(g_{k-1}y^{k-1}+\cdots+g_1y+g_0)\in GR(p^e,r)[y]$ is an Eisenstein polynomial with $g_0$ as a unit in $GR(p^e,r).$  The main goal of this paper is to thoroughly investigate algebraic structures of additive codes  over $\mathcal{R}_e$ and their character-theoretic dual codes by establishing a 1-1 correspondence between these codes   and $\mathbb{Z}_{p^e}\mathbb{Z}_{p^{e-1}}$-linear  codes, under which the character-theoretic dual codes of additive codes  over $\mathcal{R}_e$ correspond to  the Euclidean dual codes of $\mathbb{Z}_{p^e}\mathbb{Z}_{p^{e-1}}$-linear  codes, and vice versa. This induces a 1-1 correspondence between self-orthogonal and self-dual additive codes over $\mathcal{R}_e$ and  Euclidean self-orthogonal and   self-dual $\mathbb{Z}_{p^e}\mathbb{Z}_{p^{e-1}}$-linear  codes, respectively. It also induces a   1-1 correspondence between  ACD codes over $\mathcal{R}_e$ and  Euclidean complementary-dual  $\mathbb{Z}_{p^e}\mathbb{Z}_{p^{e-1}}$-linear  codes (or Euclidean $\mathbb{Z}_{p^e}\mathbb{Z}_{p^{e-1}}$-LCD codes in short).    We will further make use of these observations to  provide construction methods and enumeration formulae for all self-orthogonal and self-dual additive codes over $\mathcal{R}_e,$ where $p$ is an odd prime. We will also obtain an enumeration formula for ACD codes over $\mathcal{R}_e.$ These enumeration formulae are useful in the classification of these codes up to monomial equivalence, which we will illustrate in certain special cases.

This manuscript is organized as follows: In Section \ref{prelim}, we state some preliminaries needed to derive our main results. In Section \ref{section3}, we establish a duality-preserving 1-1 correspondence between additive codes  over $\mathcal{R}_e$ and $\mathbb{Z}_{p^e}\mathbb{Z}_{p^{e-1}}$-linear  codes, where the character-theoretic dual codes of additive codes  over $\mathcal{R}_e$ correspond to the Euclidean dual codes of $\mathbb{Z}_{p^e}\mathbb{Z}_{p^{e-1}}$-linear  codes (Theorem \ref{thm0.1}). We also list additive codes over the chain ring $\mathbb{Z}_4[y]/\langle y^2-2,2y \rangle$ achieving the Plotkin's bound for homogeneous weights (see Table \ref{table1}).   In Section \ref{SectionEnumeration}, we provide a construction method and enumeration formulae for all Euclidean self-orthogonal  and self-dual $\mathtt{R}_e\mathtt{R}_{e-1}$-linear codes of an arbitrary block-length,  where $\mathtt{R}_e$ is a finite commutative chain ring of odd characteristic with the maximal ideal $\langle \gamma\rangle $ of nilpotency index $e\geq 2$ and $\mathtt{R}_{e-1} = \mathtt{R}_e/\langle \gamma^{e-1}\rangle$ is the chain ring with the maximal ideal of nilpotency index $e-1$ (Theorems \ref{Thm4.3},\ref{Lemma3.3},\ref{Thm5.3},\ref{lem3.5},\ref{Thm6.3}   and \ref{Thm6.4}). This gives rise to a construction method and enumeration formulae for  all self-orthogonal and self-dual additive codes over $\mathcal{R}_e,$ where $p$ is an odd prime (Corollaries \ref{Corollary3.1}-\ref{Cor3.2}).   In Section \ref{ACD}, we enumerate all ACD codes over $\mathcal{R}_e$ (Corollary \ref{Cor8.1}). In Section \ref{Classification}, we first translate the notion of monomial equivalence between additive codes over $\mathcal{R}_e$ to a suitable notion of equivalence between  $\mathbb{Z}_{p^e}\mathbb{Z}_{p^{e-1}}$-linear  codes (Lemma \ref{lemeq}). With the help of this observation and the enumeration formulae obtained in Section \ref{SectionEnumeration}, we classify all self-orthogonal and self-dual additive codes of lengths $2$ and $3$ over the chain ring $\mathbb{Z}_9[y]/\langle y^2-3,3y\rangle$ up to monomial equivalence by classifying all Euclidean self-orthogonal and self-dual $\mathbb{Z}_9\mathbb{Z}_3$-linear codes of block-lengths $(2,2)$ and $(3,3),$ respectively. Furthermore, by applying the results derived   in Section \ref{ACD}, we classify all  ACD codes  of length $2$ over $\mathbb{Z}_4[y]/\langle y^2-2,2y\rangle$ and $\mathbb{Z}_9[y]/\langle y^2-3,3y\rangle$ up to monomial equivalence by classifying all Euclidean  $\mathbb{Z}_4\mathbb{Z}_2$-LCD codes and Euclidean  $\mathbb{Z}_9\mathbb{Z}_3$-LCD codes of block-length $(2,2),$ respectively.

\vspace{-4mm}
\section{Some preliminaries}\label{prelim}
\vspace{-2mm}
In this section, we will first recall some basic properties of Galois rings and their Eisenstein extensions. We will also state some basic definitions and results in the Character Theory of finite Abelian groups. We will next define linear codes over mixed alphabets of chain rings and their Euclidean dual codes,  and additive codes over chain rings in the Eisenstein additivity case and their character-theoretic dual codes.
\vspace{-3mm}
\subsection{Galois rings and their Eisenstein extensions}\label{prelim1}
\vspace{-1mm}
Let $GR(p^e,r)$ be the Galois ring of characteristic $p^e$ and cardinality $p^{er},$ where $p$ is a prime number, and $e\geq 2$ and $r \geq 1$ are integers. By Theorem 14.6 of \cite{Wan}, we note that the ring $GR(p^e,r)\simeq\mathbb{Z}_{p^e}[x]/\langle h(x)\rangle,$ where $\mathbb{Z}_{p^e}$ is the ring of integers modulo $p^e$ and $h(x)$ is a monic basic irreducible polynomial of degree $r$ over $\mathbb{Z}_{p^e}.$ Further,  each element  $s\in GR(p^e,r)$ can be uniquely expressed as $s=s_0+s_1x+\cdots+s_{r-1}x^{r-1},$ where $s_i\in\mathbb{Z}_{p^e}$ for $0\leq i\leq r-1.$ This gives rise to a $\mathbb{Z}_{p^e}$-module isomorphism $\phi$ from $GR(p^e,r)$ onto $\mathbb{Z}_{p^e}^r,$ defined as $\phi(s)= (s_0,s_1,s_2,\ldots, s_{r-1})$ for all $s=s_0+s_1x+\cdots+s_{r-1}x^{r-1},$ where $s_i\in\mathbb{Z}_{p^e}$ for $0\leq i\leq r-1.$ Thus the Galois ring $GR(p^e,r)$ is a free $\mathbb{Z}_{p^e}$-module of rank $r.$  In an analogous way, each element  $\tilde{s}\in GR(p^{e-1},r)$ can be uniquely expressed as $\tilde{s}=\tilde{s}_0+\tilde{s}_1x+\cdots+\tilde{s}_{r-1}x^{r-1},$ where $\tilde{s}_i\in\mathbb{Z}_{p^{e-1}}$ for $0\leq i\leq r-1.$ This gives rise to a $\mathbb{Z}_{p^{e-1}}$-module isomorphism from $GR(p^{e-1},r)$ onto $\mathbb{Z}_{p^{e-1}}^r,$ defined as $\tilde{s} \mapsto (\tilde{s}_0,\tilde{s}_1,\tilde{s}_2,\ldots, \tilde{s}_{r-1})$ for all $\tilde{s}=\tilde{s}_0+\tilde{s}_1x+\cdots+\tilde{s}_{r-1}x^{r-1},$ where $\tilde{s}_i\in\mathbb{Z}_{p^{e-1}}$ for $0\leq i\leq r-1,$ and we shall denote this isomorphism by $\phi$ itself for our convenience. Thus the Galois ring $GR(p^{e-1},r)$ is also a free $\mathbb{Z}_{p^{e-1}}$-module of rank $r.$ 

 Next, let us consider the quotient ring $\mathcal{R}_e=GR(p^e,r)[y]/\langle g(y),p^{e-1}y^t\rangle, $  where $t$ and $k$ are  integers satisfying $1\leq t\leq k$ when $e \geq 2,$ while $t=k$ when $e=1,$ and $g(y)=y^k+p(g_{k-1}y^{k-1}+\cdots+g_1y+g_0)\in GR(p^e,r)[y]$ is an Eisenstein polynomial with $g_0$ as a unit in $GR(p^e,r).$  By Theorem XVII.5 of \cite{McDonald}, we see that $\mathcal{R}_e$ is a finite commutative chain ring, whose maximal ideal is generated by $y+\langle g(y),p^{e-1}y^t\rangle \in \mathcal{R}_e$ and has nilpotency index $k(e-1) +t.$  The Galois ring $GR(p^e,r)$ is called the coefficient ring of $\mathcal{R}_e.$  The chain ring $\mathcal{R}_e$ is called an Eisenstein extension of  its coefficient ring $GR(p^e,r)$ with invariants $p,e,r,k$ and $t.$   One can easily see  that each element $a\in \mathcal{R}_e$ can be uniquely expressed as $a=a_{0}+a_{1}y+\cdots+a_{t-1}y^{t-1}+a_{t}y^t+\cdots+ a_{k-1}y^{k-1},$ where $a_{i}\in GR(p^e,r)$ for $0\leq i\leq t-1$ and $a_{j}\in GR(p^{e-1},r)$ for $t\leq j\leq k-1.$  This gives rise to a $\mathbb{Z}_{p^e}$-module isomorphism $\Psi$ from  $\mathcal{R}_e$ onto $\mathbb{Z}_{p^e}^{rt}\oplus \mathbb{Z}_{p^{e-1}}^{r(k-t)},$ defined as $\Psi(a)= \bigl(\phi(a_0),\phi(a_1),\ldots,\phi(a_{t-1})~|~\phi(a_t),\phi(a_{t+1}),\ldots,\phi(a_{k-1})\bigr) $ for all $a=a_{0}+a_{1}y+\cdots+a_{t-1}y^{t-1}+a_{t}y^t+a_{t+1}y^{t+1}+\cdots+ a_{k-1}y^{k-1}\in \mathcal{R}_e$ with  $a_{i}\in GR(p^e,r)$ for $0\leq i\leq t-1$ and $a_{j}\in GR(p^{e-1},r) $ for $t\leq j\leq k-1.$ Further, one can easily observe that $\mathcal{R}_e\simeq \mathbb{Z}_{p^e}^{rt} \oplus \mathbb{Z}_{p^{e-1}}^{r(k-t)}$ as $\mathbb{Z}_{p^e}$-modules.
\vspace{-3mm}
 \subsection{Character Theory of finite Abelian groups} \label{prelim2}
 \vspace{-1mm}
 A character of a finite Abelian group $G$ is defined as a group homomorphism from the group $G$ into the multiplicative group $\mathbb{C}^\ast$ of the field of complex numbers with absolute value $1.$ For example, the map $\chi:\mathbb{Z}_{p^e}\rightarrow \mathbb{C}^\ast,$ defined as  $\chi(a)=\xi^a$ for all $a\in\mathbb{Z}_{p^e},$  is a character of the additive group of  $\mathbb{Z}_{p^e},$ where $\xi \in \mathbb{C}^{\ast}$ is a primitive $(p^e)$-th root of unity. If  $\chi_1$ and $ \chi_2 $ are two characters of $G,$ then  the map $\chi_1\chi_2: G \rightarrow \mathbb{C}^{*},$ defined as $(\chi_1\chi_2)(g)=\chi_1(g)\chi_2(g)$ for all $g \in G,$ is also a character of $G$ and is called the product of the characters $\chi_1$ and $\chi_2.$  Further, the set $\widehat{G}$ consisting of all characters of $G$ forms an Abelian group under this multiplication operation. When $G$ is an Abelian group, we note, by Theorem 5.1 of \cite{Huppert}, that the groups $G$ and $\widehat{G}$ are isomorphic. In particular, the character group of a cyclic group is also a cyclic group of the same order. We refer the reader to \cite{Huppert} for more details.

From now on, for a finite Abelian group $G,$ let  $\chi_g $ denote the character of $G$ corresponding to the group element $g \in G,$ \textit{i.e.,} we assume that   $\widehat{G}=\{\chi_g:g\in G\}.$  If $G$ and $H$ are two finite Abelian groups, then for $g \in G$ and $h \in H,$  the map $\chi_{(g,h)}: G \times H \rightarrow \mathbb{C}^{*},$ defined as $\chi_{(g,h)}(g',h')=\chi_g(g')\chi_h(h')$ for all $(g',h') \in G \times H,$ is a character of their direct product $G \times H,$ where $ \chi_g\in\widehat{G}$ and $\chi_h\in\widehat{H}.$ In fact, we see, by Theorem 5.1 of \cite{Huppert},  that the character group of their direct product $G\times H$ is isomorphic to the direct product $\widehat{G}\times\widehat{H}$   under the group isomorphism $\chi_{(g,h)} \mapsto (\chi_g ,\chi_h)$ for all $g \in G$ and $h \in H.$  

In particular, when $G= \mathbb{Z}_{p^e}$ (a cyclic group of order $p^e$), we have $\mathbb{Z}_{p^e}\simeq\widehat{\mathbb{Z}}_{p^e},$ and hence we can assume that $\widehat{\mathbb{Z}}_{p^e}=\{\chi_a:a\in\mathbb{Z}_{p^e}\},$ where for $a\in\mathbb{Z}_{p^e},$ the character $\chi_a \in \widehat{\mathbb{Z}}_{p^e}$ is defined as $\chi_a(b)=\xi^{ab}$ for all $b\in\mathbb{Z}_{p^e}.$ As $\xi\in \mathbb{C}^{\ast}$ is a primitive $(p^e)$-th root of unity, we note that $\xi^{p}$ is a primitive $(p^{e-1})$-th root of unity. In an analogous manner, we see that $\widehat{\mathbb{Z}}_{p^{e-1}}=\{\chi_\alpha:\alpha\in\mathbb{Z}_{p^{e-1}}\},$ where for $\alpha\in\mathbb{Z}_{p^{e-1}},$ the character $\chi_\alpha \in \widehat{\mathbb{Z}}_{p^{e-1}}$ is defined as $\chi_\alpha(\beta)=\xi^{p\alpha\beta}$ for all $\beta\in\mathbb{Z}_{p^{e-1}}.$ 
\vspace{-3mm}
\subsection{Codes over mixed alphabets of chain rings}\label{prelim3}
\vspace{-1mm}
Let $\mathtt{R}_e$ be a finite commutative chain ring with the maximal ideal $\langle\gamma\rangle$  of nilpotency index $e \geq 2$ and the residue field $\overline{\mathtt{R}}_e=\mathtt{R}_e/\langle \gamma \rangle$ of prime power order $q.$ By Theorem XVII.5 of \cite{McDonald}, there exists an element $\zeta_e \in \mathtt{R}_{e}$ of multiplicative order $q-1,$ and the set $\mathcal{T}_e=\{0,1,\zeta_e, \zeta_e^2, \ldots, \zeta_e^{q-2}\}$ is called the  Teichm$\ddot{u}$ller set of $\mathtt{R}_e.$ Further,  by Theorem XVII.5 of \cite{McDonald} again, we see that each element $\mathtt{s} \in \mathtt{R}_{e}$ can be uniquely expressed as $\mathtt{s}=\mathtt{s}_0+\mathtt{s}_1 \gamma +\mathtt{s}_2 \gamma^2+\cdots+\mathtt{s}_{e-1}\gamma^{e-1},$ where $\mathtt{s}_0,\mathtt{s}_1,\ldots, \mathtt{s}_{e-1}\in \mathcal{T}_e.$ 

Throughout this section, we assume that $\mu$ is an integer satisfying $2 \leq 
\mu  \leq e.$ Here, we see that the quotient ring $\mathtt{R}_e/\langle \gamma^{\mu}\rangle$ is also a finite commutative chain ring with the maximal ideal  generated by the element $\gamma+\langle \gamma^{\mu}\rangle $ of nilpotency index $\mu,$ and we will denote this chain ring by $\mathtt{R}_{\mu}$ for our convenience. The residue field of $\mathtt{R}_{\mu}$ is given by $\overline{\mathtt{R}}_{\mu}=\mathtt{R}_{\mu}/\langle \gamma \rangle \simeq \mathtt{R}_e/\langle \gamma \rangle=\overline{\mathtt{R}}_e.$ Since the chain ring $\mathtt{R}_{\mu}$ can be embedded into $\mathtt{R}_e,$ throughout this paper, we will represent elements of $\mathtt{R}_{\mu}$ by their representatives in $\mathtt{R}_e,$ and we will perform addition and multiplication in $\mathtt{R}_{\mu}$ modulo $\gamma^{\mu}.$ Under this assumption, we assume, without any loss of generality, that  the set $\mathcal{T}_e$ is also the  Teichm$\ddot{u}$ller set of $\mathtt{R}_{\mu}$ and that $\overline{\mathtt{R}}_{\mu}=\overline{\mathtt{R}}_e.$  From now on, we will represent elements of the residue field   $\overline{\mathtt{R}}_e$ as $\overline{\mathtt{c}}=\mathtt{c}+\langle \gamma \rangle$ for all $\mathtt{c} \in \mathtt{R}_e.$

In this section, we will state some basic properties of $\mathtt{R}_\mu\mathtt{R}_{\mu-1}$-linear codes and their dual codes.  For positive integers $N_1$ and $N_2,$ a non-empty subset $\mathscr{C}_{\mu}$ of $\mathtt{R}_\mu^{N_1}\oplus\mathtt{R}_{\mu-1}^{N_2}$ is called an $\mathtt{R}_\mu\mathtt{R}_{\mu-1}$-linear code of block-length $(N_1,N_2)$ if it is an $\mathtt{R}_\mu$-submodule of $\mathtt{R}_\mu^{N_1}\oplus\mathtt{R}_{\mu-1}^{N_2}$ \cite{Borges}.  We see, by Proposition 3.2 of Borges \etal~\cite{Borges}, that an $\mathtt{R}_\mu\mathtt{R}_{\mu-1}$-linear code $\mathscr{C}_{\mu}$ of the type $\{k_0,k_1,\ldots,k_{\mu-1};\ell_0,\ell_1,\ldots,\ell_{\mu-2}\}$ and  block-length $(N_1,N_2)$ is permutation equivalent to an $\mathtt{R}_\mu\mathtt{R}_{\mu-1}$-linear code with a generator matrix $\mathtt{G}_{\mu}$  of the form
\vspace{-2mm}\begin{equation}\label{eqn1.1}
\mathtt{G}_\mu=\left[\begin{array}{c|c}
    \mathtt{A}_{\mu} & \mathtt{B}_{\mu} \\ 
    \mathtt{C}_{\mu} & \mathtt{D}_{\mu}
\end{array}\right],\vspace{-2mm}
\end{equation}
where $k_0,k_1,\ldots,k_{\mu-1}, \ell_0,\ell_1,\ldots, \ell_{\mu-2}$ are non-negative integers such that $K=k_0+k_1+\cdots+k_{\mu-1},$ $L=\ell_0+\ell_1+\cdots+\ell_{\mu-2},$ the block matrix $\mathtt{A}_{\mu}$ is a $(K \times N_1)$-matrix over $\mathtt{R}_{\mu}$ whose columns  are grouped into blocks of sizes   $k_0,k_1,\ldots,k_{\mu-1},N_1-K$ and is given by 
\vspace{-2mm}\begin{equation*}
\mathtt{A}_{\mu}=\left[\begin{array}{ccccccc}
 I_{k_0} & \mathtt{A}_{0,1}^{(\mu)} & \mathtt{A}_{0,2}^{(\mu)} & \cdots & \mathtt{A}_{0,\mu-2}^{(\mu)} & \mathtt{A}_{0,\mu-1}^{(\mu)} & \mathtt{A}_{0,\mu}^{(\mu)}\\
 0 & \gamma I_{k_1} & \gamma \mathtt{A}_{1,2}^{(\mu)} & \cdots & \gamma \mathtt{A}_{1,\mu-2}^{(\mu)} & \gamma \mathtt{A}_{1,\mu-1}^{(\mu)} & \gamma \mathtt{A}_{1,\mu}^{(\mu)}\\
 0 & 0 & \gamma^2I_{k_2}  & \cdots & \gamma^2\mathtt{A}_{2,\mu-2}^{(\mu)} & \gamma^2\mathtt{A}_{2,\mu-1}^{(\mu)} & \gamma^2\mathtt{A}_{2,\mu}^{(\mu)} \\
 \vdots  & \vdots & \vdots & \ddots & \vdots & \vdots & \vdots \\
 0 & 0  & 0 & \cdots & \gamma^{\mu-2}I_{k_{\mu-2}} & \gamma^{\mu-2}\mathtt{A}_{\mu-2,\mu-1}^{(\mu)} & \gamma^{\mu-2}\mathtt{A}_{\mu-2,\mu}^{(\mu)}\\
 0 & 0  & 0 & \cdots & 0 & \gamma^{\mu-1}I_{k_{\mu-1}} & \gamma^{\mu-1}\mathtt{A}_{\mu-1,\mu}^{(\mu)}
\end{array}
\right]=\left[\begin{array}{c}
\mathtt{A}_0^{(\mu)}\\
\gamma\mathtt{A}_1^{(\mu)}\\
\gamma^2\mathtt{A}_2^{(\mu)}\\
\vdots\\
\gamma^{\mu-2}\mathtt{A}_{\mu-2}^{(\mu)}\\
\gamma^{\mu-1}\mathtt{A}_{\mu-1}^{(\mu)}
\end{array}
\right],\vspace{-2mm}
\end{equation*}
the block matrix $\mathtt{B}_{\mu}$ is a $(K \times N_2)$-matrix over  $\mathtt{R}_{\mu-1}$ whose columns  are grouped into blocks of sizes   $\ell_0,\ell_1,\ldots,\ell_{\mu-2},N_2-L$ and is given by 
\vspace{-2mm}\begin{equation*}
    \mathtt{B}_{\mu}= \left[\begin{array}{cccccc}
0 & \mathtt{B}_{0,1}^{(\mu)} & \mathtt{B}_{0,2}^{(\mu)} & \cdots & \mathtt{B}_{0,\mu-2}^{(\mu)} & \mathtt{B}_{0,\mu-1}^{(\mu)}\\
0 & 0 & \gamma \mathtt{B}_{1,2}^{(\mu)} &  \cdots & \gamma \mathtt{B}_{1,\mu-2}^{(\mu)} & \gamma \mathtt{B}_{1,\mu-1}^{(\mu)}\\
0 & 0 & 0 &  \cdots & \gamma^2\mathtt{B}_{2,\mu-2}^{(\mu)} & \gamma^2\mathtt{B}_{2,\mu-1}^{(\mu)}\\
\vdots &  \vdots & \vdots & \ddots & \vdots & \vdots\\
0 & 0 & 0 & \cdots & 0 & \gamma^{\mu-2}\mathtt{B}_{\mu-2,\mu-1}^{(\mu)}\\
0 & 0 & 0 & \cdots & 0 & 0 \\
\end{array}
\right]=\left[\begin{array}{c}
\mathtt{B}_0^{(\mu)}\\
\gamma\mathtt{B}_1^{(\mu)}\\
\gamma^2\mathtt{B}_2^{(\mu)}\\
\vdots\\
\gamma^{\mu-2}\mathtt{B}_{\mu-2}^{(\mu)}\\
0
\end{array}
\right],\vspace{-2mm}
\end{equation*}
the block matrix $\mathtt{C}_{\mu}$ is an $(L \times N_1)$-matrix over $\mathtt{R}_\mu$ whose columns  are grouped into blocks of sizes   $k_0,k_1,\ldots,k_{\mu-1},N_1-K$ and is given by 
\vspace{-2mm}\begin{equation*}
    \mathtt{C}_{\mu}=\left[\begin{array}{ccccccc}
0 & 0 & \gamma \mathtt{C}_{0,2}^{(\mu)} & \cdots & \gamma \mathtt{C}_{0,\mu-2}^{(\mu)} & \gamma \mathtt{C}_{0,\mu-1}^{(\mu)} & \gamma \mathtt{C}_{0,\mu}^{(\mu)}\\
0 & 0 & 0 & \cdots & \gamma^2\mathtt{C}_{1,\mu-2}^{(\mu)} & \gamma^2\mathtt{C}_{1,\mu-1}^{(\mu)} & \gamma^2\mathtt{C}_{1,\mu}^{(\mu)}\\
\vdots  & \vdots & \vdots & \ddots & \vdots & \vdots & \vdots\\
0 & 0 & 0 & \cdots & 0 & 0 & \gamma^{\mu-1}\mathtt{C}_{\mu-2,\mu}^{(\mu)}
\end{array}
\right]=\left[\begin{array}{c}
\gamma\mathtt{C}_0^{(\mu)}\\
\gamma^2\mathtt{C}_1^{(\mu)}\\
\vdots\\
\gamma^{\mu-1}\mathtt{C}_{\mu-2}^{(\mu)}
\end{array}
\right],\vspace{-2mm}
\end{equation*} and the block matrix $\mathtt{D}_{\mu}$ is an $(L \times N_2)$-matrix over $\mathtt{R}_{\mu-1}$ whose columns  are grouped into blocks of sizes   $\ell_0,\ell_1,\ldots,\ell_{\mu-2},N_2-L$ and is given by 
\vspace{-2mm}\begin{equation*}
    \mathtt{D}_{\mu}= \left[\begin{array}{cccccc}
I_{\ell_0} & \mathtt{D}_{0,1}^{(\mu)} & \mathtt{D}_{0,2}^{(\mu)} & \cdots & \mathtt{D}_{0,\mu-2}^{(\mu)} & \mathtt{D}_{0,\mu-1}^{(\mu)}\\
0 & \gamma I_{\ell_1} & \gamma \mathtt{D}_{1,2}^{(\mu)} & \cdots & \gamma \mathtt{D}_{1,\mu-2}^{(\mu)} & \gamma \mathtt{D}_{1,\mu-1}^{(\mu)}\\
\vdots &  \vdots & \vdots & \ddots & \vdots & \vdots\\
0 & 0 & 0 & \cdots & \gamma^{\mu-2}I_{\ell_{\mu-2}} & \gamma^{\mu-2}\mathtt{D}_{\mu-2,\mu-1}^{(\mu)}
\end{array}
\right]=\left[\begin{array}{c}
\mathtt{D}_0^{(\mu)}\\
\gamma\mathtt{D}_1^{(\mu)}\\
\vdots\\
\gamma^{\mu-2}\mathtt{D}_{\mu-2}^{(\mu)}
\end{array}
\right].\vspace{-2mm}
\end{equation*}
 Here $I_{k_i}$ is the $k_i \times k_i$ identity matrix over $\mathtt{R}_\mu$ for $0 \leq i \leq \mu-1,$ $I_{\ell_j}$ is the $\ell_j \times \ell_j$ identity matrix over  $\mathtt{R}_{\mu-1}$ for $0 \leq j \leq \mu-2,$ the matrix $\gamma^i\mathtt{A}_{i,j}^{(\mu)} \in M_{k_i\times k_j}(\mathtt{R}_\mu)$ is to be considered modulo $\gamma^{j}$ for $0\leq i <  j\leq \mu-1,$ the matrix $\gamma^{i+1}\mathtt{C}_{i,j}^{(\mu)} \in M_{\ell_i\times k_j}(\mathtt{R}_\mu)$ is to be considered modulo $\gamma^{j}$ for $0\leq i \leq  \mu-3$ and $ i+2\leq j\leq \mu-1,$  the matrices $\gamma^i\mathtt{B}_{i,j}^{(\mu)}\in M_{k_i\times \ell_j}(\mathtt{R}_{\mu-1})$ and $\gamma^i\mathtt{D}_{i,j}^{(\mu)}\in M_{\ell_i\times \ell_j}(\mathtt{R}_{\mu-1})$ are to be considered modulo $\gamma^{j}$ for $0\leq i < j\leq \mu-2,$  the matrices $\gamma^i\mathtt{A}_{i,\mu}^{(\mu)} \in M_{k_i\times (N_1-K)}(\mathtt{R}_\mu)$ and $\gamma^{j+1}\mathtt{C}_{j,\mu}^{(\mu)} \in M_{\ell_j\times (N_1-K)}(\mathtt{R}_\mu)$ are to be  considered modulo $\gamma^\mu$ for $0\leq i \leq \mu-1$ and $0\leq j \leq \mu-2,$ and the matrices $\gamma^i\mathtt{B}_{i,\mu-1}^{(\mu)}\in M_{k_i\times (N_2-L)}(\mathtt{R}_{\mu-1})$ and $\gamma^i\mathtt{D}_{i,\mu-1}^{(\mu)}\in M_{\ell_i\times (N_2-L)}(\mathtt{R}_{\mu-1})$ are to be considered modulo $\gamma^{\mu-1}$ for $0\leq i \leq \mu-2,$ (throughout this paper, for positive integers $s$ and $n,$  we denote  the set of all $s \times n$ matrices over $Y$ by $M_{s\times n}(Y)$). The matrix $\mathtt{G}_{\mu}$ of the form \eqref{eqn1.1} is said to be in standard form.  Further, by Proposition 3.2 of Borges \etal~\cite{Borges}, each $\mathtt{R}_\mu\mathtt{R}_{\mu-1}$-linear code of block-length $(N_1,N_2)$ has a unique type $\{k_0,k_1,\ldots,k_{\mu-1};\ell_0,\ell_1,\ldots,\ell_{\mu-2}\},$ \textit{i.e.,} the non-negative integers $k_0,k_1,\ldots,k_{\mu-1},\ell_0,\ell_1,\ldots,\ell_{\mu-2}$ are fixed uniquely for any $\mathtt{R}_\mu\mathtt{R}_{\mu-1}$-linear code of block-length $(N_1,N_2).$

The dual codes of $\mathtt{R}_\mu\mathtt{R}_{\mu-1}$-linear codes of block-length $(N_1,N_2)$ are generally studied \cite{Siap1,Bajalan,Borges} with  respect to a symmetric and non-degenerate bilinear  form $\langle \cdot,\cdot\rangle_E:(\mathtt{R}_\mu^{N_1}\oplus\mathtt{R}_{\mu-1}^{N_2})\times(\mathtt{R}_\mu^{N_1}\oplus\mathtt{R}_{\mu-1}^{N_2})\rightarrow\mathtt{R}_{\mu},$ defined as \vspace{-2mm}\begin{equation}\label{Eucmix}\langle m_1,m_2\rangle_E=\sum\limits_{i=1}^{N_1}\mathtt{c}_i\mathtt{c}'_i+\gamma\sum\limits_{j=1}^{N_2}\mathtt{d}_j\mathtt{d}'_j
\vspace{-2mm}\end{equation} for all $m_1=(\mathtt{c}_1,\mathtt{c}_2,\ldots,\mathtt{c}_{N_1}|\mathtt{d}_{1},\mathtt{d}_2,\ldots,\mathtt{d}_{N_2}),~m_2=(\mathtt{c}'_1,\mathtt{c}'_2,\ldots,\mathtt{c}'_{N_1}|\mathtt{d}'_{1},\mathtt{d}'_2,\ldots,\mathtt{d}'_{N_2})\in\mathtt{R}_\mu^{N_1}\oplus\mathtt{R}_{\mu-1}^{N_2}.$ The bilinear form $\langle \cdot,\cdot\rangle_E$ is called the Euclidean bilinear form on $\mathtt{R}_\mu^{N_1}\oplus\mathtt{R}_{\mu-1}^{N_2}$ \cite{Bajalan}. The Euclidean dual code of an $\mathtt{R}_\mu\mathtt{R}_{\mu-1}$-linear code $\mathscr{C}_{\mu}$ of block-length $(N_1,N_2),$  denoted by $\mathscr{C}_{\mu}^{\perp_E},$ is defined as  \begin{equation}\label{Eucdual}\mathscr{C}_{\mu}^{\perp_E}=\{m_1\in\mathtt{R}_\mu^{N_1}\oplus\mathtt{R}_{\mu-1}^{N_2}:\langle m_1,m_2\rangle_E=0~\text{for all $m_2\in\mathscr{C}_{\mu}$} \}.\end{equation} Note that the Euclidean dual code $\mathscr{C}_{\mu}^{\perp_E}$ is also an  $\mathtt{R}_\mu\mathtt{R}_{\mu-1}$-linear code of block-length $(N_1,N_2).$ By Corollary 3.2 of Bajalan \etal~\cite{Bajalan2}, we note that $|\mathscr{C}_{\mu}||\mathscr{C}_{\mu}^{\perp_E}|=q^{\mu N_1+(\mu-1)N_2}.$ Further, the code $\mathscr{C}_{\mu}$ is said to be (i) Euclidean self-orthogonal if $\mathscr{C}_{\mu} \subseteq \mathscr{C}_{\mu}^{\perp_E},$ (ii) Euclidean self-dual if $\mathscr{C}_{\mu} = \mathscr{C}_{\mu}^{\perp_E}$ and (iii) a linear code with Euclidean complementary-dual (or a Euclidean $\mathtt{R}_\mu\mathtt{R}_{\mu-1}$-LCD code) if  $\mathscr{C}_{\mu} \cap \mathscr{C}_{\mu}^{\perp_E}=\{0\}.$

Now to study Euclidean dual codes of $\mathtt{R}_\mu\mathtt{R}_{\mu-1}$-linear codes of block-length $(N_1,N_2)$, let $\mathfrak{M}$ be the set consisting of matrices whose rows are elements of $\mathtt{R}_\mu^{N_1}\oplus\mathtt{R}_{\mu-1}^{N_2}$.  We observe that any matrix in $\mathfrak{M}$ is of the form $[A_1~| ~B_1],$ where $A_1$ is a matrix over $\mathtt{R}_\mu$ with $N_1$ columns and $B_1$ is a matrix over $\mathtt{R}_{\mu-1}$ with $N_2$ columns, and both the matrices $A_1$ and $B_1$ have the same number of rows. Let us define an operation $\diamond$ on the set $\mathfrak{M}$ as
\vspace{-1mm}$$G\diamond H^T=A_1A_2^T+\gamma B_1B_2^T\vspace{-1mm}$$
for all $G=[A_1~|~B_1]$ and $H=[A_2~|~B_2]$ in $\mathfrak{M},$ where $P^T$ denotes the transpose of a matrix $P.$ Note that the matrix $G\diamond H^T$ is a matrix over $\mathtt{R}_\mu.$ One can easily verify that  an $\mathtt{R}_\mu\mathtt{R}_{\mu-1}$-linear code $\mathscr{C}_{\mu}\subseteq\mathtt{R}_\mu^{N_1}\oplus\mathtt{R}_{\mu-1}^{N_2}$ with a generator matrix $\mathtt{G}_\mu$ is Euclidean self-orthogonal if and only if $\mathtt{G}_\mu\diamond\mathtt{G}_\mu^T=0.$ We next make the following observation. 
\vspace{-1mm}\begin{lemma}\label{Lem6.1}
Let $\mathscr{C}_{\mu}$ be an $\mathtt{R}_\mu\mathtt{R}_{\mu-1}$-linear code of block-length $(N_1,N_2)$ with a generator matrix $\mathtt{G}_\mu$  given by \eqref{eqn1.1}. Let $\gamma^{i-1}R_i^{(\mu)}\in \mathfrak{M}$ denote the $i$-th row block of order $k_{i-1}\times (N_1+N_2) $ of the matrix $
    [\mathtt{A}_{\mu} ~|~ \mathtt{B}_{\mu}]$  for $1\leq i\leq \mu,$ and let $\gamma^{j-1}S_j^{(\mu)}\in \mathfrak{M}$ denote the $j$-th row block of order $\ell_{j-1}\times (N_1+N_2) $ of the matrix $
    [\mathtt{C}_{\mu} ~| ~\mathtt{D}_{\mu}]$ for $1\leq j\leq \mu-1.$ Then the code $\mathscr{C}_{\mu}$ is Euclidean  self-orthogonal if and only if \vspace{-2mm}\begin{eqnarray*}(\gamma^{i-1} R_i^{(\mu)})\diamond \big(\gamma^{j-1} {R_j^{(\mu)}}^T\big) & \equiv & 0~(\textnormal{mod }\gamma^\mu)\text{ ~~for~~ }1\leq i\leq j \leq \mu,\vspace{-2mm}\\(\gamma^{i-1} R_i^{(\mu)})\diamond (\gamma^{j-1}{S_j^{(\mu)}}^T) & \equiv & 0~(\textnormal{mod }\gamma^\mu)\text{ ~~ for ~~}1\leq i \leq \mu\text{ ~and~ }1\leq j\leq \mu-1,\text{ and}\vspace{-2mm}\\ (\gamma^{i-1}S_i^{(\mu)})\diamond ({\gamma^{j-1}S_j^{(\mu)}}^T) & \equiv & 0~(\textnormal{mod }\gamma^\mu)\text{ ~~for ~~ }1\leq i\leq j \leq \mu-1.\vspace{-2mm}\end{eqnarray*} \end{lemma}
\begin{proof}
Its proof is a straightforward exercise.
\end{proof}

Next, for an $\mathtt{R}_\mu\mathtt{R}_{\mu-1}$-linear code $\mathscr{C}_{\mu}$ of block-length $(N_1,N_2)$ with a generator matrix $\mathtt{G}_\mu$  given by \eqref{eqn1.1}, we assume, throughout this paper, that $\mathscr{C}_{\mu}^{(X)}$ is a linear code of length $N_1$ over $\mathtt{R}_\mu$ with a generator matrix $\mathtt{A}_\mu$ and  that $\mathscr{C}_{\mu}^{(Y)}$ is a linear code of length $N_2$ over $\mathtt{R}_{\mu-1}$ with a generator matrix $\mathtt{D}_\mu.$  Further, for an element $\mathtt{c}_1\in\mathtt{R}_\mu$
(or $\mathtt{R}_{\mu-1}$), let us define $\overline{\mathtt{c}}_1\in\overline{\mathtt{R}}_e$ as $\overline{\mathtt{c}}_1=\mathtt{c}_1+\langle \gamma\rangle.$  Also, for any matrix $V$ over $\mathtt{R}_\mu$ (or $\mathtt{R}_{\mu-1}$), let  $\overline{V}$ be a matrix over $\overline{\mathtt{R}}_e$ with the $(i,j)$-th entry  $\overline{v}_{i,j}$ if $v_{i,j}$ is the $(i,j)$-th entry of  $V$ for each $i$ and $j.$ Now, for $1\leq i\leq \mu$ and $1\leq j\leq \mu-1,$  the $i$-th torsion code  of $\mathscr{C}_{\mu}^{(X)}$ and the $j$-th torsion code  of $\mathscr{C}_{\mu}^{(Y)}$  are defined \cite{Dougherty} as 
\vspace{-1mm}\begin{align*}
Tor_i(\mathscr{C}_{\mu}^{(X)})&=\{\overline{\mathtt{c}}=(\overline{\mathtt{c}}_1,\overline{\mathtt{c}}_2,\ldots,\overline{\mathtt{c}}_{N_1})\in\overline{\mathtt{R}}_e^{N_1}:\gamma^{i-1}\mathtt{c}\in\mathscr{C}_{\mu}^{(X)}~\text{for some $\mathtt{c}=(\mathtt{c}_1,\mathtt{c}_2,\ldots,\mathtt{c}_{N_1})\in\mathtt{R}_\mu^{N_1}$}\}~\text{and}\vspace{-1mm}\\
Tor_j(\mathscr{C}_{\mu}^{(Y)})&=\{\overline{\mathtt{d}}=(\overline{\mathtt{d}}_1,\overline{\mathtt{d}}_2,\ldots,\overline{\mathtt{d}}_{N_2})\in\overline{\mathtt{R}}_e^{N_2}:\gamma^{j-1}\mathtt{d}\in\mathscr{C}_{\mu}^{(Y)}~\text{for some $\mathtt{d}=(\mathtt{d}_1,\mathtt{d}_2,\ldots,\mathtt{d}_{N_2})\in\mathtt{R}_{\mu-1}^{N_2}$}\},\vspace{-1mm}
\end{align*}
respectively. For $1\leq i\leq \mu,$ we note, by Lemma 3.4 of Norton and S\u{a}l\u{a}gean~\cite{Norton},  that $Tor_i(\mathscr{C}_{\mu}^{(X)})$ is a linear code of  length $N_1$ and dimension $k_0+k_1+\cdots+k_{i-1}$ over $\overline{\mathtt{R}}_e$  with a generator matrix 
\begin{equation*}
\left[\begin{array}{ccccccc}
 \overline{I_{k_0}} & \overline{\mathtt{A}_{0,1}^{(\mu)}} & \overline{\mathtt{A}_{0,2}^{(\mu)}} & \cdots & \overline{\mathtt{A}_{0,i-1}^{(\mu)}} & \cdots & \overline{\mathtt{A}_{0,\mu}^{(\mu)}}\\
 0 & \overline{I_{k_1}} &  \overline{\mathtt{A}_{1,2}^{(\mu)}} & \cdots & \overline{\mathtt{A}_{1,i-1}^{(\mu)}} & \cdots & \overline{\mathtt{A}_{1,\mu}^{(\mu)}}\\
 \vdots  & \vdots & \vdots & \ddots & \vdots & \ddots & \vdots \\
 0 & 0  & 0 & \cdots & \overline{I_{k_{i-1}}} & \cdots & \overline{\mathtt{A}_{i-1,\mu}^{(\mu)}}
\end{array}
\right],
\end{equation*}
and that for $1\leq j\leq \mu-1,$ $Tor_j(\mathscr{C}_{\mu}^{(Y)})$ is a linear code of length $N_2$  and dimension $\ell_0+\ell_1+\cdots+\ell_{j-1}$ over $\overline{\mathtt{R}}_e$ with a generator matrix  
\vspace{-1mm}\begin{equation*}
\left[\begin{array}{ccccccc}
 \overline{I_{\ell_0}} & \overline{\mathtt{D}_{0,1}^{(\mu)}} & \overline{\mathtt{D}_{0,2}^{(\mu)}} & \cdots & \overline{\mathtt{D}_{0,j-1}^{(\mu)}} & \cdots & \overline{\mathtt{D}_{0,\mu-1}^{(\mu)}}\\
 0 & \overline{I_{\ell_1}} &  \overline{\mathtt{D}_{1,2}^{(\mu)}}  & \cdots & \overline{\mathtt{D}_{1,j-1}^{(\mu)}} & \cdots & \overline{\mathtt{D}_{1,\mu-1}^{(\mu)}}\\
 \vdots  & \vdots & \vdots & \ddots & \vdots & \ddots & \vdots \\
 0 & 0  & 0 & \cdots & \overline{I_{\ell_{j-1}}} & \cdots & \overline{\mathtt{D}_{j-1,\mu-1}^{(\mu)}}
\end{array}
\right].
\end{equation*}
We next make the following observation.
\vspace{-1mm}\begin{lemma}\label{Lem1.1} Let $\mathscr{C}_{\mu}$ be a Euclidean self-orthogonal $\mathtt{R}_\mu \mathtt{R}_{\mu-1}$-linear code of block-length $(N_1,N_2).$ The following hold.
\vspace{-2mm}\begin{itemize}
\item[(a)] $Tor_{i}(\mathscr{C}_{\mu}^{(X)})\subseteq Tor_{i}(\mathscr{C}_{\mu}^{(X)})^{\perp_E}$   and $Tor_{j}(\mathscr{C}_{\mu}^{(Y)})\subseteq Tor_{j}(\mathscr{C}_{\mu}^{(Y)})^{\perp_E}$ for  $1\leq i\leq \floor{\frac{\mu+1}{2}}$ and $1\leq j\leq \floor{\frac{\mu}{2}}.$
\vspace{-2mm}\item[(b)] $Tor_{i}(\mathscr{C}_{\mu}^{(X)})\subseteq Tor_{\mu-i+1}(\mathscr{C}_{\mu}^{(X)})^{\perp_E}$  and $Tor_{j}(\mathscr{C}_{\mu}^{(Y)})\subseteq Tor_{\mu-j}(\mathscr{C}_{\mu}^{(Y)})^{\perp_E}$ for  $\ceil{\frac{\mu+1}{2}}\leq i\leq \mu$ and $\ceil{\frac{\mu}{2}}\leq j\leq \mu-1.$\vspace{-2mm}
\end{itemize}
In particular, if the code $\mathscr{C}_{\mu}$ is Euclidean self-dual, then we have $Tor_{i}(\mathscr{C}_{\mu}^{(X)}) = Tor_{\mu-i+1}(\mathscr{C}_{\mu}^{(X)})^{\perp_E}$ and  $Tor_{j}(\mathscr{C}_{\mu}^{(Y)}) = Tor_{\mu-j}(\mathscr{C}_{\mu}^{(Y)})^{\perp_E}$ for $\ceil{\frac{\mu+1}{2}}\leq i\leq \mu$ and $\ceil{\frac{\mu}{2}}\leq j\leq \mu-1,$ (throughout this paper, the floor and ceiling functions are denoted by  $\floor{\cdot}$ and $\ceil{\cdot},$ 
 respectively). 
  \end{lemma}
\begin{proof} Proof is similar to that of Lemma 7 of Dougherty \etal~\cite{Dougherty}.
\end{proof}

Bajalan \etal~\cite{Bajalan}  observed that if $\mathscr{C}_{\mu}\subseteq\mathtt{R}_\mu^{N_1}\oplus\mathtt{R}_{\mu-1}^{N_2}$ is  an  $\mathtt{R}_\mu\mathtt{R}_{\mu-1}$-linear code of the type $\{k_0,k_1,\ldots,k_{\mu-1};\ell_0,\ell_1,\\\ldots,\ell_{\mu-2}\},$ then its Euclidean dual code $\mathscr{C}_{\mu}^{\perp_E}$ is of the type $\{N_1-\sum\limits_{i=0}^{\mu-1}k_i,k_{\mu-1},k_{\mu-2},\ldots,k_1;N_2-\sum\limits_{j=0}^{\mu-2}\ell_j,\ell_{\mu-2},\ell_{\mu-3},\\\ldots,\ell_1\}.$ In particular, if the code $\mathscr{C}_{\mu}$ is Euclidean self-dual, then by  uniqueness of the types of $\mathscr{C}_{\mu}$ and $\mathscr{C}_{\mu}^{\perp_E},$ we obtain $2k_0+k_1+\cdots+k_{\mu-1}=N_1,$ $2\ell_0+\ell_1+\cdots+\ell_{\mu-2}=N_2,$ $k_i=k_{\mu-i}$ for $1\leq i\leq \mu-1$ and $\ell_j=\ell_{\mu-1-j}$ for $1\leq j\leq \mu-2.$ From this and by Lemma \ref{Lem1.1}, we deduce the following:
\vspace{-1mm}\begin{remark}\label{Remark1.1}
If $\mathscr{C}_{\mu}\subseteq\mathtt{R}_\mu^{N_1}\oplus\mathtt{R}_{\mu-1}^{N_2}$ is a Euclidean self-orthogonal $\mathtt{R}_\mu\mathtt{R}_{\mu-1}$-linear code of the type $\{k_0,k_1,\ldots,k_{\mu-1};\ell_0,\\\ell_1,\ldots,\ell_{\mu-2}\},$ then we have \vspace{-2mm}$$2k_0+2k_1+\cdots+2k_{\mu-1-i}+k_{\mu-1-i+1}+k_{\mu-1-i+2}+\cdots+k_i\leq N_1\text{ ~ for ~ }\lfloor{\tfrac{\mu}{2}\rfloor}\leq i\leq \mu-1\vspace{-2mm}$$ and \vspace{-1mm}$$2\ell_0+2\ell_1+\cdots+2\ell_{\mu-2-j}+\ell_{\mu-2-j+1}+\ell_{\mu-2-j+2}+\cdots+\ell_j\leq N_2\text{ ~ for ~ }\floor{\tfrac{\mu-1}{2}}\leq j\leq \mu-2.\vspace{-1mm}$$ Furthermore, if the code $\mathscr{C}_{\mu}$ is Euclidean self-dual, then we  have \vspace{-2mm}$$2(k_0+k_1+\cdots+k_{\frac{\mu-1}{2}})=N_1\text{  ~ and ~ }2(\ell_0+\ell_1+\cdots+\ell_{\frac{\mu-3}{2}})+\ell_{\frac{\mu-1}{2}}=N_2\text{ ~ if }\mu \text{ ~is odd,}\vspace{-2mm}$$ while \vspace{-1mm}$$2(k_0+k_1+\cdots+k_{\frac{\mu-2}{2}})+k_{\frac{\mu}{2}}=N_1\text{ ~ and ~ } 2(\ell_0+\ell_1+\cdots+\ell_{\frac{\mu-2}{2}})=N_2\text{ ~if }\mu \text{ ~is even.}$$  
\end{remark}
\vspace{-3mm}\subsection{Additive codes over chain rings} \label{prelim4}
\vspace{-1mm}
Recall, from Section \ref{prelim1}, that $\mathcal{R}_e=GR(p^e,r)[y]/\langle g(y),p^{e-1}y^t\rangle$ is a finite commutative chain ring, where $t$ and $k$ are  integers satisfying $1\leq t\leq k$ when $e \geq 2,$ while $t=k$ when $e=1,$ and $g(y)=y^k+p(g_{k-1}y^{k-1}+\cdots+g_1y+g_0)\in GR(p^e,r)[y]$ is an Eisenstein polynomial with $g_0$ as a unit in $GR(p^e,r).$ Now for a positive integer $N,$ let $\mathcal{R}_e^N$ denote the set consisting of all $N$-tuples over $\mathcal{R}_e.$ Clearly, $\mathcal{R}_e^N$ is an Abelian group under the component-wise addition, or equivalently, $\mathcal{R}_e^N$ can be viewed as a $\mathbb{Z}_{p^e}$-module. A non-empty subset $\mathcal{C}$ of $\mathcal{R}_e^N$ is called an additive code of length $N$ over $\mathcal{R}_e$ if it is an additive subgroup of $\mathcal{R}_e^N,$ or equivalently, $\mathcal{C}$ is a $\mathbb{Z}_{p^e}$-submodule of $\mathcal{R}_e^N$.

In this paper, we will  study the dual codes of additive codes of length $N$ over $\mathcal{R}_e$ with respect to the notion of   
character-theoretic duality considered in \cite{Edgar} and \cite{Wood}.  For this, we first note that the additive group of $\mathcal{R}_e$ is isomorphic to its character group $\widehat{\mathcal{R}}_e,$ and hence $\widehat{\mathcal{R}}_e=\{\chi_a: a \in \mathcal{R}_e\}.$ For $\textbf{c}=(c_0,c_1,\ldots,c_{N-1}),\textbf{d}=(d_0,d_1,\ldots,d_{N-1})\in \mathcal{R}_e^N,$ let us define \vspace{-2mm}$$ \chi_\textbf{d}(\textbf{c})=\chi_{d_0}(c_0)\chi_{c_1}(d_1)\ldots \chi_{c_{N-1}}(d_{N-1}).\vspace{-1mm}$$ The $\chi$-dual code of an additive code $\mathcal{C}$ of length $N$ over $\mathcal{R}_e,$ denoted by $\mathcal{C}^{\perp_{\chi}},$ is defined as \vspace{-2mm}\begin{align}\label{eqtn3.2}
    \mathcal{C}^{\perp_{\chi}}=\{\textbf{d}\in \mathcal{R}_e^N:\chi_\textbf{d}(\textbf{c})=1\text{ for all }\textbf{c}\in \mathcal{C}\}.
\end{align}
It is easy to verify that $\mathcal{C}^{\perp_{\chi}}$ is also an additive code of length $N$ over $\mathcal{R}_e.$ By Corollaries 3.6 and 3.7 of Dougherty \etal~\cite{Dougherty2}, we note  that $(\mathcal{C}^{\perp_{\chi}})^{\perp_{\chi}}=\mathcal{C}$ and $|\mathcal{C}||\mathcal{C}^{\perp_{\chi}}|=|\mathcal{R}_e^N|.$ Further, the code $\mathcal{C}$ is said to be (i) self-orthogonal if $\mathcal{C} \subseteq \mathcal{C}^{\perp_\chi},$ (ii) self-dual if $\mathcal{C} = \mathcal{C}^{\perp_\chi}$ and (ii) an additive code with complementary-dual (or an ACD code) if  $\mathcal{C} \cap \mathcal{C}^{\perp_\chi}=\{0\}.$ One may refer to Dougherty \etal~\cite{Dougherty2} for more details on the character-theoretic duality of additive codes. 

Throughout this paper, we will follow the same notations as introduced  in Section \ref{prelim}.
\vspace{-4mm}\section{Additive codes over $\mathcal{R}_e$ and $\mathbb{Z}_{p^e}\mathbb{Z}_{p^{e-1}}$-linear codes}\label{section3}
\vspace{-2mm}
 In this section, we will establish a duality-preserving $\mathbb{Z}_{p^e}$-module isomorphism between additive codes over $\mathcal{R}_e$ and $\mathbb{Z}_{p^e}\mathbb{Z}_{p^{e-1}}$-linear codes induced by the $\mathbb{Z}_{p^e}$-module isomorphism  $\Psi$ defined in Section \ref{prelim1}.  For this, we first recall, from Section \ref{prelim1},  that the additive group of $\mathcal{R}_e$ is isomorphic to the direct sum $\mathbb{Z}_{p^e}^{rt}\oplus\mathbb{Z}_{p^{e-1}}^{r(k-t)}.$ We further observe, from Section \ref{prelim2},  that the character group $\widehat{\mathcal{R}}_e$ of $\mathcal{R}_e$ is isomorphic to the direct sum $\widehat{\mathbb{Z}}_{p^e}^{rt}\oplus\widehat{\mathbb{Z}}_{p^{e-1}}^{r(k-t)}.$ We also recall that each element $a \in \mathcal{R}_e$ can be uniquely expressed as $a=a_{0}+a_{1}y+\cdots+a_{t-1}y^{t-1}+a_{t}y^t+\cdots+ a_{k-1}y^{k-1} ,$ where $a_{i}=\sum\limits_{s=0}^{r-1}a_{i,s}x^s\in GR(p^e,r)$  with $a_{i,0}, a_{i,1},\ldots, a_{i,r-1} \in \mathbb{Z}_{p^e}$ for $0\leq i\leq t-1$ and $a_{j}=\sum\limits_{s=0}^{r-1}a_{j,s}x^s\in GR(p^{e-1},r) $ with $a_{j,0}, a_{j,1},\ldots,a_{j,r-1} \in \mathbb{Z}_{p^{e-1}}$ for $t\leq j\leq k-1.$ We also recall that \vspace{-1mm}\begin{equation*}\Psi(a)=(a_{0,0},a_{0,1},\ldots,a_{0,r-1},\ldots,a_{t-1,0},a_{t-1,1},\ldots,a_{t-1,r-1}~|~a_{t,0},a_{t,1},\ldots,a_{t,r-1},\ldots,a_{k-1,0},a_{k-1,1},\ldots,\\a_{k-1,r-1}),\vspace{-1mm}\end{equation*} as $\phi(a_i)=(a_{i,0},a_{i,1},\ldots,a_{i,r-1})$ for $0\leq i\leq k-1.$  
Corresponding to the element $a \in \mathcal{R}_e,$ the character $\chi_a \in \widehat{\mathcal{R}}_e$  is defined as \vspace{-4mm}\begin{eqnarray}\label{LR}
\chi_{a}(b)&:=& \chi_{\Psi(a)}(\Psi(b))=\prod\limits_{s=0}^{r-1}\Big(\prod\limits_{i=0}^{t-1}\chi_{a_{i,s}}(b_{i,s})\prod\limits_{j=t}^{k-1}\chi_{a_{j,s}}(b_{j,s})\Big)\nonumber \vspace{-2mm}\\
&= &\prod\limits_{s=0}^{r-1}\Big(\prod\limits_{i=0}^{t-1}\xi^{a_{i,s}b_{i,s}}\prod\limits_{j=t}^{k-1}\xi^{pa_{j,s}b_{j,s}}\Big)= \xi^{\sum\limits_{i=0}^{t-1}\sum\limits_{s=0}^{r-1} a_{i,s}b_{i,s}+p\sum\limits_{j=t}^{k-1}\sum\limits_{s=0}^{r-1} a_{j,s}b_{j,s}}
\vspace{-2mm}\end{eqnarray}\vspace{-1mm}
for all $b=b_{0}+b_{1}y+\cdots+b_{t-1}y^{t-1}+b_{t}y^t+\cdots+b_{k-1}y^{k-1}\in \mathcal{R}_e,$ where $b_{i}=\sum\limits_{s=0}^{r-1}b_{i,s}x^s\in\ GR(p^e,r)$ with $b_{i,0}, b_{i,1}, \ldots, b_{i,r-1} \in \mathbb{Z}_{p^e}$ for $0\leq i\leq t-1$ and $b_{j}=\sum\limits_{s=0}^{r-1}b_{j,s}x^s \in GR(p^{e-1},r)$ with $b_{j,0}, b_{j,1}, \ldots, b_{j,r-1} \in \mathbb{Z}_{p^{e-1}}$ for $t\leq j\leq k-1.$ 
The character group of the additive group of  $\mathcal{R}_e$ is given by 
$\widehat{\mathcal{R}}_e=\{\chi_{a}:a\in \mathcal{R}_e\},$ where $\chi_a$ is as defined by \eqref{LR}.
We next recall that the set $\mathcal{R}_e^N$ is an Abelian group under the component-wise addition. For $\textbf{u}=(u_0,u_1,\ldots,u_{N-1}) \in \mathcal{R}_e^N,$ the character $\chi_{\textbf{u}} \in \widehat{\mathcal{R}_e^N} $ is defined as \vspace{-1mm}$$\chi_{\textbf{u}}(\textbf{v})=\chi_{u_0}(v_0)\chi_{u_1}(v_1)\ldots\chi_{u_{N-1}}(v_{N-1})\vspace{-1mm}$$ for all $\textbf{v}=(v_0,v_1,\ldots,v_{N-1})\in \mathcal{R}_e^N,$ where $\chi_{u_i}(v_i)$ is as defined by \eqref{LR} for each $i.$ Further, the character group $\widehat{\mathcal{R}_e^N}$ of the additive group $\mathcal{R}_e^N$ is given by
\vspace{-3mm}\begin{align}\label{eqtn3.1}
\widehat{\mathcal{R}_e^N}=\{\chi_{\textbf{u}}:\textbf{u}\in \mathcal{R}_e^N\}.
\end{align}\vspace{-4mm}

Next, for an  additive code $\mathcal{C} (\subseteq \mathcal{R}_e^N)$  of length $N$ over $\mathcal{R}_e,$ we recall that the $\chi$-dual code of $\mathcal{C}$, denoted by $\mathcal{C}^{\perp_{\chi}},$ is defined as \vspace{-2mm}\begin{align*}\label{eqtn3.2}
    \mathcal{C}^{\perp_{\chi}}=\{\textbf{d}\in \mathcal{R}_e^N:\chi_\textbf{d}(\textbf{c})=1\text{ for all }\textbf{c}\in \mathcal{C}\}.\vspace{-2mm}
\end{align*}
\vspace{-5mm}\begin{lemma}\label{l0.1}Let $\textbf{d}=(d_0,d_1,\ldots,d_{N-1})$ and $\textbf{c}=(c_0,c_1,\ldots,c_{N-1})$ in $\mathcal{R}_e^N$ be fixed. For $0\leq \ell\leq N-1,$ let us assume that  \begin{eqnarray*}\Psi(d_\ell)\hspace{-3mm}&=&\hspace{-3mm}(d_{0,0}^{(\ell)},d_{0,1}^{(\ell)},\ldots,d_{0,r-1}^{(\ell)},\ldots,d_{t-1,0}^{(\ell)},d_{t-1,1}^{(\ell)},\ldots,d_{t-1,r-1}^{(\ell)}~|~d_{t,0}^{(\ell)},d_{t,1}^{(\ell)},\ldots,d_{t,r-1}^{(\ell)},\ldots,d_{k-1,0}^{(\ell)},d_{k-1,1}^{(\ell)},\ldots,d_{k-1,r-1}^{(\ell)})\\\text{and }&& \\\Psi(c_\ell)\hspace{-3mm}&=&\hspace{-3mm}(c_{0,0}^{(\ell)},c_{0,1}^{(\ell)},\ldots,c_{0,r-1}^{(\ell)},\ldots,c_{t-1,0}^{(\ell)},c_{t-1,1}^{(\ell)},\ldots,c_{t-1,r-1}^{(\ell)}~|~c_{t,0}^{(\ell)},c_{t,1}^{(\ell)},\ldots,c_{t,r-1}^{(\ell)},\ldots,c_{k-1,0}^{(\ell)},c_{k-1,1}^{(\ell)},\ldots,c_{k-1,r-1}^{(\ell)}),\end{eqnarray*} where both $d_{i,s}^{(\ell)}, c_{i,s}^{(\ell)}\in \mathbb{Z}_{p^e}$  and both $d_{j,s}^{(\ell)}, c_{j,s}^{(\ell)}\in \mathbb{Z}_{p^{e-1}}$ for  $0 \leq i \leq t-1,$ $t \leq j \leq k-1$ and $0 \leq s \leq r-1.$   We have $\chi_\textbf{d}(\textbf{c})=1$  if and only if \begin{equation}\label{dualityeqn1.4}
\sum\limits_{\ell=0}^{N-1}\sum\limits_{s=0}^{r-1} \bigg(\sum\limits_{i=0}^{t-1}d_{i,s}^{(\ell)}c_{i,s}^{(\ell)}+p\sum\limits_{j=t}^{k-1}d_{j,s}^{(\ell)}c_{j,s}^{(\ell)} \bigg) = 0 \text{ ~in~ } \mathbb{Z}_{p^e}.
\end{equation}
\end{lemma}
We next observe that the $\mathbb{Z}_{p^e}$-module isomorphism $\Psi$ from $\mathcal{R}_e$ onto $\mathbb{Z}_{p^e}^{rt}\oplus\mathbb{Z}_{p^{e-1}}^{r(k-t)}$ can be extended to a $\mathbb{Z}_{p^e}$-module isomorphism from $\mathcal{R}_e^N$ onto $\mathbb{Z}_{p^e}^{Nrt}\oplus\mathbb{Z}_{p^{e-1}}^{Nr(k-t)}$ component-wise, which we will denote by $\Psi$ itself for our convenience.  From the above discussion, we deduce the following:
\vspace{-1mm}\begin{thm}\label{thm0.1}
\begin{itemize}\item[(a)]A non-empty subset $\mathcal{C}$ of $\mathcal{R}_e^N$ is an additive code if and only if $\Psi(\mathcal{C})$ is a $\mathbb{Z}_{p^e}\mathbb{Z}_{p^{e-1}}$-linear code of block-length $\bigl(Nrt,Nr(k-t)\bigr).$
\vspace{-1mm}\item[(b)]  For an additive code $\mathcal{C}$ of length $N$ over $\mathcal{R}_e,$ we have $\Psi(\mathcal{C}^{\perp_\chi})=\Psi(\mathcal{C})^{\perp_E}.$   
\end{itemize}
\end{thm}
\vspace{-2mm}\begin{proof}
Part (a) follows from the fact that the map $\Psi$ is a $\mathbb{Z}_{p^e}$-module isomorphism from $\mathcal{R}_e^N$ onto $\mathbb{Z}_{p^e}^{Nrt}\oplus\mathbb{Z}_{p^{e-1}}^{Nr(k-t)},$ while part (b) follows from part (a),    Lemma \ref{l0.1} and equations \eqref{Eucmix} and \eqref{Eucdual}.\end{proof}

From the above theorem, we deduce the following:
\vspace{-1mm}\begin{remark}\label{KEY} An additive code $\mathcal{C}$ of length $N$ over $\mathcal{R}_e$ is (i) self-orthogonal if and only if $\Psi(\mathcal{C})$ is a Euclidean self-orthogonal $\mathbb{Z}_{p^e}\mathbb{Z}_{p^{e-1}}$-linear code of block-length $\bigl(Nrt,Nr(k-t)\bigr)$, (ii) self-dual if and only if $\Psi(\mathcal{C})$ is a Euclidean self-dual $\mathbb{Z}_{p^e}\mathbb{Z}_{p^{e-1}}$-linear code of block-length $\bigl(Nrt,Nr(k-t)\bigr)$, and (iii) ACD if and only if $\Psi(\mathcal{C})$ is a Euclidean $\mathbb{Z}_{p^e}\mathbb{Z}_{p^{e-1}}$-LCD code of block-length $\bigl(Nrt,Nr(k-t)\bigr).$   Thus  the study of additive codes of length $N$ over $\mathcal{R}_e$ and their $\chi$-dual codes  is equivalent to the study of $\mathbb{Z}_{p^e}\mathbb{Z}_{p^{e-1}}$-linear codes of block-length $\bigl(Nrt,Nr(k-t)\bigr)$ and their Euclidean dual codes. Further, to study and enumerate self-orthogonal and self-dual additive codes and ACD codes of length $N$ over $\mathcal{R}_e,$ it is enough to study and enumerate  Euclidean self-orthogonal and self-dual $\mathbb{Z}_{p^e}\mathbb{Z}_{p^{e-1}}$-linear codes and Euclidean  $\mathbb{Z}_{p^e}\mathbb{Z}_{p^{e-1}}$-LCD codes of block-length $\bigl(Nrt,Nr(k-t)\bigr),$ respectively. We will make use of this observation to study self-orthogonal and self-dual additive codes and ACD codes of length $N$ over $\mathcal{R}_e.$ \end{remark}

In particular, let us consider the chain ring $\mathscr{R}_2=\mathbb{Z}_4[y]/\langle y^2-2,2y\rangle.$ Note that each element $\alpha \in \mathscr{R}_2$ can be uniquely expressed as $\alpha=\alpha_0+\alpha_1y,$ where $\alpha_0 \in \mathbb{Z}_{4}$ and $\alpha_1 \in \mathbb{Z}_2.$ Here, we see, by Greferath and Schmidt \cite[p. 1]{Greferath1999}, that the homogeneous weight on $\mathscr{R}_2$  is defined as \vspace{-2mm}$$w_{hom}(\alpha)=w_{hom}(\alpha_0+\alpha_1y)=\left\{\begin{array}{cl} 0 & \text{if } \alpha=0;\\
4 & \text{if } \alpha_0=2~\text{ and }~\alpha_1=0;\\
2 & \text{otherwise,}\end{array}\right.\vspace{-2mm}$$
where $\alpha=\alpha_0+\alpha_1y~$ with $~\alpha_0 \in \mathbb{Z}_{4}$ and $\alpha_1 \in \mathbb{Z}_2.$
Further, the homogeneous weight of a word  $\textbf{a}=(a_0,a_1,\ldots,a_{N-1})\in \mathscr{R}_2^N$, denoted by $\textbf{w}_{hom}(\textbf{a}),$ is defined as $\textbf{w}_{hom}(\textbf{a})=\sum\limits_{i=0}^{N-1}w_{hom}(a_i).$ The homogeneous distance between the words $\textbf{a}, \textbf{b} \in \mathscr{R}_2^N$ is defined as $d_{hom}(\textbf{a}, \textbf{b})=\textbf{w}_{hom}(\textbf{a}- \textbf{b}).$ The homogeneous distance of an additive code $\mathcal{C}$ of length $N$ over $\mathscr{R}_2$ is defined as $d_{hom}(\mathcal{C})=\min\{ d_{hom}(\textbf{a}, \textbf{b}): \textbf{a}, \textbf{b} \in \mathcal{C} \text{ and }\textbf{a}\neq \textbf{b}\}= \min\{\textbf{w}_{hom}(\textbf{c}): \textbf{c}(\neq \textbf{0}) \in \mathcal{C}\}.$
In Table \ref{table1}, we apply Theorem \ref{thm0.1} to provide minimal generating sets of some additive codes of length $N,$  size $M$ and homogenous distance $d_{hom}$ over $\mathscr{R}_2$ achieving the Plotkin bound for homogeneous weights (see Greferath and O\textquotesingle Sullivan \cite[Th. 2.2]{Greferath}). This suggests that additive codes over $\mathcal{R}_e$ is a promising class of error-correcting codes to find optimal codes with respect to the homogeneous metric

\vspace{-1mm}\begin{table}[h]
    \centering
\resizebox{\textwidth}{!}{\begin{tabular}{ |l|c|c|c|c||l|c|c|c|c|}
 \hline
  Minimal generating set & $N$ & $M$ & $d_{hom}$ & Remarks & Minimal generating set & $N$ & $M$ & $d_{hom}$ & Remarks\\ [1ex]
 \hline
   $\{(2,2,0),(y,y,2)\}$ &$3$& $4$ & $8$ & Optimal & $\{(2,2,0),(y,2+y,2)\}$ &$3$& $4$ & $8$ & Optimal  \\
   & & & & self-orthogonal & & & & & self-orthogonal \\\hline
  $\{(2,y,y),(y,2+y,2)\}$ &$3$& $4$ & $8$ & Optimal & $\{(2,y,2+y),(y,2+y,2)\}$ &$3$& $4$ & $8$ & Optimal  \\
   & & & & LCD & & & & & LCD \\\hline
$\{(2+y,y,2,2),(y,2,2+y,y)\}$ &$4$& $4$ & $10$ & Optimal &  $\{(2+y,2,0,2),(y,y,2,2)\}$ &$4$& $4$ & $10$ & Optimal  \\
   & & & & LCD & & & & & LCD  \\\hline
 $\{(2+y,y,2,y),(y,2+y,2,2)\}$ &$4$& $4$ & $10$ & Optimal & $\{(0,2,2+y,2),(2,0,y,2)\}$ &$4$& $4$ & $10$ & Optimal  \\
   & & & & &  & & & & \\\hline 
$\{(2,y,y,y),(0,2,2,2+y)\}$ &$4$& $4$ & $10$ & Optimal & $\{(2,0,2,2),(y,2,y,y)\}$ &$4$& $4$ & $10$ & Optimal
\\
   & & & & &  & & & & \\\hline
  \end{tabular}}
 \caption{Some additive codes over $\mathscr{R}_2$ achieving the Plotkin bound for homogeneous weights}
    \label{table1}
\end{table}\normalsize
Next, we recall, from Section \ref{prelim3}, that $\mathtt{R}_e$ is a finite commutative chain ring with the maximal ideal $\langle\gamma\rangle$  of nilpotency index $e$ and that $\mathtt{R}_{e-1}$ is the chain ring $\mathtt{R}_e/\langle\gamma^{e-1}\rangle$ with the maximal ideal of nilpotency index $e-1.$ In the following section, we will provide a recursive method to construct  all Euclidean self-orthogonal and self-dual $\mathtt{R}_e\mathtt{R}_{e-1}$-linear codes of block-length $(N_1,N_2)$ by assuming that the characteristic of $\mathtt{R}_e$ is odd, which  gives rise to enumeration formulae for these two classes of codes.   These results also give rise to construction methods and  enumeration formulae for all self-orthogonal and self-dual additive codes of length $N$ over $\mathcal{R}_e$ as special cases (see Corollaries \ref{Corollary3.1}-\ref{Cor3.2}). 

\vspace{-4mm}\section{Enumeration of Euclidean self-orthogonal and self-dual $\mathtt{R}_e\mathtt{R}_{e-1}$-linear codes of block-length $(N_1,N_2)$}\label{SectionEnumeration}
\vspace{-2mm}
Throughout this section, we assume that  the chain ring $\mathtt{R}_e$ is of odd characteristic, or equivalently, the residue field $\overline{\mathtt{R}}_e$ is of odd prime power order $q.$ Let   $Sym_{s}(Y)$ denote the set of all $s \times s$ symmetric matrices over $Y.$  We also need the following well-known lemma. 
\vspace{-1mm}\begin{lemma}\cite{Betty,Yadav}\label{Matrixlemma}
For a matrix $A\in M_{s\times n}(\overline{\mathtt{R}}_e)$ of rank $s,$ the following hold.
\vspace{-1mm}\begin{enumerate}
\item[(a)] The map $\Phi_A:M_{s\times n}(\overline{\mathtt{R}}_e)\rightarrow Sym_s(\overline{\mathtt{R}}_e),$ defined as $\Phi_A(B)=AB^T+BA^T$ for all $B\in M_{s\times n}(\overline{\mathtt{R}}_e),$ is a surjective  $\overline{\mathtt{R}}_e$-linear transformation with $|Ker(\Phi_A)|=q^{\frac{s}{2}(2n-s-1)}.$ 
\vspace{-1mm}\item[(b)] For a positive integer $\ell$ and  a matrix $J\in M_{s\times \ell}(\overline{\mathtt{R}}_e),$ the matrix equation $AX=J$ in the unknown matrix $X\in M_{n\times\ell}(\overline{\mathtt{R}}_e)$ has precisely $q^{(n-s)\ell}$ distinct solutions.
\end{enumerate}
\end{lemma}
\begin{proof}
Part (a) follows from Lemma 3.1 of Betty \textit{et al.} \cite{Betty}, while  part (b) is a straightforward exercise.
\end{proof}
Next, for $0\leq s\leq n,$  let $\sigma_{q}(n,s)$  denote the number of distinct Euclidean self-orthogonal codes of length $n$ and dimension $s$  over the residue field $\overline{\mathtt{R}}_e(\simeq \mathbb{F}_q).$ Clearly, $\sigma_{q}(n,0)=1$ and $\sigma_{q}(n,s)=0$ for all  $s > \floor{\frac{n}{2}}.$ For $1 \leq s \leq \floor{\frac{n}{2}}, $ we see, by Theorem 2 of Pless \cite{Pless}, that
\vspace{-2mm}\begin{equation}\label{sigmaE}
\sigma_{q}(n,s)= \left\{\begin{array}{ll}\frac{\prod\limits_{i=0}^{s-1}(q^{(n-1-2i)}-1)}{\prod\limits_{j=1}^{s}(q^j-1)}& \text{if $n$ is odd};\\
\frac{(q^{n-s}-\nu q^{\frac{n}{2}-s}+\nu q^{\frac{n}{2}}-1)\prod\limits_{i=1}^{s-1}(q^{n-2i}-1)}{\prod\limits_{j=1}^s(q^j-1)}& \text{if $n$ is even,}\end{array} \right.
\end{equation}
where \vspace{-2mm}$$\nu=\left\{\begin{array}{cl}
1&\text{if $(-1)^\frac{n}{2}$ is a square in $\mathbb{F}_q$};\\
-1&\text{otherwise.}\\
\end{array}\right.$$
\vspace{-5mm}\begin{remark}\label{Remark2.2}
When $q$ is odd, one can easily see, from equation \eqref{sigmaE},  that there exists a Euclidean self-dual code of length $n$ over $\mathbb{F}_q$  if and only if $n$ is even and $(-1)^\frac{n}{2}$ is a square in $\mathbb{F}_q$.
\end{remark}
For positive integers $s$ and $n$ satisfying $s \leq n,$ the Gaussian binomial coefficient (or the $q$-binomial coefficient) is defined as  $$\qbin{n}{s}{q}=\frac{(q^n-1)(q^n-q)\ldots (q^n-q^{s-1})}{(q^s-1)(q^s-q)\cdots (q^s-q^{s-1})}.$$ Note that $\qbin{n}{0}{q}$ is assigned the value $1.$ Clearly, $\qbin{n}{s}{q}$ equals the number of distinct $s$-dimensional  subspaces of an $n$-dimensional vector space   over $\mathbb{F}_q.$

We now proceed to count  Euclidean self-orthogonal and self-dual $\mathtt{R}_e\mathtt{R}_{e-1}$-linear codes of block-length $(N_1,N_2).$ For this, we will first consider the cases $e=2$ and $e=3.$  Further, for an integer $e \geq 4,$ we will extend the recursive method provided in Section 5 of Yadav and Sharma \cite{Yadav} to construct and enumerate Euclidean self-orthogonal $\mathtt{R}_e\mathtt{R}_{e-1}$-linear codes of block-length $(N_1,N_2).$ Towards this, let us define a map $\theta:\mathtt{R}_e \rightarrow \mathcal{T}_e$ as $\theta(\mathtt{s})=\mathtt{s}_0$ for all $\mathtt{s}=\mathtt{s}_0+\mathtt{s}_1 \gamma +\mathtt{s}_2 \gamma^2+\cdots+\mathtt{s}_{e-1}\gamma^{e-1},$ where $\mathtt{s}_0,\mathtt{s}_1,\ldots, \mathtt{s}_{e-1}\in \mathcal{T}_e.$ The map $\theta$ induces a binary operation $\oplus$ on $\mathcal{T}_e,$ defined as $a \oplus b=\theta(a+b)$ for all $a, b \in \mathcal{T}_e.$ One can easily see that the set $\mathcal{T}_e$ is the finite field of order $q$ under the addition operation $\oplus$ and the usual multiplication operation of $\mathtt{R}_e.$ So from this point on, we assume, without any loss of generality, that $\overline{\mathtt{R}}_e=\mathcal{T}_e.$   Further,  let $\mathfrak{N}_{e}(N_1,N_2;\textbf{k}; \boldsymbol{\ell})$ denote the number of distinct Euclidean self-orthogonal $\mathtt{R}_e\mathtt{R}_{e-1}$-linear codes of block-length $(N_1,N_2)$ and type $(\textbf{k}; \boldsymbol{\ell}),$ where $\textbf{k}$ is an $e$-tuple $(k_0,k_1,\ldots,k_{e-1})$ and $\boldsymbol{\ell}$ is an $(e-1)$-tuple $(\ell_0,\ell_1,\ldots,\ell_{e-2})$ of non-negative integers.  Furthermore, let $\mathfrak{N}_e(N_1,N_2)$ and $\mathfrak{M}_e(N_1,N_2)$ denote the number of  distinct Euclidean self-orthogonal and self-dual $\mathtt{R}_e\mathtt{R}_{e-1}$-linear codes of block-length $(N_1,N_2),$ respectively.

In the following section, we will consider the case $e=2$ and determine the numbers $\mathfrak{N}_{2}(N_1,N_2)$ and $\mathfrak{M}_{2}(N_1,N_2).$
\vspace{-3mm}\subsection{The case $e=2$ }
\vspace{-1mm}
Throughout this section, let us suppose that $e=2.$ Let $\mathscr{C}_2$ be an $\mathtt{R}_2 \mathtt{R}_1$-linear code  of block-length $(N_1,N_2)$ and type $\{k_0,k_1;\ell_0\}.$ Since $\mathcal{T}_2=\overline{\mathtt{R}}_2$ is the finite field of order $q,$  we note, by Section \ref{prelim3}, that the torsion codes $Tor_1(\mathscr{C}_2^{(X)})$ and $Tor_2(\mathscr{C}_2^{(X)})$ are linear codes of length $N_1$ and dimensions $k_0$ and $k_0+k_1$ over $\mathcal{T}_2,$ respectively, while the torsion code  $Tor_1(\mathscr{C}_2^{(Y)})$ is a  linear code of length $N_2$  and dimension $\ell_0$ over $\mathcal{T}_2.$ Further, if the code $\mathscr{C}_2$ is Euclidean self-orthogonal, then by Lemma \ref{Lem1.1}, we see that \vspace{-2mm}$$Tor_1(\mathscr{C}_2^{(X)})\subseteq Tor_2(\mathscr{C}_2^{(X)})\subseteq{Tor_1(\mathscr{C}_2^{(X)})}^{\perp_E} \text{~~and~~} Tor_1(\mathscr{C}_2^{(Y)})\subseteq{Tor_1(\mathscr{C}_2^{(Y)})}^{\perp_E}.\vspace{-2mm}$$ Furthermore, by Remark \ref{Remark1.1}, we must have $2k_0+k_1\leq N_1$ and $2\ell_0\leq N_2.$ 

Here, we will first count all Euclidean self-orthogonal $\mathtt{R}_2\mathtt{R}_1$-linear  codes of block-length $(N_1,N_2)$ and  type $\{k_0,k_1;\ell_0\}$ with prescribed torsion codes. To do this, let $C_1$ and $C_2$  be   linear codes of length $N_1$ over $\mathcal{T}_2$ with dimensions  $k_0$ and $k_0+k_1,$ and  generator matrices \vspace{-2mm}$$[I_{k_0}~A_{0,1}~A_{0,2}] \text{ ~ and ~}
\left[\begin{array}{ccc}
    I_{k_0} & A_{0,1} & A_{0,2}\\
    0 & I_{k_1} & A_{1,2}
\end{array}\right],\vspace{-2mm}$$ respectively, and let $C_3$ be a linear code of length $N_2$ and dimension $\ell_0$ over $\mathcal{T}_2$ with a generator matrix $[I_{\ell_0} ~ D_{0,1}],$
where $A_{0,1}\in M_{k_0\times k_1}(\mathcal{T}_2),$  $A_{0,2}\in M_{k_0\times (N_1-k_0-k_1)}(\mathcal{T}_2),$ $A_{1,2}\in M_{k_1\times (N_1-k_0-k_1)}(\mathcal{T}_2),$ and $D_{0,1}\in M_{\ell_0\times (N_2-\ell_0)}(\mathcal{T}_2).$ Clearly, $C_1\subseteq C_2.$ Now, we have the following lemma.
\vspace{-1mm}\begin{lemma}\label{Lem4.1}
If $\mathscr{C}_2$ is an $\mathtt{R}_2\mathtt{R}_1$-linear code of block-length $(N_1,N_2)$ with $Tor_1(\mathscr{C}_2^{(X)})=C_1,$ $Tor_2(\mathscr{C}_2^{(X)})=C_2$ and $Tor_1(\mathscr{C}_2^{(Y)})=C_3,$ then there exist matrices $A'_{0,2}\in M_{k_0\times (N_1-k_0-k_1)}(\mathcal{T}_2),$  $B_{0,1}\in M_{k_0\times (N_2-\ell_0)}(\mathcal{T}_2)$ and  $C_{0,2}\in M_{\ell_0\times (N_1-k_0-k_1)}(\mathcal{T}_2),$ such that the matrix  \vspace{-1mm}\begin{align}\label{eqnLem1.1}
\left[\begin{array}{ccc|cc}
    I_{k_0} & A_{0,1} & A_{0,2}+\gamma A_{0,2}' & 0 & B_{0,1} \\
    0 & \gamma I_{k_1} & \gamma A_{1,2} & 0 & 0 \\
    0 & 0 & \gamma C_{0,2} & I_{\ell_0} & D_{0,1} \\
\end{array}\right]=\left[\begin{array}{c}
    R_1^{(2)}\\
    \gamma R_2^{(2)}\\
    S_1^{(2)}
\end{array}\right]
\vspace{-2mm}\end{align} 
is a generator matrix of the code $\mathscr{C}_2.$
\end{lemma}
\begin{proof}
Working as in Lemma 3.1 of Yadav and Sharma \cite{Yadav} and choosing the matrices $B_{0,1}$ and $C_{0,2}$ arbitrarily, the desired result follows. 
\end{proof}

For the rest of this section, we assume that the codes $C_1,$ $C_2$ and $C_3$ satisfy $C_1\subseteq C_1^{\perp_E},$ $C_2\subseteq C_1^{\perp_E}$ and $C_3\subseteq C_3^{\perp_E}.$ It is easy to see that  $2k_0+k_1\leq N_1$ and $2\ell_0\leq N_2.$ Further, working as in Remark 2.2 of Yadav and Sharma \cite{Yadav}, we see that the matrices $A_{0,2}$ and $D_{0,1}$ are of full row-ranks. Now, in the following proposition, we  count all Euclidean self-orthogonal $\mathtt{R}_2\mathtt{R}_1$-linear codes $\mathscr{C}_2$ of block-length $(N_1,N_2)$ satisfying $Tor_1(\mathscr{C}_2^{(X)})=C_1,$ $Tor_2(\mathscr{C}_2^{(X)})=C_2$ and $Tor_1(\mathscr{C}_3^{(Y)})=C_3.$
\vspace{-1mm}\begin{prop}\label{Thm4.1}
There are precisely \vspace{-2mm}$$q^{k_0(N_1+N_2-\frac{3}{2}k_0-k_1-\ell_0-\frac{1}{2})+\ell_0(N_1-2k_0-k_1)}\vspace{-2mm}$$ distinct Euclidean self-orthogonal $\mathtt{R}_2\mathtt{R}_1$-linear codes $\mathscr{C}_2$ of block-length $(N_1,N_2)$ satisfying $Tor_1(\mathscr{C}_2^{(X)})=C_1,$ $Tor_2(\mathscr{C}_2^{(X)})=C_2$ and $Tor_1(\mathscr{C}_2^{(Y)})=C_3.$ 
\end{prop}
\begin{proof}
Let $\mathscr{C}_2$ be a Euclidean self-orthogonal $\mathtt{R}_2\mathtt{R}_1$-linear code of block-length $(N_1,N_2)$ satisfying $Tor_1(\mathscr{C}_2^{(X)})=C_1,$ $Tor_2(\mathscr{C}_2^{(X)})=C_2$ and $Tor_1(\mathscr{C}_2^{(Y)})=C_3.$ By Lemma \ref{Lem4.1},  we see that there exist matrices $B_{0,1}\in M_{k_0\times (N_2-\ell_0)}(\mathcal{T}_2),$ $A'_{0,2}\in M_{k_0\times (N_1-k_0-k_1)}(\mathcal{T}_2)$ and  $C_{0,2}\in M_{\ell_0\times (N_1-k_0-k_1)}(\mathcal{T}_2),$ such that the code $\mathscr{C}_2$ has a generator matrix of the form \eqref{eqnLem1.1}. Here, we need to count the choices for the matrices $A'_{0,2},$  $B_{0,1}$ and  $C_{0,2}.$ To do this, we see, by Lemma \ref{Lem6.1}, that the code $\mathscr{C}_2$ is Euclidean self-orthogonal if and only if 
$R_1^{(2)}\diamond (R_1^{(2)})^T\equiv 0~(\text{mod}~\gamma^2),$ $R_1^{(2)}\diamond (R_2^{(2)})^T\equiv 0~(\text{mod}~\gamma),$ $R_1^{(2)}\diamond (S_1^{(2)})^T\equiv 0~(\text{mod}~\gamma^2)$ and $S_1^{(2)}\diamond (S_1^{(2)})^T\equiv 0~(\text{mod}~\gamma^2),$ which is equivalent to the following system of matrix congruences: \vspace{-2mm}\begin{align}
I_{k_0}+A_{0,1}A_{0,1}^T+A_{0,2}A_{0,2}^T+\gamma(A_{0,2}A_{0,2}'^{T}+A_{0,2}'A_{0,2}^T+B_{0,1}B_{0,1}^T)~\equiv 0~(\text{mod}~\gamma^2),\label{eqn1.6}\\
A_{0,1}+ A_{0,2}A_{1,2}^T~\equiv 0~(\text{mod}~\gamma)\label{eqn1.7},\\
 A_{0,2}C_{0,2}^T+ B_{0,1}D_{0,1}^T~\equiv 0~(\text{mod}~\gamma)\label{eqn1.8},\\
 I_{\ell_0}+D_{0,1}D_{0,1}^T~\equiv 0~(\text{mod}~\gamma)\label{eqn1.9}.
\vspace{-2mm}\end{align}
 Since $C_2 \subseteq C_1^{\perp_E}$ and $C_3 \subseteq C_3^{\perp_E},$  we see that 
 \eqref{eqn1.7} and \eqref{eqn1.9} hold. We first choose the matrix $B_{0,1}$ arbitrarily, \textit{i.e.,} the matrix $B_{0,1}$ has precisely $q^{k_0(N_2-\ell_0)}$ distinct choices. Since $C_1\subseteq C_1^{\perp_E}$ and the matrix $A_{0,2}$ is of full row-rank, we see, by applying Lemma \ref{Matrixlemma}(a), that for a given choice of $B_{0,1},$ there are precisely $q^{\frac{k_0}{2}(2N_1-3k_0-2k_1-1)}$ distinct choices for the matrix $A_{0,2}'$  satisfying \eqref{eqn1.6}. Further, by Lemma \ref{Matrixlemma}(b), we see that there are  precisely $q^{\ell_0(N_1-2k_0-k_1)}$ distinct choices for the matrix $C_{0,2}$   satisfying  \eqref{eqn1.8}. Moreover, each of the distinct choices for the matrices $A_{0,2}',$  $B_{0,1}$ and  $C_{0,2}$ gives rise to a   distinct Euclidean self-orthogonal $\mathtt{R}_2\mathtt{R}_1$-linear code of block-length $(N_1,N_2)$ satisfying $Tor_1(\mathscr{C}_2^{(X)})=C_1,$ $Tor_2(\mathscr{C}_2^{(X)})=C_2$ and $Tor_1(\mathscr{C}_2^{(Y)})=C_3.$   From this, the desired result follows. 
\end{proof}

In the following theorem, we obtain the number $\mathfrak{N}_{2}(N_1,N_2;\textbf{k};\boldsymbol{\ell}),$ where $\textbf{k}=(k_0,k_1)$ and $\boldsymbol{\ell}=(\ell_0).$
\vspace{-1mm}\begin{thm}\label{Thm4.2} For $\textbf{k}=(k_0,k_1)$ and $\boldsymbol{\ell}=(\ell_0),$ we have 
\vspace{-2mm}\begin{equation}
\mathfrak{N}_{2}(N_1,N_2;\textbf{k}; \boldsymbol{\ell})=\left\{\begin{array}{cl}\sigma_q(N_1,k_0)\sigma_q(N_2,\ell_0)\qbin{
N_1-2k_0}{
k_1}{q} q^{\Theta_2(\textbf{k}; \boldsymbol{\ell})} & \text{if $2k_0+k_1\leq N_1$ and $2\ell_0\leq N_2$};\\
0 & \text{otherwise,}\end{array} \right.\vspace{-2mm}\end{equation}
where the numbers  $\sigma_q(N_1,k_0)$\textquotesingle s and $\sigma_q(N_2,\ell_0)$\textquotesingle s are given by \eqref{sigmaE}, and $\Theta_2(\textbf{k}; \boldsymbol{\ell})$ is given by \vspace{-2mm}\begin{equation}\label{theta2}\Theta_2(\textbf{k}; \boldsymbol{\ell})=k_0(N_1+N_2-\tfrac{3}{2}k_0-k_1-\ell_0-\tfrac{1}{2})+\ell_0(N_1-2k_0-k_1).\vspace{-2mm}\end{equation}
\end{thm}
\vspace{-3mm}\begin{proof}Let $\mathscr{D}_2$ be a Euclidean self-orthogonal $\mathtt{R}_2\mathtt{R}_1$-linear  code of block-length $(N_1,N_2)$ and type $\{k_0,k_1;\ell_0\}$  satisfying $Tor_1(\mathscr{D}_2^{(X)})=D_1,$ $Tor_2(\mathscr{D}_2^{(X)})=D_2,$ and $Tor_1(\mathscr{D}_2^{(Y)})=D_3.$ Here, we have $\dim D_1=k_0,$ $\dim D_2=k_0+k_1$ and $\dim D_3=\ell_0.$ Further, by Lemma \ref{Lem1.1}, we see that $D_1\subseteq D_2\subseteq D_1^{\perp_E}$ and $D_3\subseteq D_3^{\perp_E}.$  Also, by Remark \ref{Remark1.1}, we note that  $2k_0+k_1\leq N_1$ and $2\ell_0\leq N_2.$ We further see, by equation \eqref{sigmaE}, that the codes $D_1$ and $D_3$ have precisely $\sigma_q(N_1,k_0)$ and $\sigma_q(N_2,\ell_0)$ distinct choices, respectively.  Further, for a given choice of $D_1,$ one can easily see that there are precisely $
\qbin{N_1-2k_0}{
k_1}{q}$ distinct choices for the code $D_2$ satisfying $D_1\subseteq D_2\subseteq D_1^{\perp_E}.$  From this and by applying Lemma \ref{Lem4.1} and Proposition \ref{Thm4.1}, the desired result follows.
\end{proof}
In the following theorem, we obtain an enumeration formula for $\mathfrak{N}_{2}(N_1,N_2).$
\vspace{-1mm}\begin{thm}\label{Thm4.3}
 We have 
 \vspace{-2mm}\begin{equation*}
\mathfrak{N}_{2}(N_1,N_2)=\sum\limits_{k_0=0}^{\floor{\frac{N_1}{2}}}\sum\limits_{k_1=0}^{N_1-2k_0}\sum\limits_{\ell_0=0}^{\floor{\frac{N_2}{2}}}\sigma_q(N_1,k_0)\sigma_q(N_2,\ell_0)\qbin{
N_1-2k_0}{
k_1}{q}q^{k_0(N_1+N_2-\frac{3}{2}k_0-k_1-\ell_0-\frac{1}{2})+\ell_0(N_1-2k_0-k_1)}, 
 \vspace{-2mm}\end{equation*}
 where the numbers $\sigma_q(N_1,k_0)$\textquotesingle s and $\sigma_q(N_2,\ell_0)$\textquotesingle s  are given by \eqref{sigmaE}.
\end{thm}
\vspace{-3mm}\begin{proof}
It follows immediately from Theorem \ref{Thm4.2}.
\end{proof}
\vspace{-4mm}\begin{example}\label{ex4.1} Using Magma, we see that there are precisely $5$ and $2635$ distinct non-zero Euclidean self-orthogonal $\mathbb{Z}_9\mathbb{Z}_3$-linear codes of block-lengths $(2,2)$ and $(3,3),$ respectively.  This agrees with Theorem \ref{Thm4.3}.
\end{example}
As a consequence of the above theorem, we obtain an enumeration formula for self-orthogonal additive codes  of length $N$ over $\mathcal{R}_2$ in the following corollary.
\vspace{-1mm}\begin{cor}\label{Corollary3.1}
The number of self-orthogonal additive codes  of length $N$ over $\mathcal{R}_2$ is given by
\vspace{-2mm}$$\sum\limits_{k_0=0}^{\floor{\frac{Nrt}{2}}}\sum\limits_{k_1=0}^{Nrt-2k_0}\sum\limits_{\ell_0=0}^{\floor{\frac{Nr(k-t)}{2}}}\sigma_p(Nrt,k_0)\sigma_p\bigl(Nr(k-t),\ell_0\bigr)\qbin{
Nrt-2k_0}{
k_1}{p}p^{k_0(Nrk-\frac{3}{2}k_0-k_1-\ell_0-\frac{1}{2})+\ell_0(Nrt-2k_0-k_1)},\vspace{-2mm}$$
where the numbers $\sigma_p(Nrt,k_0)$\textquotesingle s and $\sigma_p\bigl(Nr(k-t),\ell_0\bigr)$\textquotesingle s  are given by \eqref{sigmaE}.
\end{cor}
\vspace{-3mm}\begin{proof} By Remark \ref{KEY}, we see that  the number of self-orthogonal additive codes  of length $N$ over $\mathcal{R}_2$ is equal to the number of distinct Euclidean self-orthogonal $\mathbb{Z}_{p^2}\mathbb{Z}_{p}$-linear codes of block-length $\bigl(Nrt,Nr(k-t)\bigr).$ Now on taking $N_1=Nrt$ and $N_2=Nr(k-t)$ in Theorem  \ref{Thm4.3}, the desired result follows immediately.
\end{proof}
\vspace{-4mm}\begin{example}By Corollary \ref{Corollary3.1}, we see that 
there are precisely $5$ and $2635$ distinct non-zero self-orthogonal additive codes of lengths $2$ and $3$ over $\mathbb{Z}_9[y]/\langle y^2-3,3y \rangle,$ respectively. This agrees with Example \ref{ex4.1} and Magma computations.
\end{example}
In the following theorem, we derive a necessary and sufficient condition under which a Euclidean self-dual $\mathtt{R}_2\mathtt{R}_1$-linear code of block-length $(N_1,N_2)$ exists. We further  obtain an enumeration formula for the number $\mathfrak{M}_{2}(N_1,N_2).$
\vspace{-1mm}\begin{thm}\label{Lemma3.3}
There exists a Euclidean self-dual $\mathtt{R}_2\mathtt{R}_1$-linear code of block-length $(N_1,N_2)$ if and only if $N_2$ is even and $(-1)^{\frac{N_2}{2}}$ is a square in $\mathcal{T}_e(\simeq \mathbb{F}_{q}).$ Furthermore, when $N_2$ is even and $(-1)^{\frac{N_2}{2}}$ is a square in $\mathcal{T}_e,$ we have 
\vspace{-2mm}\begin{equation*}
\mathfrak{M}_{2}(N_1,N_2)=\sum\limits_{k_0=0}^{\floor{\frac{N_1}{2}}}\sigma_q(N_1,k_0)\sigma_q(N_2,\tfrac{N_2}{2})q^{\frac{k_0}{2}(N_2+k_0-1)}, 
 \vspace{-2mm}\end{equation*}
 where the numbers $\sigma_q(N_1,k_0)$\textquotesingle s and $\sigma_q(N_2,\frac{N_2}{2})$\textquotesingle s  are given by \eqref{sigmaE}.
 \end{thm}
\vspace{-3mm}\begin{proof}
It follows from Theorem \ref{Thm4.3}, Lemma \ref{Lem1.1} and Remarks \ref{Remark1.1} and \ref{Remark2.2}.
 \end{proof}
\vspace{-4mm}\begin{example}\label{Eg.4.3} Using Magma, we see that there are precisely $22$ distinct Euclidean self-dual $\mathbb{Z}_{25}\mathbb{Z}_5$-linear codes of block-length $(2,2).$ This agrees with Theorem \ref{Lemma3.3}.
\end{example}
 In the following corollary, we derive a necessary and sufficient condition for the existence of a self-dual additive code of length $N$ over $\mathcal{R}_2.$ We also obtain an enumeration formula for self-dual additive codes  of length $N$ over $\mathcal{R}_2.$ 
\vspace{-1mm}\begin{cor}\label{Corollary3.2}
There exists a self-dual additive code  of length $N$ over $\mathcal{R}_2$ if and only if $Nr(k-t)$ is even and $(-1)^{\frac{Nr(k-t)}{2}}$ is a square in $\mathbb{Z}_p.$ Furthermore, when $Nr(k-t)$ is even and $(-1)^{\frac{Nr(k-t)}{2}}$ is a square in $\mathbb{Z}_p,$  the number of self-dual additive codes  of length $N$ over $\mathcal{R}_2$ is given by
\vspace{-2mm}\small$$\sum\limits_{k_0=0}^{\floor{\frac{Nrt}{2}}}\sigma_p(Nrt,k_0)\sigma_p\bigl(Nr(k-t),\tfrac{Nr(k-t)}{2}\bigr)p^{\frac{k_0}{2}\bigl(Nr(k-t)+k_0-1\bigr)},\vspace{-2mm}$$
\normalsize where the numbers $\sigma_p(Nrt,k_0)$\textquotesingle s and $\sigma_p\bigl(Nr(k-t),\frac{Nr(k-t)}{2}\bigr)$\textquotesingle s are given by \eqref{sigmaE}.
\end{cor}
\vspace{-3mm}\begin{proof} By Remark \ref{KEY}, we see that  the number of self-dual additive codes  of length $N$ over $\mathcal{R}_2$ is equal to the number of distinct Euclidean self-dual $\mathbb{Z}_{p^2}\mathbb{Z}_{p}$-linear codes of block-length $\bigl(Nrt,Nr(k-t)\bigr).$ Now on taking $N_1=Nrt$ and $N_2=Nr(k-t)$ in Theorem  \ref{Lemma3.3}, we get the desired result.
\end{proof}
\vspace{-3mm}\begin{example}By Corollary \ref{Corollary3.2}, we see that 
there are precisely $22$ distinct self-dual additive codes of length $2$ over $\mathbb{Z}_{25}[y]/\langle y^2-5,5y \rangle.$ This agrees with Example \ref{Eg.4.3} and Magma computations.
\end{example}
\vspace{-5mm}\subsection{The case $e=3$}
\vspace{-1mm}
Throughout this section, let us suppose that $e=3.$ Let $\mathscr{C}_3$ be an $\mathtt{R}_3\mathtt{R}_{2}$-linear code of block-length $(N_1,N_2)$ and type $\{k_0,k_1,k_2;\ell_0,\ell_1\}.$ Since $\mathcal{T}_3=\overline{\mathtt{R}}_3$ is the finite field of order $q,$  we note, by Section \ref{prelim3}, that the torsion codes $Tor_1(\mathscr{C}_3^{(X)}),$ $Tor_2(\mathscr{C}_3^{(X)})$ and $Tor_3(\mathscr{C}_3^{(X)})$ are linear codes of length $N_1$ and dimensions $k_0,$ $k_0+k_1$ and $k_0+k_1+k_2$ over $\mathcal{T}_3,$ respectively, while the torsion codes  $Tor_1(\mathscr{C}_3^{(Y)})$ and $Tor_2(\mathscr{C}_3^{(Y)})$ are linear codes of length $N_2$ and dimensions $\ell_0$ and $\ell_0+\ell_1$ over $\mathcal{T}_3,$ respectively. Further, if the code $\mathscr{C}_3$ is Euclidean self-orthogonal, then by Lemma \ref{Lem1.1}, we see that 
\vspace{-2mm}\begin{align*}
Tor_1(\mathscr{C}_3^{(X)})\subseteq Tor_2(\mathscr{C}_3^{(X)})\subseteq Tor_3(\mathscr{C}_3^{(X)})&\subseteq Tor_1(\mathscr{C}_3^{(X)})^{\perp_E},
\vspace{-1mm}\\ Tor_1(\mathscr{C}_3^{(Y)})\subseteq Tor_2(\mathscr{C}_3^{(Y)})&\subseteq Tor_1(\mathscr{C}_3^{(Y)})^{\perp_E}~\text{ and}\vspace{-1mm}\\ Tor_2(\mathscr{C}_3^{(X)})&\subseteq{Tor_2(\mathscr{C}_3^{(X)})}^{\perp_E}.
\vspace{-2mm}\end{align*}
Further, by Remark \ref{Remark1.1}, we have $2k_0+2k_1\leq N_1,$ $2k_0+k_1+k_2\leq N_1$  and $2\ell_0+\ell_1\leq N_2.$

Here, we will first count all Euclidean self-orthogonal $\mathtt{R}_3\mathtt{R}_2$-linear codes of block-length $(N_1,N_2)$ and type $\{k_0,k_1,k_2;\ell_0,\ell_1\}$ with prescribed torsion codes. To do this, let $E_1,$ $E_2$ and $ E_3$ be linear codes of length $N_1$ over $\mathcal{T}_3$ with dimensions $k_0,$ $k_0+k_1$ and $k_0+k_1+k_2,$ and generator matrices  \vspace{-2mm}$$[I_{k_0} ~ A_{0,1} ~ A_{0,2} ~ A_{0,3}],~~\left[\begin{array}{cccc}
    I_{k_0} & A_{0,1} & A_{0,2} & A_{0,3}\\
    0 & I_{k_1} & A_{1,2} & A_{1,3}
\end{array}\right]\text{~ and ~}\left[\begin{array}{cccc}
    I_{k_0} & A_{0,1} & A_{0,2} & A_{0,3}\\
    0 & I_{k_1} & A_{1,2} & A_{1,3}\\
    0 & 0 & I_{k_2} & A_{2,3}
\end{array}\right],\vspace{-2mm}$$ respectively, where  $A_{i,j}\in M_{k_i\times k_j}(\mathcal{T}_3)$ for $0\leq i<j\leq 2$ and $A_{i,3}\in M_{k_i\times (N_1-k_0-k_1-k_2)}(\mathcal{T}_3)$ for $0\leq i\leq 2.$ Also, let  $E_4$ and $E_5$ be linear codes of length $N_2$  over $\mathcal{T}_3$ with  dimensions $\ell_0$ and $\ell_0+\ell_1,$ and  generator matrices \vspace{-2mm}$$[I_{\ell_0} ~ D_{0,1} ~ D_{0,2}]\text{ ~ and ~}\left[\begin{array}{cccc}
    I_{\ell_0} & D_{0,1} & D_{0,2}\\
    0 & I_{\ell_1} & D_{1,2}
\end{array}\right],\vspace{-2mm}$$ respectively,  
where  $D_{0,1}\in M_{\ell_0\times \ell_1}(\mathcal{T}_3)$ and $D_{i,2}\in M_{\ell_i\times (N_2-\ell_0-\ell_1)}(\mathcal{T}_3)$ for $0\leq i\leq 1.$ Clearly, $E_1\subseteq E_2\subseteq E_3$ and $ E_4\subseteq E_5.$ Now, we have the following lemma. 
\vspace{-2mm}\begin{lemma}\label{Lem5.1}
If $\mathscr{C}_3$ is an $\mathtt{R}_3\mathtt{R}_2$-linear code of block-length $(N_1,N_2)$ with $Tor_1(\mathscr{C}_3^{(X)})=E_1,$ $Tor_2(\mathscr{C}_3^{(X)})=E_2,$ $Tor_3(\mathscr{C}_3^{(X)})=E_3,$ $Tor_1(\mathscr{C}_3^{(Y)})=E_4$ and $Tor_2(\mathscr{C}_3^{(Y)})=E_5,$ then there exist matrices $ A'_{0,2}\in M_{k_0\times k_2}(\mathcal{T}_3),$ $A'_{0,3},A''_{0,3}\in M_{k_0\times (N_1-k_0-k_1-k_2)}(\mathcal{T}_3),$ $ A'_{1,3}\in M_{k_1\times (N_1-k_0-k_1-k_2)}(\mathcal{T}_3),$ $C_{0,2}\in M_{\ell_0\times k_2}(\mathcal{T}_3),$ $B_{0,1}\in M_{k_0\times \ell_1}(\mathcal{T}_3),$ $C_{i,3}\in M_{\ell_i\times (N_1-k_0-k_1-k_2)}(\mathcal{T}_3)$ for $0\leq i\leq 1,$ $C'_{0,3}\in M_{\ell_0\times (N_1-k_0-k_1-k_2)}(\mathcal{T}_3),$  $B_{i,2}\in M_{k_i\times (N_2-\ell_0-\ell_1)}(\mathcal{T}_3)$ for $0\leq i\leq 1,$ $B_{0,2}'\in M_{k_0\times (N_2-\ell_0-\ell_1)}(\mathcal{T}_3)$ and $D_{0,2}'\in M_{\ell_0\times (N_2-\ell_0-\ell_1)}(\mathcal{T}_3),$ such that the matrix
\vspace{-2mm}\begin{align}\label{eqnLem5.1}
\left[\begin{array}{cccc|ccc}
    I_{k_0} & A_{0,1} & A_{0,2}+\gamma A_{0,2}' & A_{0,3}+\gamma A_{0,3}'+\gamma^2 A_{0,3}'' & 0 & B_{0,1} & B_{0,2}+\gamma B_{0,2}' \\
    0 & \gamma I_{k_1} & \gamma A_{1,2} & \gamma A_{1,3}+\gamma^2 A_{1,3}' & 0 & 0 & \gamma B_{1,2} \\
    0 & 0 & \gamma^2 I_{k_2} & \gamma^2 A_{2,3} & 0 & 0 & 0
    \\ 0 & 0 & \gamma C_{0,2} & \gamma C_{0,3}+\gamma^2 C_{0,3}' & I_{\ell_0} & D_{0,1} & D_{0,2}+\gamma D_{0,2}' \\
    0 & 0 & 0 & \gamma^2 C_{1,3} & 0 & \gamma I_{\ell_1} & \gamma D_{1,2}
\end{array}\right]=\left[\begin{array}{c}
    R_1^{(3)}\\
    \gamma R_2^{(3)}\\
    \gamma^2 R_3^{(3)}\\
    S_1^{(3)}\\
    \gamma S_2^{(3)}
\end{array}\right]\vspace{-2mm}
\end{align} is a generator matrix of the code $\mathscr{C}_3.$
\end{lemma}
\begin{proof}
Working as in Lemma 3.1 of Yadav and Sharma \cite{Yadav} and choosing the matrices $B_{0,1},$ $B_{0,2},$ $B_{1,2},$ $B_{0,2}',$ $C_{0,2},$ $C_{0,3},$ $C_{1,3}$   and $C_{0,3}'$ arbitrarily, the desired result follows. 
\end{proof}

For the rest of this section, we assume that the codes $E_1,$ $E_2,$ $E_3,$ $E_4$ and $E_5$ satisfy $E_1\subseteq E_2\subseteq E_3\subseteq E_1^{\perp_E},$ $ E_2\subseteq E_2^{\perp_E}$ and $ E_4\subseteq E_5\subseteq E_4^{\perp_E}.$ This hold if and only if 
\vspace{-2mm}\begin{align}
I_{k_0}+A_{0,1}A_{0,1}^T+A_{0,2}A_{0,2}^T+A_{0,3}A_{0,3}^T&\equiv0~\text{(mod $\gamma$)},\label{eqn2.8}\\
I_{k_1}+A_{1,2}A_{1,2}^T+A_{1,3}A_{1,3}^T&\equiv0~\text{(mod $\gamma$)},\label{eqn2.9}\\
A_{0,1}+A_{0,2}A_{1,2}^T+A_{0,3}A_{1,3}^T&\equiv0~\text{(mod $\gamma$)},\label{eqn2.10}\\
A_{0,2}+A_{0,3}A_{2,3}^T&\equiv0~\text{(mod $\gamma$)},\label{eqn2.11}\\
I_{\ell_0}+D_{0,1}D_{0,1}^T+D_{0,2}D_{0,2}^T&\equiv0~\text{(mod $\gamma$)}\label{eqn2.12}~\text{and}\\
D_{0,1}+D_{0,2}D_{1,2}^T&\equiv0~\text{(mod $\gamma$)}.\label{eqn2.13}
\vspace{-2mm}\end{align} 
Here, it is easy to see that
$2k_0+2k_1\leq N_1,$ $2k_0+k_1+k_2\leq N_1$ and $2\ell_0+\ell_1\leq N_2.$ Further,  working as in Remark 2.2 of Yadav and Sharma \cite{Yadav}, one can observe that the matrices $A_{0,3}$ and $D_{0,2}$ are of full row-ranks. 
Now, in the following proposition, we will count  all Euclidean self-orthogonal $\mathtt{R}_3\mathtt{R}_2$-linear codes $\mathscr{C}_3$ of block-length $(N_1,N_2)$ satisfying $Tor_1(\mathscr{C}_3^{(X)})= E_1,$ $Tor_2(\mathscr{C}_3^{(X)})= E_2,$ $Tor_3(\mathscr{C}_3^{(X)})= E_3,$ $Tor_1(\mathscr{C}_3^{(Y)})= E_4$ and $Tor_2(\mathscr{C}_3^{(Y)})= E_5.$ 
\vspace{-1mm}\begin{prop}\label{Thm5.1}
There are precisely \vspace{-2mm}$$ q^{(N_1-3k_0-k_1-k_2)(k_0+2\ell_0+k_1+\ell_1)+k_0(N_1+2N_2-1)+k_1(N_2-2\ell_0-\ell_1)+\ell_0(N_2-\frac{3}{2}\ell_0-\ell_1+k_2-\frac{1}{2})}\vspace{-2mm}$$ distinct  Euclidean self-orthogonal $\mathtt{R}_3\mathtt{R}_2$-linear  codes $\mathscr{C}_3$ of block-length $(N_1,N_2)$  satisfying $Tor_1(\mathscr{C}_3^{(X)})= E_1,$ $Tor_2(\mathscr{C}_3^{(X)})= E_2,$ $Tor_3(\mathscr{C}_3^{(X)})= E_3,$ $Tor_1(\mathscr{C}_3^{(Y)})= E_4$ and $Tor_2(\mathscr{C}_3^{(Y)})= E_5.$ 
\end{prop}
\begin{proof}
Let $\mathscr{C}_3$ be a Euclidean self-orthogonal $\mathtt{R}_3\mathtt{R}_2$-linear code of block-length $(N_1,N_2)$  satisfying $Tor_1(\mathscr{C}_3^{(X)})= E_1,$ $Tor_2(\mathscr{C}_3^{(X)})= E_2,$ $Tor_3(\mathscr{C}_3^{(X)})= E_3,$ $Tor_1(\mathscr{C}_3^{(Y)})= E_4$ and $Tor_2(\mathscr{C}_3^{(Y)})= E_5.$ Since $ E_3\subseteq E_1^{\perp_E},$ $ E_2\subseteq E_2^{\perp_E}$ and $ E_5\subseteq E_4^{\perp_E},$ we note that \eqref{eqn2.8}-\eqref{eqn2.13} hold. Now by Lemma \ref{Lem5.1}, we see that there exist matrices $A'_{0,2}\in M_{k_0\times k_2}(\mathcal{T}_3),$ $A'_{0,3},A''_{0,3}\in M_{k_0\times (N_1-k_0-k_1-k_2)}(\mathcal{T}_3),$ $A'_{1,3}\in M_{k_1\times (N_1-k_0-k_1-k_2)}(\mathcal{T}_3),$ $C_{0,2}\in M_{\ell_0\times k_2}(\mathcal{T}_3),$ $C_{i,3}\in M_{\ell_i\times (N_1-k_0-k_1-k_2)}(\mathcal{T}_3)$ for $0\leq i\leq 1,$ $C'_{0,3}\in M_{\ell_0\times (N_1-k_0-k_1-k_2)}(\mathcal{T}_3),$ $B_{0,1}\in M_{k_0\times \ell_1}(\mathcal{T}_3),$ $B_{i,2}\in M_{k_i\times (N_2-\ell_0-\ell_1)}(\mathcal{T}_3)$ for $0\leq i\leq 1,$ $B_{0,2}'\in M_{k_0\times (N_2-\ell_0-\ell_1)}(\mathcal{T}_3)$ and $D_{0,2}'\in M_{\ell_0\times (N_2-\ell_0-\ell_1)}(\mathcal{T}_3),$ such that the code $\mathscr{C}_3$ has a generator matrix of the form \eqref{eqnLem5.1}. Further, by Lemma \ref{Lem6.1}, we see that the code $\mathscr{C}_3$ is Euclidean self-orthogonal if and only if 
\vspace{-2mm}\begin{equationarray}{rl}
I_{k_0}+A_{0,1}A_{0,1}^T+A_{0,2}A_{0,2}^T+A_{0,3}A_{0,3}^T+\gamma(A_{0,2}A_{0,2}'^{T}+A_{0,2}'A_{0,2}^T+A_{0,3}A_{0,3}'^{T}&~\nonumber\vspace{1mm}\\ +A_{0,3}'A_{0,3}^T +B_{0,1}B_{0,1}^T+B_{0,2}B_{0,2}^T) +\gamma^2(A_{0,2}'A_{0,2}'^{T}+A_{0,3}A_{0,3}''^{T}+A_{0,3}''A_{0,3}^T&~ \nonumber \vspace{1mm}\\+A_{0,3}'A_{0,3}'^{T}+B_{0,2}B_{0,2}'^{T}+B_{0,2}'B_{0,2}^T)&\equiv0\Mod{\gamma^3}, \label{eqn03.15}\vspace{1mm}\\
A_{0,2}C_{0,2}^T+A_{0,3}C_{0,3}^T+B_{0,1}D_{0,1}^T+B_{0,2}D_{0,2}^T+\gamma(A_{0,2}'C_{0,2}^T+A_{0,3}C_{0,3}'^{T}+A_{0,3}'C_{0,3}^T&~\nonumber  \vspace{1mm}\\+B_{0,2}D_{0,2}'^{T}+B_{0,2}'D_{0,2}^T)&\equiv0\Mod{\gamma^2},\label{eqn03.18}\vspace{1mm}\\
A_{0,1}+A_{0,2}A_{1,2}^T+A_{0,3}A_{1,3}^T+\gamma(A_{0,2}'A_{1,2}^T+A_{0,3}A_{1,3}'^{T}+A_{0,3}'A_{1,3}^T+B_{0,2}B_{1,2}^T)&\equiv0\Mod{\gamma^2}, \label{eqn03.16}\vspace{1mm}\\
I_{\ell_0}+D_{0,1}D_{0,1}^T+D_{0,2}D_{0,2}^T+\gamma(C_{0,2}C_{0,2}^T+C_{0,3}C_{0,3}^T+D_{0,2}D_{0,2}'^{T}+D_{0,2}'D_{0,2}^T)&\equiv0\Mod{\gamma^2},\label{eqn03.22}\vspace{1mm}\\
A_{0,3}C_{1,3}^T+B_{0,1}+B_{0,2}D_{1,2}^T   &\equiv0\Mod{\gamma},\label{eqn03.19}\vspace{1mm}\\
C_{0,2}A_{1,2}^T+C_{0,3}A_{1,3}^T+D_{0,2}B_{1,2}^T&\equiv 0\Mod{\gamma},\label{eqn03.21}\vspace{1mm}\\
I_{k_1}+A_{1,2}A_{1,2}^T+A_{1,3}A_{1,3}^T&\equiv 0\Mod{\gamma},\label{eqn03.20}\vspace{1mm}\\
A_{0,2}+A_{0,3}A_{2,3}^T &\equiv 0\Mod{\gamma}~\text{and}\label{eqn03.17}\vspace{1mm}\\
D_{0,1}+D_{0,2}D_{1,2}^T  &\equiv 0\Mod{\gamma}.\label{eqn03.23}
\vspace{-2mm}\end{equationarray}

By \eqref{eqn2.9}, \eqref{eqn2.11}   and \eqref{eqn2.13}, we see that the congruences  \eqref{eqn03.20}-\eqref{eqn03.23} hold.  So we need to count the choices for the matrices $A'_{0,2},$ $A'_{0,3},$ $A''_{0,3},$ $A'_{1,3},$ $C_{0,2},$ $C_{0,3},$ $C_{1,3},$  $C'_{0,3},$ $B_{0,1},$ $B_{0,2},$ $B_{1,2},$  $B_{0,2}'$ and $D_{0,2}'$ satisfying the  congruences \eqref{eqn03.15}-\eqref{eqn03.21}. 
Towards this, we first note that the congruence  \eqref{eqn2.8} gives
\vspace{-1mm}\begin{align*}
I_{k_0}+A_{0,1}A_{0,1}^T+A_{0,2}A_{0,2}^T+A_{0,3}A_{0,3}^T\equiv\gamma\mathcal{P}_1+\gamma^2\mathcal{P}_2~\text{(mod~$\gamma^3$)}  \end{align*}
for some $\mathcal{P}_1,\mathcal{P}_2\in Sym_{k_0}(\mathcal{T}_3)$. On substituting this into the congruence   \eqref{eqn03.15}, we get
\vspace{-1mm}\begin{align}\nonumber
&\gamma(\mathcal{P}_1+A_{0,2}A_{0,2}'^{T}+A_{0,2}'A_{0,2}^T+A_{0,3}A_{0,3}'^{T}+A_{0,3}'A_{0,3}^T+B_{0,1}B_{0,1}^T+B_{0,2}B_{0,2}^T)\\
&+\gamma^2(\mathcal{P}_2+A_{0,2}'A_{0,2}'^{T}+A_{0,3}A_{0,3}''^{T}+A_{0,3}''A_{0,3}^T+A_{0,3}'A_{0,3}'^{T}+B_{0,2}B_{0,2}'^{T}+B_{0,2}'B_{0,2}^T)\equiv0~\text{(mod~$\gamma^3$)}.\label{eqn2.19}
\end{align}
Here, we will choose the matrices $B_{0,1},B_{0,2},$ $B_{0,2}'$ and $A_{0,2}'$  arbitrarily,  which can be done in $q^{k_0(b-2\ell_0-\ell_1+k_2)}$ distinct ways. Now by Lemma \ref{Matrixlemma}(a), we see that there are precisely $q^{\frac{k_0}{2}(2a-3k_0-2k_1-2k_2-1)}$ distinct choices for the matrix $A_{0,3}'$  satisfying 
\vspace{-1mm}\begin{align*}
\mathcal{P}_1+A_{0,2}A_{0,2}'^{T}+A_{0,2}'A_{0,2}^T+A_{0,3}A_{0,3}'^{T}+A_{0,3}'A_{0,3}^T+B_{0,1}B_{0,1}^T+B_{0,2}B_{0,2}^T\equiv0~\text{(mod $\gamma$).}
\vspace{-2mm}\end{align*}
Thus, we get 
\vspace{-1mm}\begin{align*}
\gamma(\mathcal{P}_1+A_{0,2}A_{0,2}'^{T}+A_{0,2}'A_{0,2}^T+A_{0,3}A_{0,3}'^{T}+A_{0,3}'A_{0,3}^T+B_{0,1}B_{0,1}^T+B_{0,2}B_{0,2}^T)\equiv\gamma^2\mathcal{P}_3~\text{(mod $\gamma^3$)}
\vspace{-2mm}\end{align*}
for some $\mathcal{P}_3\in Sym_{k_0}(\mathcal{T}_3).$ Now on substituting this into the congruence \eqref{eqn2.19} and applying Lemma \ref{Matrixlemma}(a) again, we see that there are precisely $q^{\frac{k_0}{2}(2a-3k_0-2k_1-2k_2-1)}$ distinct choices for the matrix $A_{0,3}''$ satisfying \eqref{eqn2.19}. 

Now, for given choices of the matrices $B_{0,1},B_{0,2},A_{0,2}',A_{0,3}',B_{0,2}'$ and $A_{0,3}'',$  we will  count all possible   choices for the matrices $C_{0,2},C_{0,3},C_{0,3}',D_{0,2}',B_{1,2},A_{1,3}'$ and $C_{1,3}$ satisfying  \eqref{eqn03.18}-\eqref{eqn03.21}.
To do this, we choose the matrix $C_{0,2}$  arbitrarily, which can be chosen in  $q^{\ell_0k_2}$ distinct ways. Further, by Lemma \ref{Matrixlemma}(b), we see that there are precisely $q^{\ell_0(a-2k_0-k_1-k_2)}$ distinct choices for the matrix $C_{0,3}$ satisfying $A_{0,2}C_{0,2}^T+A_{0,3}C_{0,3}^T+B_{0,1}D_{0,1}^T+B_{0,2}D_{0,2}^T\equiv0~\text{(mod $\gamma$)}.$ From this, we get $A_{0,2}C_{0,2}^T+A_{0,3}C_{0,3}^T+B_{0,1}D_{0,1}^T+B_{0,2}D_{0,2}^T\equiv\gamma\mathcal{Q}_1~\text{(mod $\gamma^2$)}$ for some $\mathcal{Q}\in M_{k_0\times\ell_0}(\mathcal{T}_3).$ Furthermore, on substituting this into the congruence \eqref{eqn03.18}, we get
\vspace{-1mm}\begin{align}\label{eqn2.21}
\gamma(\mathcal{Q}_1+A_{0,2}'C_{0,2}^T+A_{0,3}C_{0,3}'^{T}+A_{0,3}'C_{0,3}^T+B_{0,2}D_{0,2}'^{T}+B_{0,2}'D_{0,2}^T)\equiv0~\text{(mod $\gamma^2$)}.
\vspace{-1mm}\end{align}

We next proceed to count the choices for the matrices $C_{0,3}'$ and $D_{0,2}'$ satisfying \eqref{eqn2.21}.  To count the choices for the matrix $D_{0,2}',$ we note that \eqref{eqn2.12} gives
\vspace{-1mm}\begin{align*}
I_{\ell_0}+D_{0,1}D_{0,1}^T+D_{0,2}D_{0,2}^T\equiv\gamma\mathcal{P}_4~\text{(mod $\gamma^2$)}   
\vspace{-1mm}\end{align*}
for some $\mathcal{P}_4\in Sym_{\ell_0}(\mathcal{T}_3).$ On substituting this into the congruence  \eqref{eqn03.22}, we get
\vspace{-1mm}\begin{align}\label{eqn2.22}
\gamma(\mathcal{P}_4+C_{0,2}C_{0,2}^T+C_{0,3}C_{0,3}^T+D_{0,2}D_{0,2}'^{T}+D_{0,2}'D_{0,2}^T)\equiv0~\text{(mod $\gamma^2$)}.
\end{align}
Now, by Lemma \ref{Matrixlemma}(a), we see that there are precisely  $q^{\frac{\ell_0}{2}(2b-3\ell_0-2\ell_1-1)}$ distinct choices for the matrix $D_{0,2}'$ satisfying \eqref{eqn2.22}. Further, for a given choice of the matrix $D_{0,2}',$ we see, by Lemma \ref{Matrixlemma}(b), that there are precisely $q^{\ell_0(a-2k_0-k_1-k_2)}$ distinct choices for the matrix $C_{0,3}'$ satisfying \eqref{eqn2.21}.   

We will finally count the  choices for the remaining matrices $B_{1,2},A_{1,3}'$ and $C_{1,3}$ satisfying the congruences \eqref{eqn03.16}, \eqref{eqn03.19} and \eqref{eqn03.21} for given choices of the matrices $B_{0,1},B_{0,2},A_{0,2}',A_{0,3}',B_{0,2}',A_{0,3}'',C_{0,2},C_{0,3},C_{0,3}'$ and $D_{0,2}'.$ Towards this, we see, by Lemma \ref{Matrixlemma}(b), that there are precisely $q^{k_1(b-2\ell_0-\ell_1)}$ distinct choices for the matrix $B_{1,2}$ satisfying \eqref{eqn03.21}. Now to count the choices for the matrix $A_{1,3}',$ we note that the congruence \eqref{eqn2.10} gives 
\vspace{-1mm}\begin{align*}
A_{0,1}+A_{0,2}A_{1,2}^T+A_{0,3}A_{1,3}^T\equiv\gamma \mathcal{Q}_2~\text{(mod $\gamma^2$)}  \vspace{-1mm}\end{align*}
for some $\mathcal{Q}_2\in M_{k_0\times k_1}(\mathcal{T}_3).$ On substituting this into the congruence   \eqref{eqn03.16}, we get
\vspace{-1mm}\begin{align}\label{eqn2.24}
\gamma(\mathcal{Q}_2+A_{0,2}'A_{1,2}^T+A_{0,3}A_{1,3}'^{T}+A_{0,3}'A_{1,3}^T+B_{0,2}B_{1,2}^T)\equiv0~\text{(mod $\gamma^2$)}.
\end{align}
Next, by Lemma \ref{Matrixlemma}(b), we see that the matrix $A_{1,3}'$ satisfying \eqref{eqn2.24} has precisely $q^{k_1(a-2k_0-k_1-k_2)}$ distinct choices. 
Further, by Lemma \ref{Matrixlemma}(b) again, we see that the matrix $C_{1,3}$ satisfying \eqref{eqn03.19} can be chosen in $q^{\ell_1(a-2k_0-k_1-k_2)}$ distinct ways. Furthermore, each of the distinct choices for the matrices $A'_{0,2},$ $A'_{0,3},$ $A''_{0,3},$ $A'_{1,3},$ $C_{0,2},$ $C_{0,3},$ $C_{1,3},$ $C'_{0,3},$ $B_{0,1},$ $B_{0,2},$ $B_{1,2},$ $B_{0,2}'$ and $D_{0,2}'$ gives rise to a distinct Euclidean self-orthogonal $\mathtt{R}_3\mathtt{R}_2$-linear code $\mathscr{C}_3$ of block-length $(N_1,N_2)$ satisfying $Tor_1(\mathscr{C}_3^{(X)})= E_1,$ $Tor_2(\mathscr{C}_3^{(X)})= E_2,$ $Tor_3(\mathscr{C}_3^{(X)})= E_3,$ $Tor_1(\mathscr{C}_3^{(Y)})= E_4$ and $Tor_2(\mathscr{C}_3^{(Y)})= E_5.$  From this, the desired result follows.
\vspace{-4mm}\end{proof}

In the following theorem, we obtain the number $\mathfrak{N}_{3}(N_1,N_2;\textbf{k};\boldsymbol{\ell}),$ where $\textbf{k}=(k_0,k_1,k_2)$ and $\boldsymbol{\ell}=(\ell_0,\ell_1).$
\vspace{-1mm}\begin{thm}\label{Thm5.2}
For $\textbf{k}=(k_0,k_1,k_2)$ and $\boldsymbol{\ell}=(\ell_0,\ell_1),$ we have
\vspace{-1mm}\begin{equation*} \mathfrak{N}_{3}(N_1,N_2;\textbf{k};\boldsymbol{\ell})=\left\{\begin{array}{ll}\sigma_q(N_1,k_0+k_1)\sigma_q(N_2,\ell_0) \qbin{N_1-2k_0-k_1}{k_2}{q} 
\qbin{N_2-2 \ell_0}{\ell_1}{q} \qbin{k_0+k_1}{k_0}{q} q^{\Theta_3(\textbf{k};\boldsymbol{\ell})}\vspace{1mm} \\
\quad\text{if $2k_0+2k_1\leq N_1,$ $2k_0+k_1+k_2\leq N_1$ and $2\ell_0+\ell_1\leq N_2$};\vspace{2mm}\\
0\quad~\text{otherwise,}\end{array} \right.\vspace{-2mm}\end{equation*}
where the numbers $\sigma_q(N_1,k_0+k_1)$\textquotesingle s and $\sigma_q(N_2,\ell_0)$\textquotesingle s are given by \eqref{sigmaE}, and $\Theta_3(\textbf{k};\boldsymbol{\ell})$ is given by \vspace{-1mm}\small{\begin{equation}\label{theta3}\Theta_3(\textbf{k};\boldsymbol{\ell})=\bigl(N_1-3k_0-k_1-k_2\bigr)\bigl(k_0+2\ell_0+k_1+\ell_1\bigr)+k_0\bigl(N_1+2N_2-1\bigr)+k_1\bigl(N_2-2\ell_0-\ell_1\bigr)+\ell_0\bigl(N_2-\tfrac{3}{2}\ell_0-\ell_1+k_2-\tfrac{1}{2}\bigr).\vspace{-1mm}\end{equation}}\normalsize
\end{thm}
\begin{proof}
Let $\mathscr{D}_3$ be a Euclidean self-orthogonal $\mathtt{R}_3\mathtt{R}_2$-linear code of block-length $(N_1,N_2)$  and type $\{k_0,k_1,k_2;\ell_0,\ell_1\}$ satisfying $Tor_1(\mathscr{D}_3^{(X)})= F_1,$ $Tor_2(\mathscr{D}_3^{(X)})= F_2,$ 
$Tor_3(\mathscr{D}_3^{(X)})=F_3,$
$Tor_1(\mathscr{D}_3^{(Y)})=F_4$
and $Tor_2(\mathscr{D}_3^{(Y)})= F_5.$  Here, we have $\dim  F_1=k_0,$ $\dim  F_2=k_0+k_1,$ $\dim  F_3=k_0+k_1+k_2,$ $\dim  F_4=\ell_0$ and $\dim  F_5=\ell_0+\ell_1.$ Further, by Lemma \ref{Lem1.1}, we see that $ F_1\subseteq  F_2\subseteq  F_3\subseteq  F_1^{\perp_E},$ $  F_2\subseteq  F_2^{\perp_E}$ and $  F_4\subseteq  F_5\subseteq  F_4^{\perp_E}$ hold. Also, by Remark \ref{Remark1.1}, we note that $2k_0+2k_1\leq N_1,$ $2k_0+k_1+k_2\leq N_1$ and $2\ell_0+\ell_1\leq N_2.$ We further see, by equation \eqref{sigmaE}, that the codes $  F_2$ and $  F_4$ have precisely $\sigma_q(N_1, k_0+k_1)$ and $\sigma_q(N_2, \ell_0)$ distinct choices, respectively. Further, for a given choice of $  F_2,$ there are precisely $\qbin{k_0+k_1}{
k_0}{q}$ distinct choices for the code $  F_1$ satisfying $  F_1\subseteq  F_2,$ and hence the codes $  F_1$ and $  F_2$ satisfy $  F_1\subseteq  F_2\subseteq  F_2^{\perp_E}\subseteq  F_1^{\perp_E}.$ Further, for a given choice of $  F_2$ satisfying $  F_2\subseteq  F_1^{\perp_E},$ there are precisely $\qbin{
N_1-2k_0-k_1}{
k_2}{q}$ distinct choices for the code $  F_3$ satisfying $  F_2\subseteq  F_3\subseteq  F_1^{\perp_E}.$ Similarly, for a given choice of $  F_4,$ one can easily see that there are precisely $\qbin{
N_2-2\ell_0}{
\ell_1}{q}$ distinct choices for the code $  F_5$ satisfying $  F_4\subseteq  F_5\subseteq  F_4^{\perp_E}.$ From this and by applying Lemma \ref{Lem5.1} and Proposition \ref{Thm5.1}, the desired result follows.
\end{proof}
In the following theorem, we obtain an enumeration formula for $\mathfrak{N}_{3}(N_1,N_2).$
\vspace{-1mm}\begin{thm}\label{Thm5.3}
We have 
\vspace{-1mm}\small{\begin{align*}
\mathfrak{N}_{3}(N_1,N_2)=&\sum\limits_{k_0=0}^{\floor{\frac{N_1}{2}}}\sum\limits_{k_1=0}^{\floor{\frac{N_1}{2}}-k_0}\sum\limits_{k_2=0}^{N_1-2k_0-k_1}\sum\limits_{\ell_0=0}^{\floor{\frac{N_2}{2}}}\sum\limits_{\ell_1=0}^{N_2-2\ell_0}\sigma_q(N_1,k_0+k_1)\sigma_q(N_2,\ell_0)\qbin
{N_1-2k_0-k_1}{
k_2}{q}\qbin{N_2-2\ell_0}{
\ell_1}{q}
\qbin{k_0+k_1}{k_0}{q}q^{\Theta_3(\textbf{k};\boldsymbol{\ell})} , 
 \end{align*}}\normalsize
where the numbers $\sigma_q(N_1,k_0+k_1)$\textquotesingle s and $\sigma_q(N_2,\ell_0)$\textquotesingle s are given by \eqref{sigmaE}, and $\Theta_3(\textbf{k};\boldsymbol{\ell})$ is given by \eqref{theta3}. 
\end{thm}
\begin{proof}
It follows immediately from Theorem \ref{Thm5.2}.
\end{proof}
\vspace{-4mm}\begin{example}\label{ex4.3} Using Magma, we see that there are precisely $499$ and $2064151$ distinct non-zero Euclidean self-orthogonal $\mathbb{Z}_{27}\mathbb{Z}_9$-linear codes of block-lengths $(2,2)$ and $(3,3),$ respectively. This agrees with Theorem \ref{Thm5.3}.
\end{example}
In the following corollary, we count all self-orthogonal additive codes of length $N$ over $\mathcal{R}_3.$
\vspace{-1mm}\begin{cor}\label{Corollary3.3}
The number of self-orthogonal additive codes of length $N$ over $\mathcal{R}_3$ is given by
\vspace{-1mm}\small{\begin{align*}
\sum\limits_{k_0=0}^{\floor{\frac{Nrt}{2}}}\sum\limits_{k_1=0}^{\floor{\frac{Nrt}{2}}-k_0}\sum\limits_{k_2=0}^{Nrt-2k_0-k_1}\sum\limits_{\ell_0=0}^{\floor{\frac{Nr(k-t)}{2}}}\sum\limits_{\ell_1=0}^{Nr(k-t)-2\ell_0}&\sigma_p(Nrt,k_0+k_1)\sigma_p\bigl(Nr(k-t),\ell_0\bigr)\qbin
{Nrt-2k_0-k_1}{
k_2}{p}\qbin{Nr(k-t)-2\ell_0}{
\ell_1}{p}
\vspace{-1mm}\\&\times\qbin{k_0+k_1}{k_0}{p}p^{\omega_3(\textbf{k};\boldsymbol{\ell})} , 
 \vspace{-1mm}\end{align*}}\normalsize
where the numbers $\sigma_p(Nrt,k_0+k_1)$\textquotesingle s and $\sigma_p\bigl(Nr(k-t),\ell_0\bigr)$\textquotesingle s are given by \eqref{sigmaE}, and $\omega_3(\textbf{k};\boldsymbol{\ell})$ is given by
\vspace{-1mm}\small{\begin{equation*}\omega_3(\textbf{k};\boldsymbol{\ell})=\bigl(Nrt-3k_0-k_1-k_2\bigr)\bigl(k_0+2\ell_0+k_1+\ell_1\bigr)+k_0\bigl(2Nrk-Nrt-1\bigr)+k_1\bigl(Nr(k-t)-2\ell_0-\ell_1\bigr)+\ell_0\bigl(Nr(k-t)-\frac{3}{2}\ell_0-\ell_1+k_2-\frac{1}{2}\bigr).\vspace{-1mm}\end{equation*}}\normalsize
\end{cor}
\begin{proof}By Remark \ref{KEY}, we see that  the number of self-orthogonal additive codes  of length $N$ over $\mathcal{R}_3$ is equal to the number of distinct Euclidean self-orthogonal $\mathbb{Z}_{p^3}\mathbb{Z}_{p^2}$-linear codes of block-length $\bigl(Nrt,Nr(k-t)\bigr).$ On taking $N_1=Nrt$ and $N_2=Nr(k-t)$ in Theorem  \ref{Thm5.3}, the desired result follows immediately.\end{proof}
\vspace{-4mm}\begin{example}By Corollary \ref{Corollary3.3}, we see that 
there are precisely $499$ and $2064151$ distinct non-zero self-orthogonal additive codes of lengths $2$ and $3$ over $\mathbb{Z}_{27}[y]/\langle y^2-3,9y \rangle,$ respectively. This agrees with Example \ref{ex4.3} and Magma computations.
\end{example}
In the following theorem, we derive a necessary and sufficient condition for the existence of a Euclidean self-dual $\mathtt{R}_3\mathtt{R}_2$-linear code of block-length $(N_1,N_2).$ We also obtain an enumeration formula for $\mathfrak{M}_{3}(N_1,N_2).$
\vspace{-1mm}\begin{thm}\label{lem3.5}
There exists a Euclidean self-dual $\mathtt{R}_3\mathtt{R}_2$-linear code of block-length $(N_1,N_2)$ if and only if $N_1$ is even and $(-1)^{\frac{N_1}{2}}$ is a square in $\mathcal{T}_e(\simeq \mathbb{F}_q).$ Furthermore, when $N_1$ is even and $(-1)^{\frac{N_1}{2}}$ is a square in $\mathcal{T}_e,$ we have 
\vspace{-1mm}\small{\begin{align*}
\mathfrak{M}_{3}(N_1,N_2)=\sum\limits_{k_0=0}^{\frac{N_1}{2}}\sum\limits_{\ell_0=0}^{\floor{\frac{N_2}{2}}}\sigma_q(N_1,\tfrac{N_1}{2})\sigma_q(N_2,\ell_0)
\qbin{N_1/2}{k_0}{q}q^{k_0(\frac{N_1}{2}+N_2-1)+\frac{\ell_0}{2}(N_1-2k_0+\ell_0-1)} , 
 \end{align*}}\normalsize
where the numbers $\sigma_q(N_1,\frac{N_1}{2})$\textquotesingle s and $\sigma_q(N_2,\ell_0)$\textquotesingle s are given by \eqref{sigmaE}. 
\end{thm}
\begin{proof}
It follows from Theorem \ref{Thm5.3}, Lemma \ref{Lem1.1}, and Remarks \ref{Remark1.1} and \ref{Remark2.2}.
\end{proof}
\vspace{-4mm}\begin{example}\label{Eg.4.7} Using Magma, we see that there are precisely $172$ distinct Euclidean self-dual $\mathbb{Z}_{5^3}\mathbb{Z}_{5^2}$-linear codes of block-length $(2,2).$ This agrees with Theorem \ref{lem3.5}.
\end{example}
In the following corollary, we derive a necessary and sufficient condition for the existence of a self-dual  additive code of length $N$ over $\mathcal{R}_3.$ We also count all self-dual additive codes of length $N$ over $\mathcal{R}_3.$ 
\vspace{-1mm}\begin{cor}\label{Corollary3.4}
There exists a self-dual additive code of length $N$ over $\mathcal{R}_3$ if and only if $Nrt$ is even and $(-1)^{\frac{Nrt}{2}}$ is a square in $\mathbb{Z}_p.$ Furthermore, when $Nrt$ is even and $(-1)^{\frac{Nrt}{2}}$ is a square in $\mathbb{Z}_p,$ the number of self-dual additive codes of length $N$ over $\mathcal{R}_3$ is given by
\vspace{-1mm}\small$$\sum\limits_{k_0=0}^{\frac{Nrt}{2}}\sum\limits_{\ell_0=0}^{\floor{\frac{Nr(k-t)}{2}}}\sigma_p(Nrt,\tfrac{Nrt}{2})\sigma_p\bigl(Nr(k-t),\ell_0\bigr)
\qbin{Nrt/2}{k_0}{p}p^{k_0\bigl(\frac{Nrt}{2}+Nr(k-t)-1\bigr)+\frac{\ell_0}{2}(Nrt-2k_0+\ell_0-1)},\vspace{-1mm}$$
\normalsize where the numbers $\sigma_p(Nrt,\frac{Nrt}{2})$\textquotesingle s and $\sigma_p\bigl(Nr(k-t),\ell_0\bigr)$\textquotesingle s are given by \eqref{sigmaE}.
\end{cor}
\begin{proof}By Remark \ref{KEY}, we see that  the number of self-dual additive codes  of length $N$ over $\mathcal{R}_3$ is equal to the number of distinct Euclidean self-dual $\mathbb{Z}_{p^3}\mathbb{Z}_{p^2}$-linear codes of block-length $\bigl(Nrt,Nr(k-t)\bigr).$ Now the desired result follows by taking $N_1=Nrt$ and $N_2=Nr(k-t)$ in Theorem  \ref{lem3.5}.\end{proof}
\vspace{-4mm}\begin{example}By Corollary \ref{Corollary3.4}, we see that 
there are precisely $172$ distinct self-dual additive codes of length $2$ over $\mathbb{Z}_{5^3}[y]/\langle y^2-5,5^2y \rangle.$ This agrees with Example \ref{Eg.4.7} and Magma computations.
\end{example}
\vspace{-4mm}\subsection{The case $e\geq 4$}
\vspace{-1mm}
Throughout this section, let $\mu$ be an integer satisfying $4\leq\mu\leq e.$ Here, we will first provide a recursive method to construct a Euclidean self-orthogonal $\mathtt{R}_\mu\mathtt{R}_{\mu-1}$-linear code of block-length $(N_1,N_2)$ from a Euclidean self-orthogonal $\mathtt{R}_{\mu-2}\mathtt{R}_{\mu-3}$-linear code of the same block-length $(N_1,N_2),$ which gives rise to a recurrence relation between enumeration formulae for these two classes of codes. This method is an extension of the recursive method employed by Yadav and Sharma \cite[Sec. 5]{Yadav} to count Euclidean self-orthogonal and self-dual codes over $\mathtt{R}_e.$

 To outline this recursive construction method,  we will follow the same notations as introduced in Section \ref{prelim3} from now on. Here, we first recall, from Section \ref{prelim3}, that the $\mathtt{R}_{\mu-2}$-module  $\mathtt{R}_{\mu-2}^{N_1}\oplus\mathtt{R}_{\mu-3}^{N_2}$ can be embedded into the $\mathtt{R}_{\mu}$-module $\mathtt{R}_{\mu}^{N_1}\oplus\mathtt{R}_{\mu-1}^{N_2}.$ We next assume, throughout this paper, that for  positive integers $n$ and $s,$ 
$\mathscr{K}_{n,s}$ is the set consisting of all $s$-tuples $(k_0,k_1,\ldots,k_{s-1})$ of non-negative integers satisfying $2k_0+2k_1+\cdots+2k_{s-1-i}+k_{s-1-i+1}+k_{s-1-i+2}+\cdots+k_i\leq n$ for $\floor{\frac{s}{2}}\leq i\leq s-1.$ Further, for an $s$-tuple $\textbf{k}=(k_0,k_1,\ldots,k_{s-1})\in\mathscr{K}_{n,s},$  let us define the number
\vspace{-2mm}\begin{align}
m_i(\textbf{k})&=k_0+k_1+\cdots+k_i~\text{for}~0\leq i\leq s-1.\label{defmi}
\vspace{-1mm}\end{align}
Note that $m_i(\mathbf{k})=0$ if $i<0.$ We see, by Remark  \ref{Remark1.1}, that if there exists a Euclidean self-orthogonal $\mathtt{R}_\mu\mathtt{R}_{\mu-1}$-linear code of block-length $(N_1,N_2)$ and type  $\{k_0,k_1,\ldots,k_{\mu-1};\ell_0,\ell_1,\ldots,\ell_{\mu-2}\},$ then   $\textbf{k}=(k_0,k_1,\ldots,k_{\mu-1})\in\mathscr{K}_{N_1,\mu}$ and  $\boldsymbol{\ell}=(\ell_0,\ell_1,\ldots,\ell_{\mu-2})\in\mathscr{K}_{N_2,\mu-1}.$ Moreover, if there exists a Euclidean self-orthogonal $\mathtt{R}_{\mu-2}\mathtt{R}_{\mu-3}$-linear code of block-length $(N_1,N_2)$ and type $\{k_0+k_1,k_2,\ldots,k_{\mu-2};\ell_0+\ell_1,\ell_2,\ldots,\ell_{\mu-3}\},$ then $\textbf{k}=(k_0,k_1,\ldots,k_{\mu-1})\in\mathscr{K}_{N_1,\mu}$ and $\boldsymbol{\ell}=(\ell_0,\ell_1,\ldots,\ell_{\mu-2})\in\mathscr{K}_{N_2,\mu-1}.$
Accordingly, we assume, throughout this section, that $\textbf{k}=(k_0,k_1,\ldots,k_{\mu-1})\in\mathscr{K}_{N_1,\mu}$ and  $\boldsymbol{\ell}=(\ell_0,\ell_1,\ldots,\ell_{\mu-2})\in\mathscr{K}_{N_2,\mu-1}.$ Note that $m_{\mu-2}(\textbf{k})=k_0+k_1+\cdots+k_{\mu-2}$ and $m_{\mu-3}(\boldsymbol{\ell})=\ell_0+\ell_1+\cdots+\ell_{\mu-3}.$ 

Throughout this section, let $\mathscr{C}_{\mu-2}$ be a Euclidean self-orthogonal $\mathtt{R}_{\mu-2}\mathtt{R}_{\mu-3}$-linear code of block-length $(N_1,N_2)$ and type $\{k_0+k_1,k_2,\ldots,k_{\mu-2};\ell_0+\ell_1,\ell_2,\ldots,\ell_{\mu-3}\}.$ Since $Tor_1(\mathscr{C}_{\mu-2}^{(X)})$
is a $(k_0+k_1)$-dimensional subspace of $\mathcal{T}_e^{N_1}$ over $\mathcal{T}_e$ and $Tor_1(\mathscr{C}_{\mu-2}^{(Y)})$ is an $(\ell_0+\ell_1)$-dimensional subspace of $\mathcal{T}_e^{N_2}$ over $\mathcal{T}_e,$ there are precisely $
\qbin{k_0+k_1}{k_{0}}{q}\qbin{\ell_0+\ell_1}{\ell_0}{q}$ distinct ways of choosing  a $k_0$-dimensional subspace $P_1$ of $Tor_1(\mathscr{C}_{\mu-2}^{(X)})$ and an $\ell_0$-dimensional subspace $P_2$ of $Tor_1(\mathscr{C}_{\mu-2}^{(Y)}).$ Further, for  given choices of $P_1$ and $P_2,$  without any loss of generality, we assume, by applying suitable row and columns operations (if necessary), that the code $\mathscr{C}_{\mu-2}$ has a generator matrix $\mathtt{G}_{\mu-2}$ of the form 
\vspace{-2mm}\begin{equation}\label{Gmu-2}
\mathtt{G}_{\mu-2}=\left[\begin{array}{c|c}
    \mathtt{A}_{\mu-2} & \mathtt{B}_{\mu-2} \\ 
    \mathtt{C}_{\mu-2} & \mathtt{D}_{\mu-2}
\end{array}\right],
\vspace{-2mm}\end{equation}
where \vspace{-2mm}\begin{eqnarray*}\small{[\mathtt{A}_{\mu-2}~|~\mathtt{B}_{\mu-2}]}\hspace{-3mm}&=&\hspace{-4mm}\small{\left[\begin{array}{cccccc|cccccc}
 I_{k_0} & \mathtt{A}^{(\mu-2)}_{0,1} & \mathtt{A}^{(\mu-2)}_{0,2} & \cdots & \mathtt{A}^{(\mu-2)}_{0,\mu-2} & \mathtt{A}^{(\mu-2)}_{0,\mu-1}& 0 & \mathtt{B}^{(\mu-2)}_{0,1} & \mathtt{B}^{(\mu-2)}_{0,2} & \cdots & \mathtt{B}^{(\mu-2)}_{0,\mu-3} & \mathtt{B}^{(\mu-2)}_{0,\mu-2} \\
 0 & I_{k_1} & \mathtt{A}^{(\mu-2)}_{1,2}  & \cdots &  \mathtt{A}^{(\mu-2)}_{1,\mu-2} &  \mathtt{A}^{(\mu-2)}_{1,\mu-1}& 0 & 0 & \mathtt{B}^{(\mu-2)}_{1,2} &  \cdots  & \mathtt{B}^{(\mu-2)}_{1,\mu-3}&  \mathtt{B}^{(\mu-2)}_{1,\mu-2}\\
 0 & 0 & \gamma I_{k_2}  & \cdots & \gamma \mathtt{A}^{(\mu-2)}_{2,\mu-2} & \gamma \mathtt{A}^{(\mu-2)}_{2,\mu-1}& 0 & 0 & 0 &  \cdots & \gamma \mathtt{B}^{(\mu-2)}_{2,\mu-3} & \gamma \mathtt{B}^{(\mu-2)}_{2,\mu-2} \\
 \vdots  & \vdots & \vdots & \ddots & \vdots & \vdots & \vdots &  \vdots & \vdots & \ddots & \vdots &\vdots \\
 0 & 0  & 0 & \cdots & \gamma^{\mu-3}I_{k_{\mu-2}} & \gamma^{\mu-3}\mathtt{A}^{(\mu-2)}_{\mu-2,\mu-1}& 0 & 0 & 0 & \cdots & 0 &0
\end{array}
\right]}\\&=&\hspace{-4mm}\small{\left[\begin{array}{c|c}
\mathtt{A}_0^{(\mu-2)}&\mathtt{B}_0^{(\mu-2)}\\
\mathtt{A}_1^{(\mu-2)}&\mathtt{B}_1^{(\mu-2)}\\
\gamma\mathtt{A}_2^{(\mu-2)}&\gamma\mathtt{B}_2^{(\mu-2)}\\
\vdots&\vdots\\
\gamma^{\mu-4}\mathtt{A}_{\mu-3}^{(\mu-2)}&\gamma^{\mu-4}\mathtt{B}_{\mu-3}^{(\mu-2)}\\
\gamma^{\mu-3}\mathtt{A}_{\mu-2}^{(\mu-2)}&0
\end{array}
\right]=\left[\begin{array}{c}
\hat{R}_1^{(\mu-2)}\\
\hat{R}_2^{(\mu-2)}\\
\gamma \hat{R}_3^{(\mu-2)}\\
\vdots\\
\gamma^{\mu-4}\hat{R}_{\mu-2}^{(\mu-2)}\\
\gamma^{\mu-3}\hat{R}_{\mu-1}^{(\mu-2)}
\end{array}
\right]}\normalsize\vspace{-2mm}\end{eqnarray*}
with $\mathtt{A}_{0,i}^{(\mu-2)}=\mathtt{A}_{0,i}'^{(\mu-2)}+\gamma^{i-1}\mathtt{A}_{0,i}''^{(\mu-2)}$ for $2\leq i\leq \mu-2$  and $\mathtt{B}_{0,j}^{(\mu-2)}=\mathtt{B}_{0,j}'^{(\mu-2)}+\gamma^{j-1}\mathtt{B}_{0,j}''^{(\mu-2)}$ for $2\leq j\leq \mu-3,$ and  
\vspace{-2mm}\begin{eqnarray*}\small{[\mathtt{C}_{\mu-2}~|~\mathtt{D}_{\mu-2}]}\hspace{-3mm}&=&\hspace{-4mm}\small{\left[\begin{array}{cccccc|cccccc}
0 & 0 & \gamma \mathtt{C}^{(\mu-2)}_{0,2} & \cdots & \gamma \mathtt{C}^{(\mu-2)}_{0,\mu-2} & \gamma \mathtt{C}^{(\mu-2)}_{0,\mu-1} & I_{\ell_0} & \mathtt{D}^{(\mu-2)}_{0,1} & \mathtt{D}^{(\mu-2)}_{0,2} & \cdots &  \mathtt{D}^{(\mu-2)}_{0,\mu-3} & \mathtt{D}^{(\mu-2)}_{0,\mu-2} \\
0 & 0 & 0 & \cdots & \gamma \mathtt{C}^{(\mu-2)}_{1,\mu-2} & \gamma \mathtt{C}^{(\mu-2)}_{1,\mu-1} & 0 &  I_{\ell_1} & \mathtt{D}^{(\mu-2)}_{1,2} & \cdots &  \mathtt{D}^{(\mu-2)}_{1,\mu-3}  &  \mathtt{D}^{(\mu-2)}_{1,\mu-2} \\
0 & 0 & 0 & \cdots & \gamma^2 \mathtt{C}^{(\mu-2)}_{2,\mu-2} & \gamma^2 \mathtt{C}^{(\mu-2)}_{2,\mu-1} & 0 &  0 & \gamma I_{\ell_2} & \cdots & \gamma \mathtt{D}^{(\mu-2)}_{2,\mu-3}  & \gamma \mathtt{D}^{(\mu-2)}_{2,\mu-2} \\
\vdots  & \vdots & \vdots & \ddots & \vdots & \vdots & \vdots &  \vdots & \vdots & \ddots & \vdots &\vdots \\
0 & 0 & 0 & \cdots & 0 & \gamma^{\mu-3}\mathtt{C}^{(\mu-2)}_{\mu-3,\mu-1} & 0 & 0 & 0 & \cdots & \gamma^{\mu-4}I_{\ell_{\mu-3}} & \gamma^{\mu-4}\mathtt{D}^{(\mu-2)}_{\mu-3,\mu-2}
\end{array}
\right]}\\&=&\hspace{-4mm}\small{\left[\begin{array}{c|c}
\gamma\mathtt{C}_0^{(\mu-2)}&\mathtt{D}_0^{(\mu-2)}\\
\gamma\mathtt{C}_1^{(\mu-2)}&\mathtt{D}_1^{(\mu-2)}\\
\gamma^2\mathtt{C}_2^{(\mu-2)}&\gamma\mathtt{D}_2^{(\mu-2)}\\
\vdots&\vdots\\
\gamma^{\mu-3}\mathtt{C}_{\mu-3}^{(\mu-2)}&\gamma^{\mu-4}\mathtt{D}_{\mu-3}^{(\mu-2)}
\end{array}
\right]=\left[\begin{array}{c}
\hat{S}_1^{(\mu-2)}\\
\hat{S}_2^{(\mu-2)}\\
\gamma \hat{S}_3^{(\mu-2)}\\
\vdots\\
\gamma^{\mu-4}\hat{S}_{\mu-2}^{(\mu-2)}\\
\end{array}
\right]}\normalsize\vspace{-2mm}\end{eqnarray*}
with $\gamma\mathtt{C}_{0,i}^{(\mu-2)}=\gamma\mathtt{C}_{0,i}'^{(\mu-2)}+\gamma^{i-1}\mathtt{C}_{0,i}''^{(\mu-2)}$ for $2\leq i\leq \mu-2$ and $\mathtt{D}_{0,j}^{(\mu-2)}=\mathtt{D}_{0,j}'^{(\mu-2)}+\gamma^{j-1}\mathtt{D}_{0,j}''^{(\mu-2)}$ for $2\leq j\leq \mu-3.$ 
Here,  $I_{k_i}$ is the $k_i \times k_i$ identity matrix over $\mathtt{R}_{\mu-2}$ for $0 \leq i \leq \mu-2$ and  $I_{\ell_j}$ is the $\ell_j \times \ell_j$ identity matrix over  $\mathtt{R}_{\mu-3}$ for $0 \leq j \leq \mu-3.$ Further, we have $\mathtt{A}_{0,1}^{(\mu-2)}\in M_{k_0\times k_1}(\mathcal{T}_e),$   $\mathtt{D}_{0,1}^{(\mu-2)}\in M_{\ell_0\times \ell_1}(\mathcal{T}_e),$  $\mathtt{A}^{(\mu-2)}_{0,\mu-1}\in M_{k_0\times (N_1-m_{\mu-2}(\textbf{k}))}(\mathtt{R}_{\mu-2}),$   $\mathtt{B}^{(\mu-2)}_{0,\mu-2}\in M_{k_0\times (N_2-m_{\mu-3}(\boldsymbol{\ell}))}(\mathtt{R}_{\mu-3}),$   $\mathtt{D}^{(\mu-2)}_{0,\mu-2}\in M_{\ell_0\times (N_2-m_{\mu-3}(\boldsymbol{\ell}))}(\mathtt{R}_{\mu-3})$ and   $\gamma \mathtt{C}^{(\mu-2)}_{0,\mu-1}\in M_{\ell_0\times (N_1-m_{\mu-2}(\textbf{k}))}(\mathtt{R}_{\mu-2}).$ Also, the matrix  $\gamma^{i-1}\mathtt{A}^{(\mu-2)}_{i,\mu-1}\in M_{k_i\times (N_1-m_{\mu-2}(\textbf{k}))}(\mathtt{R}_{\mu-2})$ for $1\leq i\leq \mu-2,$  and
the matrices $\gamma^{j-1}\mathtt{B}^{(\mu-2)}_{j,\mu-2}\in M_{k_j\times (N_2-m_{\mu-3}(\boldsymbol{\ell}))}(\mathtt{R}_{\mu-3}),$ $\gamma^{j}\mathtt{C}^{(\mu-2)}_{j,\mu-1}\in M_{\ell_j\times (N_1-m_{\mu-2}(\textbf{k}))}(\mathtt{R}_{\mu-2})$ and  $\gamma^{j-1} \mathtt{D}^{(\mu-2)}_{j,\mu-2}\in M_{\ell_j\times (N_2-m_{\mu-3}(\boldsymbol{\ell}))}(\mathtt{R}_{\mu-3})$ for $1\leq j\leq \mu-3.$ Moreover, the matrix $\mathtt{A}'^{(\mu-2)}_{0,i}\in M_{k_0\times k_i}(\mathtt{R}_{\mu-2})$ is considered modulo $\gamma^{i-1}$ for $2\leq i\leq \mu-2,$  the matrix $\mathtt{A}^{(\mu-2)}_{i,j}\in M_{k_i\times k_j}(\mathtt{R}_{\mu-2})$ is considered modulo $\gamma^{j-i}$ for $1\leq i<j\leq \mu-2,$ the matrices $\mathtt{B}'^{(\mu-2)}_{0,j}\in M_{k_0\times \ell_j}(\mathtt{R}_{\mu-3})$ and $\mathtt{D}'^{(\mu-2)}_{0,j}\in M_{\ell_0\times \ell_j}(\mathtt{R}_{\mu-3})$ are  considered modulo $\gamma^{j-1}$ for $2\leq j\leq \mu-3,$ 
the matrices $\mathtt{B}^{(\mu-2)}_{i,j}\in M_{k_i\times \ell_j}(\mathtt{R}_{\mu-3})$ and $\mathtt{D}^{(\mu-2)}_{i,j}\in M_{\ell_i\times \ell_j}(\mathtt{R}_{\mu-3})$ are considered modulo $\gamma^{j-i}$ for $1\leq i<j\leq \mu-3,$ the matrix $\gamma \mathtt{C}'^{(\mu-2)}_{0,i}\in M_{\ell_0\times k_i}(\mathtt{R}_{\mu-2})$ is considered modulo $\gamma^{i-1}$ for $2\leq i\leq \mu-2,$ and the matrix $\mathtt{C}^{(\mu-2)}_{i,j}\in M_{\ell_i\times k_j}(\mathtt{R}_{\mu-2})$ is considered modulo $\gamma^{j-i-1}$ for $1\leq i\leq \mu-4$ and $i+2\leq j\leq \mu-2.$ Besides this, the matrix $\mathtt{B}_{0,1}^{(\mu-2)}\in M_{k_0\times \ell_1}(\mathcal{T}_e),$ the matrices $\mathtt{A}_{0,i}''^{(\mu-2)}\in M_{k_0\times k_i}(\mathcal{T}_e)$  and $\mathtt{C}_{0,i}''^{(\mu-2)}\in M_{\ell_0\times k_i}(\mathcal{T}_e)$ for $2\leq i\leq \mu-2,$ and  the matrices $\mathtt{B}_{0,j}''^{(\mu-2)}\in M_{k_0\times \ell_{j}}(\mathcal{T}_e)$ and $\mathtt{D}_{0,j}''^{(\mu-2)}\in M_{\ell_0\times \ell_j}(\mathcal{T}_e)$ for $2\leq j\leq \mu-3$ are chosen arbitrarily. That is, we have  $\mathtt{A}_0^{(\mu-2)}\in M_{k_0\times N_1}(\mathtt{R}_{\mu-2}),$ $\mathtt{B}_0^{(\mu-2)}\in M_{k_0\times N_2}(\mathtt{R}_{\mu-3}),$  $\gamma\mathtt{C}_0^{(\mu-2)}\in M_{\ell_0\times N_1}(\mathtt{R}_{\mu-2}),$   $\mathtt{D}_0^{(\mu-2)}\in M_{\ell_0\times N_2}(\mathtt{R}_{\mu-3})$
$\gamma^{i-1}\mathtt{A}_i^{(\mu-2)}\in M_{k_i\times N_1}(\mathtt{R}_{\mu-2})$ for $1\leq i\leq \mu-2,$  and the matrices
$\gamma^{j-1}\mathtt{B}_j^{(\mu-2)}\in M_{k_j\times N_2}(\mathtt{R}_{\mu-3}),$ $\gamma^j\mathtt{C}_j^{(\mu-2)}\in M_{\ell_j\times N_1}(\mathtt{R}_{\mu-2})$ and $\gamma^{j-1}\mathtt{D}_j^{(\mu-2)}\in M_{\ell_j\times N_2}(\mathtt{R}_{\mu-3})$  for $1\leq j\leq \mu-3.$  
  Moreover, the columns of the matrix $\begin{bmatrix}
 \mathtt{A}_{\mu-2}\\
 \mathtt{C}_{\mu-2}
\end{bmatrix}$ are grouped into blocks of sizes $k_0,k_1,\ldots,k_{\mu-2},N_1-m_{\mu-2}(\textbf{k}),$ and the columns of the matrix $\begin{bmatrix}
\mathtt{B}_{\mu-2}\\
\mathtt{D}_{\mu-2}
\end{bmatrix}$ are grouped into blocks of sizes $\ell_0,\ell_1,\ldots,\ell_{\mu-3},N_2-m_{\mu-3}(\boldsymbol{\ell}).$  Furthermore, since the code $\mathscr{C}_{\mu-2}$ is Euclidean self-orthogonal, we see, by Lemma \ref{Lem1.1} and  working as in Remark 2.2 of Yadav and Sharma \cite{Yadav},  that the matrices $\overline{\mathtt{A}_{0,\mu-1}^{(\mu-2)}}$ and $\overline{\mathtt{D}_{0,\mu-2}^{(\mu-2)}}$ are of full row-ranks.  

Now, let us define the following matrices in $\mathfrak{M}$ as
\vspace{-1mm}\small{\begin{align*}
R_1^{(\mu)}&~=~ \hat{R}_1^{(\mu-2)}+[0~0~\cdots~0~\gamma^{\mu-2}E_{0,\mu-1}+\gamma^{\mu-1}F_{0,\mu-1}~|~0~0~\cdots~0~\gamma^{\mu-3}M_{0,\mu-2}+\gamma^{\mu-2}L_{0,\mu-2}],\\
\gamma^{i}R_{i+1}^{(\mu)}&~=~ \gamma^i\hat{R}_{i+1}^{(\mu-2)}+[0~0~\cdots~0~\gamma^{\mu-1}F_{i,\mu-1}~|~0~0~\cdots~0~\gamma^{\mu-2}L_{i,\mu-2}]~\text{for}~1\leq i\leq \mu-2,\\
\gamma^{\mu-1}R_\mu^{(\mu)}&~=~[0~0~\cdots~0~\gamma^{\mu-1}\mathtt{A}_{\mu-1,\mu-1}~|~0~ 0~ \cdots ~ 0],\\
S_1^{(\mu)}&~=~\hat{S}_1^{(\mu-2)}+[0~0~\cdots~0~\gamma^{\mu-2}K_{0,\mu-1}+\gamma^{\mu-1}J_{0,\mu-1}~|~0~0~\cdots~0~\gamma^{\mu-3}G_{0,\mu-2}+\gamma^{\mu-2}H_{0,\mu-2}],\\
\gamma^jS_{j+1}^{(\mu)}&~ =~\gamma^j\hat{S}_{j+1}^{(\mu-2)}+[0~0~\cdots~0~\gamma^{\mu-1}J_{j,\mu-1}~|~0~0~\cdots~0~\gamma^{\mu-2}H_{j,\mu-1}]~\text{for}~1\leq j\leq \mu-3, \text{ and}\\
\gamma^{\mu-2}S_{\mu-1}^{(\mu)}&~=~[0~0~\cdots~0~\gamma^{\mu-1} J_{\mu-2,\mu-1}~|~0~0~\cdots~0~\gamma^{\mu-2}\mathtt{D}_{\mu-2,\mu-2}],
\end{align*}}\normalsize
where  the matrices $E_{0,\mu-1}\in M_{k_0\times (N_1-m_{\mu-2}(\textbf{k}))}(\mathcal{T}_e),$ $M_{0,\mu-2}\in M_{k_0\times (N_2-m_{\mu-3}(\boldsymbol{\ell}))}(\mathcal{T}_e),$ $K_{0,\mu-1}\in M_{\ell_0\times (N_1-m_{\mu-2}(\textbf{k}))}\\(\mathcal{T}_e),$ $G_{0,\mu-2}\in M_{\ell_0\times (N_2-m_{\mu-3}(\boldsymbol{\ell}))}(\mathcal{T}_e),$ the matrices $F_{i,\mu-1}\in M_{k_i\times (N_1-m_{\mu-2}(\textbf{k}))}(\mathcal{T}_e),$ $L_{i,\mu-2}\in M_{k_i\times (N_2-m_{\mu-3}(\boldsymbol{\ell}))}(\mathcal{T}_e)$ and $J_{i,\mu-1}\in M_{\ell_i\times (N_1-m_{\mu-2}(\textbf{k}))}(\mathcal{T}_e)$ for $0\leq i\leq \mu-2,$ the matrix  $H_{j,\mu-2}\in M_{\ell_j\times (N_2-m_{\mu-3}(\boldsymbol{\ell}))}(\mathcal{T}_e)$ for $0\leq j\leq \mu-3,$ and both the matrices $\mathtt{A}_{\mu-1,\mu-1}\in M_{k_{\mu-1}\times (N_1-m_{\mu-2}(\textbf{k}))}(\mathcal{T}_e)$ and $\mathtt{D}_{\mu-2,\mu-2}\in M_{\ell_{\mu-2}\times (N_2-m_{\mu-3}(\boldsymbol{\ell}))}(\mathcal{T}_e)$ are of full row-ranks.
We next assume, throughout this section, that  $\mathscr{C}_{\mu}$ is an $\mathtt{R}_{\mu}\mathtt{R}_{\mu-1}$-linear code of block-length $(N_1,N_2)$ with a generator matrix  \vspace{-2mm}\begin{equation}\label{Gmu}\mathtt{G}_\mu=\left[\begin{array}{c}
R_1^{(\mu)}\\
\gamma R_2^{(\mu)}\\
\gamma^2 R_3^{(\mu)}\\
\vdots\\
\gamma^{\mu-2} R_{\mu-1}^{(\mu)}\\
\gamma^{\mu-1} R_\mu^{(\mu)}\\
S_1^{(\mu)}\\
\gamma S_2^{(\mu)}\\
\vdots\\
\gamma^{\mu-3}S_{\mu-2}^{(\mu)}\\
\gamma^{\mu-2} S_{\mu-1}^{(\mu)}
\end{array}
\right].\end{equation}
Note that the code $\mathscr{C}_{\mu}$ satisfies $Tor_1(\mathscr{C}_\mu^{(X)})=P_1,$ $Tor_1(\mathscr{C}_\mu^{(Y)})=P_2,$ 
\vspace{-2mm}\begin{align} 
Tor_{i+1}(\mathscr{C}_\mu^{(X)})&=Tor_i(\mathscr{C}_{\mu-2}^{(X)})~\text{for}~1\leq i\leq \mu-2~\text{and}\label{eqn3.32}\\
Tor_{j+1}(\mathscr{C}_\mu^{(Y)})&=Tor_j(\mathscr{C}_{\mu-2}^{(Y)})~\text{for}~1\leq j\leq \mu-3.\label{eqn3.33}
\end{align}
Further, the code $\mathscr{C}_{\mu}$  is of the type $\{k_0,k_1,\ldots,k_{\mu-1};\ell_0,\ell_1,\ldots,\ell_{\mu-2}\}$ and is called a lift of the code $\mathscr{C}_{\mu-2}$ from now on. We next observe the following:
\vspace{-1mm}\begin{lemma} \label{LLemma}If the code $\mathscr{C}_\mu$ is Euclidean self-orthogonal, then the code $\mathscr{C}_{\mu-2}$ is also Euclidean self-orthogonal.  
\end{lemma}
\begin{proof} As the code $\mathscr{C}_{\mu}$ is Euclidean self-orthogonal, we have $\mathtt{G}_{\mu}\diamond\mathtt{G}_{\mu}^T=0 $ in $\mathtt{R}_{\mu},$ which implies that $\mathtt{G}_{\mu-2}\diamond\mathtt{G}_{\mu-2}^T=0$ in $\mathtt{R}_{\mu-2}.$ From this, it follows that the code $\mathscr{C}_{\mu-2}$ is Euclidean self-orthogonal.  
\end{proof}
In the following proposition, we show  that each Euclidean self-orthogonal $\mathtt{R}_{\mu-2}\mathtt{R}_{\mu-3}$-linear code of block-length $(N_1,N_2)$ and type $\{k_0+k_1,k_2,\ldots,k_{\mu-2};\ell_0+\ell_1,\ell_2,\ldots,\ell_{\mu-3}\}$ can be lifted to a Euclidean self-orthogonal $\mathtt{R}_{\mu}\mathtt{R}_{\mu-1}$-linear code of the same block-length $(N_1,N_2)$ and type $\{k_0,k_1,\ldots,k_{\mu-1};\ell_0,\ell_1,\ldots,\ell_{\mu-2}\}.$ With the help of this observation and using Lemma \ref{LLemma}, we derive  a recurrence relation between enumeration formulae for Euclidean self-orthogonal $\mathtt{R}_{\mu}\mathtt{R}_{\mu-1}$-linear codes of block-length $(N_1,N_2)$ and type $\{k_0,k_1,\ldots,k_{\mu-1};\ell_0,\ell_1,\ldots,\ell_{\mu-2}\}$ and  Euclidean self-orthogonal $\mathtt{R}_{\mu-2}\mathtt{R}_{\mu-3}$-linear codes of the same block-length $(N_1,N_2)$ and type $\{k_0+k_1,k_2,\ldots,k_{\mu-2};\ell_0+\ell_1,\ell_2,\ldots,\ell_{\mu-3}\}.$   The proof of the following proposition also gives rise to a recursive method to construct Euclidean self-orthogonal $\mathtt{R}_{\mu}\mathtt{R}_{\mu-1}$-linear codes of block-length $(N_1,N_2)$ and  type $\{k_0,k_1,\ldots,k_{\mu-1};\ell_0,\ell_1,\ldots,\ell_{\mu-2}\}$ from Euclidean self-orthogonal $\mathtt{R}_{\mu-2}\mathtt{R}_{\mu-3}$-linear codes of the same block-length $(N_1,N_2)$ and type $\{k_0+k_1,k_2,\ldots,k_{\mu-2};\ell_0+\ell_1,\ell_2,\ldots,\ell_{\mu-3}\}.$  
 \vspace{-1mm}\begin{prop}\label{Thm6.1}
Let $\textbf{k}=(k_0,k_1,\ldots,k_{\mu-1})\in\mathscr{K}_{N_1,\mu}$ and  $\boldsymbol{\ell}=(\ell_0,\ell_1,\ldots,\ell_{\mu-2})\in\mathscr{K}_{N_2,\mu-1}$ be fixed. Let $\mathscr{C}_{\mu-2}$ be a Euclidean self-orthogonal $\mathtt{R}_{\mu-2}\mathtt{R}_{\mu-3}$-linear code of block-length $(N_1,N_2)$ and type $\{k_0+k_1,k_2,\ldots,k_{\mu-2};\ell_0+\ell_1,\ell_2,\ldots,\ell_{\mu-3}\}$ with a generator matrix $\mathtt{G}_{\mu-2}$ of the form \eqref{Gmu-2}. The following hold.   
\vspace{-1mm}\begin{itemize}
\item[(a)] There exists a Euclidean self-orthogonal $\mathtt{R}_\mu\mathtt{R}_{\mu-1}$-linear code $\mathscr{C}_\mu$ of block-length $(N_1,N_2)$ and type $\{k_0,k_1,\ldots,\\k_{\mu-1};\ell_0,\ell_1,\ldots,\ell_{\mu-2}\}$ satisfying $Tor_1(\mathscr{C}_\mu^{(X)})=P_1,$ $Tor_1(\mathscr{C}_\mu^{(Y)})=P_2,$  $Tor_{i+1}(\mathscr{C}_\mu^{(X)})=Tor_i(\mathscr{C}_{\mu-2}^{(X)})$ for $1\leq i\leq \mu-2$ and $Tor_{j+1}(\mathscr{C}_\mu^{(Y)})=Tor_j(\mathscr{C}_{\mu-2}^{(Y)})$ for $1\leq j\leq \mu-3.$  
\vspace{-1mm}\item[(b)] Furthermore, each Euclidean self-orthogonal $\mathtt{R}_{\mu-2}\mathtt{R}_{\mu-3}$-linear code $\mathscr{C}_{\mu-2}$ of block-length $(N_1,N_2)$ and type $\{k_0+k_1,k_2,\ldots,k_{\mu-2};\ell_0+\ell_1,\ell_2,\ldots,\ell_{\mu-3}\}$ can be lifted to  precisely 
\vspace{-2mm}\small{$$\qbin{N_1-m_{\mu-2}(\textbf{k})-k_0}{k_{\mu-1}}{q}\qbin{N_2-m_{\mu-3}(\boldsymbol{\ell})-\ell_0}{\ell_{\mu-2}}{q}\qbin{k_0+k_1}{k_{0}}{q}\qbin{\ell_0+\ell_1}{\ell_0}{q}q^{\Theta_\mu(\mathbf{k};\boldsymbol{\ell})}\vspace{-2mm}$$}\normalsize distinct Euclidean self-orthogonal $\mathtt{R}_\mu\mathtt{R}_{\mu-1}$-linear codes $\mathscr{C}_{\mu}$ of block-length $(N_1,N_2)$ and type $\{k_0,k_1,\ldots,k_{\mu-1};\\\ell_0,\ell_1,\ldots,\ell_{\mu-2}\}$ satisfying  $Tor_{i+1}(\mathscr{C}_\mu^{(X)})=Tor_i(\mathscr{C}_{\mu-2}^{(X)})$ for $1\leq i\leq \mu-2$ and $Tor_{j+1}(\mathscr{C}_\mu^{(Y)})=Tor_j(\mathscr{C}_{\mu-2}^{(Y)})$ for $1\leq j\leq \mu-3,$  where 
\vspace{-2mm}\small{\begin{eqnarray}\Theta_\mu(\mathbf{k};\boldsymbol{\ell})&=&k_0\bigl(2N_1-2m_{\mu-2}(\textbf{k})-m_1(\textbf{k})-1\bigr)+\ell_0\bigl(2N_2-2m_{\mu-3}(\boldsymbol{\ell})-m_1(\boldsymbol{\ell})-k_1-1\bigr)-k_{\mu-1}\ell_{\mu-2}\nonumber\\
&&+\bigl(m_{\mu-2}(\textbf{k})+m_{\mu-3}(\boldsymbol{\ell})-k_0-\ell_0\bigr)\bigl(N_1+N_2-m_{\mu-2}(\textbf{k})-m_{\mu-3}(\boldsymbol{\ell})\bigr)+2k_0\bigl(N_2-m_{\mu-3}(\boldsymbol{\ell})\bigr)\nonumber\\
&&+\bigl(k_0+2\ell_0\bigr)\bigl(N_1-m_{\mu-2}(\textbf{k})-k_0\bigr)-(m_{\mu-2}(\textbf{k})+m_{\mu-3}(\boldsymbol{\ell}))(k_{\mu-1}+\ell_{\mu-2}).\label{thetamu}\end{eqnarray}}\normalsize
\vspace{-10mm}\item[(c)] We have \vspace{-2mm}\small{$$\mathfrak{N}_\mu(N_1,N_2; \mathbf{k},\boldsymbol{\ell})=\mathfrak{N}_{\mu-2}(N_1,N_2;\mathbf{k}^{(1)},\boldsymbol{\ell}^{(1)})\qbin{N_1-m_{\mu-2}(\textbf{k})-k_0}{k_{\mu-1}}{q}\qbin{N_2-m_{\mu-3}(\boldsymbol{\ell})-\ell_0}{\ell_{\mu-2}}{q}\qbin{k_0+k_1}{k_{0}}{q}\qbin{\ell_0+\ell_1}{\ell_0}{q}q^{\Theta_\mu(\mathbf{k};\boldsymbol{\ell})},\vspace{-2mm}$$}\normalsize  where $\mathbf{k}^{(1)}=(k_0+k_1,k_2,\ldots,k_{\mu-2})$ and $\boldsymbol{\ell}^{(1)}=(\ell_0+\ell_1,\ell_2,\ldots,\ell_{\mu-3}).$  
\end{itemize}
\end{prop}
\begin{proof}To prove the result,  we note that the code $\mathscr{C}_{\mu-2}$ is Euclidean self-orthogonal. This,   by Lemma \ref{Lem6.1}, gives
\vspace{-1mm}\small\begin{align}
\hat{R}_1^{(\mu-2)}\diamond (\hat{R}_1^{(\mu-2)})^T&\equiv\gamma^{\mu-2}\mathcal{P}_5+\gamma^{\mu-1}\mathcal{P}_6\Mod{\gamma^{\mu}},\label{eq2.40}\\
\hat{S}_1^{(\mu-2)}\diamond (\hat{S}_1^{(\mu-2)})^T&\equiv\gamma^{\mu-2}\mathcal{P}_7+\gamma^{\mu-1}\mathcal{P}_8\Mod{\gamma^{\mu}},\label{eq2.41}\\
\hat{R}_1^{(\mu-2)}\diamond (\hat{S}_1^{(\mu-2)})^T&\equiv\gamma^{\mu-2}\mathcal{Q}_8+\gamma^{\mu-1}\mathcal{Q}_{9}\Mod{\gamma^{\mu}},\label{eq2.42}\\
(\hat{S}_1^{(\mu-2)})\diamond(\gamma^{j-1}\hat{S}_{j+1}^{(\mu-2)})^T &\equiv\gamma^{\mu-2}\mathcal{Q}_{4,j}\Mod{\gamma^{\mu-1}}~\text{for}~1\leq j\leq \mu-3,\label{eq2.43}\\
(\hat{R}_1^{(\mu-2)})\diamond(\gamma^{j-1}\hat{S}_{j+1}^{(\mu-2)})^T &\equiv\gamma^{\mu-2}\mathcal{Q}_{5,j}\Mod{\gamma^{\mu-1}}~\text{for}~1\leq j\leq \mu-3, \label{eq2.44}\\
(\hat{S}_1^{(\mu-2)})\diamond(\gamma^{i-1}\hat{R}_{i+1}^{(\mu-2)})^T  &\equiv\gamma^{\mu-2}\mathcal{Q}_{6,i}\Mod{\gamma^{\mu-1}}~\text{for}~1\leq i\leq \mu-2, \label{eq2.45}\\
(\hat{R}_1^{(\mu-2)})\diamond(\gamma^{i-1}\hat{R}_{i+1}^{(\mu-2)})^T &\equiv\gamma^{\mu-2}\mathcal{Q}_{7,i}\Mod{\gamma^{\mu-1}}~\text{for}~1\leq i\leq \mu-2,\label{eq2.46}\\
(\gamma^{i-1}\hat{S}_{i+1}^{(\mu-2)})\diamond (\gamma^{j-1}\hat{S}_{j+1}^{(\mu-2)})^T&\equiv0\Mod{\gamma^{\mu-2}}~\text{for}~1\leq i\leq j\leq \mu-3,\label{eqn3.38}\\
(\gamma^{j-1}\hat{S}_{j+1}^{(\mu-2)})\diamond (\gamma^{i-1}\hat{R}_{i+1}^{(\mu-2)})^T&\equiv0\Mod{\gamma^{\mu-2}}~\text{for}~1\leq i\leq \mu-2~\text{and}~1\leq j\leq \mu-3,~\text{and}\label{eqn3.39}\\
(\gamma^{i-1}\hat{R}_{i+1}^{(\mu-2)})\diamond (\gamma^{j-1}\hat{R}_{j+1}^{(\mu-2)})^T&\equiv0\Mod{\gamma^{\mu-2}}~\text{for}~1\leq i\leq j\leq \mu-2\label{eqn3.40}
\vspace{-2mm}\end{align}\normalsize
for some $\mathcal{P}_5,\mathcal{P}_6\in Sym_{k_0}(\mathcal{T}_e),$ $\mathcal{P}_7,\mathcal{P}_8\in Sym_{\ell_0}(\mathcal{T}_e),$ $\mathcal{Q}_8,\mathcal{Q}_9\in M_{k_0\times\ell_0}(\mathcal{T}_e),$ $\mathcal{Q}_{4,j}\in M_{\ell_0\times\ell_j}(\mathcal{T}_e),$ $\mathcal{Q}_{5,j}\in M_{k_0\times\ell_j}(\mathcal{T}_e)$ for $1\leq j\leq \mu-3,$ and $\mathcal{Q}_{6,i}\in M_{\ell_0\times k_i}(\mathcal{T}_e),$  $\mathcal{Q}_{7,i}\in M_{k_0\times k_i}(\mathcal{T}_e)$ for $1\leq i\leq \mu-2.$

Now let $\mathscr{C}_{\mu}$ be an $\mathtt{R}_{\mu}\mathtt{R}_{\mu-1}$-linear code of block-length $(N_1,N_2)$ and type $\{k_0,k_1,k_2,\ldots,k_{\mu-2}, k_{\mu-1};\ell_0,\ell_1,\ell_2,\ldots,\\\ell_{\mu-3}, \ell_{\mu-2}\}$ with a generator matrix $\mathtt{G}_{\mu}$ of the form \eqref{Gmu}. One can easily see that the code $\mathscr{C}_{\mu}$ is Euclidean self-orthogonal if and only if $\mathtt{G}_{\mu}\diamond \mathtt{G}_{\mu}^T=0,$ which, by Lemma \ref{Lem6.1} and congruences \eqref{eq2.40}-\eqref{eqn3.40},  holds if and only if
there exist full row-rank matrices $\mathtt{A}_{\mu-1,\mu-1}\in M_{k_{\mu-1}\times (N_1-m_{\mu-2}(\textbf{k}))}(\mathcal{T}_e)$ and $\mathtt{D}_{\mu-2,\mu-2}\in M_{\ell_{\mu-2}\times (N_2-m_{\mu-3}(\boldsymbol{\ell}))}(\mathcal{T}_e),$ and the matrices $E_{0,\mu-1} \in M_{k_0\times (N_1-m_{\mu-2}(\textbf{k}))}(\mathcal{T}_e),$ $M_{0,\mu-2}\in M_{k_0\times (N_2-m_{\mu-3}(\boldsymbol{\ell}))}(\mathcal{T}_e),$ $K_{0,\mu-1}\in M_{\ell_0\times (N_1-m_{\mu-2}(\textbf{k}))}(\mathcal{T}_e),$ $G_{0,\mu-2}\in M_{\ell_0\times (N_2-m_{\mu-3}(\boldsymbol{\ell}))}(\mathcal{T}_e),$ $F_{i,\mu-1}\in M_{k_i\times (N_1-m_{\mu-2}(\textbf{k}))}(\mathcal{T}_e),$ $L_{i,\mu-2}\in M_{k_i\times (N_2-m_{\mu-3}(\boldsymbol{\ell}))}(\mathcal{T}_e)$ and  $J_{i,\mu-1}\in \\M_{\ell_i\times (N_1-m_{\mu-2}(\textbf{k}))}(\mathcal{T}_e)$ for $0\leq i\leq \mu-2,$ and $H_{j,\mu-2}\in M_{\ell_j\times (N_2-m_{\mu-3}(\boldsymbol{\ell}))}(\mathcal{T}_e)$ for $0\leq j\leq \mu-3,$  satisfying the following system of matrix congruences:
\vspace{-1mm}\small\begin{align}
&\mathcal{P}_5+\mathtt{A}^{(\mu-2)}_{0,\mu-1}E_{0,\mu-1}^T+E_{0,\mu-1}(\mathtt{A}^{(\mu-2)}_{0,\mu-1})^T
+\mathtt{B}^{(\mu-2)}_{0,\mu-2}M_{0,\mu-2}^T+M_{0,\mu-2}(\mathtt{B}^{(\mu-2)}_{0,\mu-2})^T+\gamma\bigl(\mathcal{P}_6\nonumber\\&+\mathtt{A}^{(\mu-2)}_{0,\mu-1}F_{0,\mu-1}^T
+F_{0,\mu-1}(\mathtt{A}^{(\mu-2)}_{0,\mu-1})^T+\mathtt{B}^{(\mu-2)}_{0,\mu-2}L_{0,\mu-2}^T+L_{0,\mu-2}(\mathtt{B}^{(\mu-2)}_{0,\mu-2})^T+\gamma^{\mu-4}M_{0,\mu-2}M_{0,\mu-2}^T\bigr) ~\equiv0\Mod{\gamma^{2}},\label{eqn3.3}\\
&\mathcal{P}_7+\mathtt{D}^{(\mu-2)}_{0,\mu-2}G_{0,\mu-2}^T+G_{0,\mu-2}(\mathtt{D}^{(\mu-2)}_{0,\mu-2})^T+\gamma\bigl(\mathcal{P}_8
\nonumber\\&+\mathtt{C}^{(\mu-2)}_{0,\mu-1}K_{0,\mu-1}^T+K_{0,\mu-1}(\mathtt{C}^{(\mu-2)}_{0,\mu-1})^T+\mathtt{D}^{(\mu-2)}_{0,\mu-2}H_{0,\mu-2}^T
+H_{0,\mu-2}(\mathtt{D}^{(\mu-2)}_{0,\mu-2})^T+\gamma^{\mu-4}G_{0,\mu-2}G_{0,\mu-2}^T\bigr) \equiv0\Mod{\gamma^{2}},\label{eqn3.4}
\\
&\mathcal{Q}_8+\mathtt{A}^{(\mu-2)}_{0,\mu-1}K_{0,\mu-1}^T+\mathtt{B}^{(\mu-2)}_{0,\mu-2}G_{0,\mu-2}^T
+E_{0,\mu-1}(\mathtt{C}^{(\mu-2)}_{0,\mu-1})^T
+M_{0,\mu-2}(\mathtt{D}^{(\mu-2)}_{0,\mu-2})^T+\gamma\bigl(\mathcal{Q}_9\nonumber\\&+\mathtt{A}^{(\mu-2)}_{0,\mu-1}J_{0,\mu-1}^T
+\mathtt{B}^{(\mu-2)}_{0,\mu-2}H_{0,\mu-2}^T+F_{0,\mu-1}(\mathtt{C}^{(\mu-2)}_{0,\mu-1})^T+L_{0,\mu-2}(\mathtt{D}^{(\mu-2)}_{0,\mu-2})^T+\gamma^{\mu-4}M_{0,\mu-2}G_{0,\mu-2}^T\bigr)~\equiv0\Mod{\gamma^{2}},\label{eqn3.5}\\
&\mathcal{Q}_{4,j}+\mathtt{D}^{(\mu-2)}_{0,\mu-2}H_{j,\mu-2}^T+\gamma^{j-1}G_{0,\mu-2}(\mathtt{D}^{(\mu-2)}_{j,\mu-2})^T\equiv0\Mod{\gamma}~\text{for}~1\leq j\leq \mu-3,\label{eqn3.6}\\
&\mathcal{Q}_{5,j}+\mathtt{A}^{(\mu-2)}_{0,\mu-1}J_{j,\mu-1}^T+\mathtt{B}^{(\mu-2)}_{0,\mu-2}H_{j,\mu-2}^T+\gamma^{j-1}G_{0,\mu-2}(\mathtt{D}^{(\mu-2)}_{j,\mu-2})^T\equiv0\Mod{\gamma}~\text{for}~1\leq j\leq \mu-3,\label{eqn3.7}\\
&\mathcal{Q}_{6,i}+\mathtt{D}^{(\mu-2)}_{0,\mu-2}L_{i,\mu-2}^T+\gamma^{i-1}\bigl(K_{0,\mu-1}(\mathtt{A}^{(\mu-2)}_{i,\mu-1})^T+G_{0,\mu-2}(\mathtt{B}^{(\mu-2)}_{i,\mu-2})^T\bigr)\equiv0\Mod{\gamma}~\text{for}~1\leq i\leq \mu-2,\label{eqn3.8}\\
&\mathcal{Q}_{7,i}+\mathtt{B}^{(\mu-2)}_{0,\mu-2}L_{i,\mu-2}^T+\mathtt{A}^{(\mu-2)}_{0,\mu-1}F_{i,\mu-1}^T+\gamma^{i-1}\bigl(E_{0,\mu-1}(\mathtt{A}^{(\mu-2)}_{i,\mu-1})^T+M_{0,\mu-2}(\mathtt{B}^{(\mu-2)}_{i,\mu-2})^T\bigr)\equiv0\Mod{\gamma}~\text{for}~1\leq i\leq \mu-2,\label{eqn3.9}\\
&\mathtt{A}^{(\mu-2)}_{0,\mu-1}J_{\mu-2,\mu-1}^T+\mathtt{B}^{(\mu-2)}_{0,\mu-2}\mathtt{D}_{\mu-2,\mu-2}^T\equiv0\Mod{\gamma},\label{eqn3.11}\\
&\mathtt{A}^{(\mu-2)}_{0,\mu-1}\mathtt{A}_{\mu-1,\mu-1}^T\equiv0\Mod{\gamma},\label{eqn3.10}\\
&\mathtt{D}^{(\mu-2)}_{0,\mu-2}\mathtt{D}_{\mu-2,\mu-2}^T\equiv0\Mod{\gamma},\label{eqn3.12}\\
&(\gamma^{i}\hat{S}_{i+1}^{(\mu-2)})\diamond (\gamma^{j}\hat{S}_{j+1}^{(\mu-2)})^T\equiv0\Mod{\gamma^{\mu}}~\text{for}~1\leq i\leq j\leq \mu-3,\label{eqn3.13}\\
&(\gamma^{j}\hat{S}_{j+1}^{(\mu-2)})\diamond (\gamma^{i}\hat{R}_{i+1}^{(\mu-2)})^T\equiv0\Mod{\gamma^{\mu}}~\text{for}~1\leq i\leq \mu-2~\text{and}~1\leq j\leq \mu-3,~\text{and}\label{eqn3.14}\\
&(\gamma^{i}\hat{R}_{i+1}^{(\mu-2)})\diamond (\gamma^{j}\hat{R}_{j+1}^{(\mu-2)})^T\equiv0\Mod{\gamma^{\mu}}~\text{for}~1\leq i\leq j\leq \mu-2.\label{eqn3.15}
\vspace{-1mm}\end{align}\normalsize

Note that the congruences \eqref{eqn3.38}-\eqref{eqn3.40} are equivalent to the congruences \eqref{eqn3.13}-\eqref{eqn3.15}, respectively.
 From this and by Lemma \ref{LLemma}, we see that to prove this result, we need to establish the existence of full row-rank matrices $\mathtt{A}_{\mu-1,\mu-1}$ and $\mathtt{D}_{\mu-2,\mu-2}$ and the matrices $E_{0,\mu-1},$ $M_{0,\mu-2},$  $K_{0,\mu-1},$ $G_{0,\mu-2}$ and  $F_{i,\mu-1},$  $L_{i,\mu-2}$ and  $J_{i,\mu-1}$  for $0\leq i\leq \mu-2,$ and $H_{j,\mu-2}$ for $0\leq j\leq \mu-3$  satisfying \eqref{eqn3.3}-\eqref{eqn3.12}, and then count their choices.

Since $Tor_1(\mathscr{C}_{\mu-2}^{(X)})$
is a $(k_0+k_1)$-dimensional subspace of $\mathcal{T}_e^{N_1}$ over $\mathcal{T}_e$ and $Tor_1(\mathscr{C}_{\mu-2}^{(Y)})$ is an $(\ell_0+\ell_1)$-dimensional subspace of $\mathcal{T}_e^{N_2}$ over $\mathcal{T}_e,$ there are precisely $
\qbin{k_0+k_1}{k_{0}}{q}\qbin{\ell_0+\ell_1}{\ell_0}{q}$ distinct ways of choosing  a $k_0$-dimensional subspace $P_1=Tor_1(\mathscr{C}_{\mu}^{(X)})$ of $Tor_1(\mathscr{C}_{\mu-2}^{(X)})$ and an $\ell_0$-dimensional subspace $P_2=Tor_1(\mathscr{C}_{\mu}^{(Y)})$ of $Tor_1(\mathscr{C}_{\mu-2}^{(Y)}).$ 
We will now establish the existence of the full row-rank matrix $\mathtt{A}_{\mu-1,\mu-1}\in M_{k_{\mu-1}\times (N_1-m_{\mu-2}(\textbf{k}))}(\mathcal{T}_e)$  satisfying \eqref{eqn3.10} and the full row-rank matrix $\mathtt{D}_{\mu-2,\mu-2}\in M_{\ell_{\mu-2}\times (N_2-m_{\mu-3}(\boldsymbol{\ell}))}(\mathcal{T}_e)$ satisfying \eqref{eqn3.12}, and we will also count their choices. To do this, we note that the code $\mathscr{C}_{\mu-2}$ is Euclidean self-orthogonal. This, by Lemma \ref{Lem1.1}, implies that $Tor_{\mu-2}(\mathscr{C}_{\mu-2}^{(X)})\subseteq Tor_{1}(\mathscr{C}_{\mu-2}^{(X)})^{\perp_E},$ which, by \eqref{eqn3.32}, further implies that $Tor_{\mu-1}(\mathscr{C}_{\mu}^{(X)})\subseteq Tor_{2}(\mathscr{C}_{\mu}^{(X)})^{\perp_E}.$ Since $Tor_{1}(\mathscr{C}_{\mu}^{(X)})\subseteq Tor_{2}(\mathscr{C}_{\mu}^{(X)}),$ we have $Tor_{\mu-1}(\mathscr{C}_{\mu}^{(X)})\subseteq Tor_{1}(\mathscr{C}_{\mu}^{(X)})^{\perp_E}.$ From this, it follows that the number of choices for the  full row-rank matrix $\mathtt{A}_{\mu-1,\mu-1}\in M_{k_{\mu-1}\times (N_1-m_{\mu-2}(\textbf{k}))}(\mathcal{T}_e)$ satisfying \eqref{eqn3.10} is equal to the number of choices for the torsion code  $Tor_\mu(\mathscr{C}_\mu^{(X)})$ satisfying  $Tor_{\mu-1}(\mathscr{C}_\mu^{(X)})\subseteq Tor_{\mu}(\mathscr{C}_\mu^{(X)})\subseteq Tor_{1}(\mathscr{C}_\mu^{(X)})^{\perp_E}.$ One can easily see that such a torsion code $Tor_\mu(\mathscr{C}_\mu^{(X)})$ can be chosen in $\qbin{N_1-m_{\mu-2}(\textbf{k})-k_0}{k_{\mu-1}}{q}$ distinct ways. Further, working similarly as above and by \eqref{eqn3.33},  we see that the number of choices for  the full row-rank matrix $\mathtt{D}_{\mu-2,\mu-2}\in M_{\ell_{\mu-2}\times (N_2-m_{\mu-3}(\boldsymbol{\ell}))}(\mathcal{T}_e)$ satisfying \eqref{eqn3.12} is equal to the number of choices for the torsion code $Tor_{\mu-1}(\mathscr{C}_\mu^{(Y)})$ satisfying $Tor_{\mu-2}(\mathscr{C}_\mu^{(Y)})\subseteq Tor_{\mu-1}(\mathscr{C}_\mu^{(Y)})\subseteq Tor_{1}(\mathscr{C}_\mu^{(Y)})^{\perp_E}.$ Here, one can easily see that such a torsion code $Tor_{\mu-1}(\mathscr{C}_\mu^{(Y)})$ can be chosen in  $\qbin{N_2-m_{\mu-3}(\boldsymbol{\ell})-\ell_0}{\ell_{\mu-2}}{q}$ distinct ways. 

Since the matrix $\overline{\mathtt{A}^{(\mu-2)}_{0,\mu-1}}$ is of full row-rank, we see,  by Lemma \ref{Matrixlemma}(b), that there are precisely $q^{k_0\big(N_1-m_{\mu-2}(\textbf{k})-k_0\big)}$ distinct choices for the matrix $J_{\mu-2,\mu-1}$ satisfying \eqref{eqn3.11}.

Now, for given choices of the matrices $\mathtt{A}_{\mu-1,\mu-1},$ $\mathtt{D}_{\mu-2,\mu-2}$ and $J_{\mu-2,\mu-1},$ we will count all possible  choices for the matrices $E_{0,\mu-1},$ $M_{0,\mu-2},$ $K_{0,\mu-1},$ $G_{0,\mu-2},$ the matrices $F_{i,\mu-1},$ $L_{i,\mu-2}$ for $0\leq i\leq \mu-2$ and the matrices $J_{j,\mu-1},$ $H_{j,\mu-2}$ for $0\leq j\leq \mu-3,$ satisfying the congruences \eqref{eqn3.3}-\eqref{eqn3.9}. 

Towards this,  we will first count  the choices for the matrices $L_{0,\mu-2},$ $M_{0,\mu-2},$ $E_{0,\mu-1}$ and $F_{0,\mu-1}$ satisfying \eqref{eqn3.3}. To do this, we first choose the matrices $L_{0,\mu-2}$ and $M_{0,\mu-2}$ arbitrarily, which can be chosen in $q^{2k_0\big(N_2-m_{\mu-3}(\boldsymbol{\ell})\big)}$ distinct ways. Now, since the matrix $\overline{\mathtt{A}^{(\mu-2)}_{0,\mu-1}}$ is of full row-rank, we see, by Lemma \ref{Matrixlemma}(a), that there are precisely $q^{\frac{k_0}{2}\bigl(2N_1-2m_{\mu-2}(\textbf{k})-k_0-1\bigr)}$ distinct choices for the matrix $E_{0,\mu-1}$ satisfying 
\vspace{-1mm}$$\mathcal{P}_5+\mathtt{A}^{(\mu-2)}_{0,\mu-1}E_{0,\mu-1}^T+E_{0,\mu-1}(\mathtt{A}^{(\mu-2)}_{0,\mu-1})^T
+\mathtt{B}^{(\mu-2)}_{0,\mu-2}M_{0,\mu-2}^T+M_{0,\mu-2}(\mathtt{B}^{(\mu-2)}_{0,\mu-2})^T\equiv0\Mod{\gamma}.\vspace{-1mm}$$
From this, we get 
\vspace{-1mm}$$\gamma^{\mu-2}\bigl(\mathcal{P}_5+\mathtt{A}^{(\mu-2)}_{0,\mu-1}E_{0,\mu-1}^T+E_{0,\mu-1}(\mathtt{A}^{(\mu-2)}_{0,\mu-1})^T
+\mathtt{B}^{(\mu-2)}_{0,\mu-2}M_{0,\mu-2}^T+M_{0,\mu-2}(\mathtt{B}^{(\mu-2)}_{0,\mu-2})^T\bigr)\equiv\gamma^{\mu-1}\mathcal{P}_{9}\Mod{\gamma}\vspace{-1mm}$$
for some $\mathcal{P}_{9}\in Sym_{k_0}(\mathcal{T}_e).$ Now, on substituting this into the congruence \eqref{eqn3.3} and by applying Lemma \ref{Matrixlemma}(a) again, we see that there are precisely $q^{\frac{k_0}{2}\bigl(2N_1-2m_{\mu-2}(\textbf{k})-k_0-1\bigr)}$ distinct choices for the matrix $F_{0,\mu-1}$ satisfying \eqref{eqn3.3}.

Next, we proceed to count choices for the matrices $G_{0,\mu-2},$ $K_{0,\mu-1},$ $H_{0,\mu-2}$ and $J_{0,\mu-1}$ satisfying the congruences \eqref{eqn3.4} and \eqref{eqn3.5}. Since the matrix $\overline{\mathtt{D}^{(\mu-2)}_{0,\mu-2}}$ is of full row-rank, we see, by Lemma \ref{Matrixlemma}(a), that there are precisely $q^{\frac{\ell_0}{2}\bigl(2N_2-2m_{\mu-3}(\boldsymbol{\ell})-\ell_0-1\bigr)}$ distinct choices for the matrix $G_{0,\mu-2}$ satisfying \vspace{-1mm}$$\mathcal{P}_7+\mathtt{D}^{(\mu-2)}_{0,\mu-2}G_{0,\mu-2}^T+G_{0,\mu-2}(\mathtt{D}^{(\mu-2)}_{0,\mu-2})^T\equiv0\Mod{\gamma}.\vspace{-1mm}$$ 
From this,  we get  \vspace{-1mm}$$\gamma^{\mu-2}\bigl(\mathcal{P}_7+\mathtt{D}^{(\mu-2)}_{0,\mu-2}G_{0,\mu-2}^T+G_{0,\mu-2}(\mathtt{D}^{(\mu-2)}_{0,\mu-2})^T\bigr)\equiv\gamma^{\mu-1}\mathcal{P}_{10}\Mod{\gamma^\mu}\vspace{-1mm}$$ for some $\mathcal{P}_{10}\in Sym_{\ell_{0}}(\mathcal{T}_e).$ Now, on substituting this into the congruence \eqref{eqn3.4}, we get 
\vspace{-1mm}\small{\begin{align}\label{eqn3.50}
\mathcal{P}_{10}+\mathcal{P}_8+C_{0,\mu-1}K_{0,\mu-1}^T+K_{0,\mu-1}(\mathtt{C}^{(\mu-2)}_{0,\mu-1})^T+\mathtt{D}^{(\mu-2)}_{0,\mu-2}H_{0,\mu-2}^T+H_{0,\mu-2}(\mathtt{D}^{(\mu-2)}_{0,\mu-2})^T+\gamma^{\mu-4}G_{0,\mu-2}G_{0,\mu-2}^T\equiv0\Mod{\gamma}.
\vspace{-1mm}\end{align}}\normalsize
Further, to count the  choices for the matrix $K_{0,\mu-1},$ we see that  the congruence \eqref{eqn3.5} gives \vspace{-1mm}\begin{align}\label{K0mu-1}
\mathcal{Q}_8+\mathtt{A}^{(\mu-2)}_{0,\mu-1}K_{0,\mu-1}^T+\mathtt{B}^{(\mu-2)}_{0,\mu-2}G_{0,\mu-2}^T+E_{0,\mu-1}(\mathtt{C}^{(\mu-2)}_{0,\mu-1})^T+M_{0,\mu-2}(\mathtt{D}^{(\mu-2)}_{0,\mu-2})^T\equiv0\Mod{\gamma}.
\vspace{-1mm}\end{align} Using the fact that  the matrix $\overline{\mathtt{A}^{(\mu-2)}_{0,\mu-1}}$ is of full row-rank and by applying Lemma \ref{Matrixlemma}(b), we see that there are precisely $q^{\ell_0\bigl(N_1-m_{\mu-2}(\textbf{k})-k_0\bigr)}$ distinct choices for the matrix $K_{0,\mu-1}$ satisfying \eqref{K0mu-1}.
Further, for a given choice of the matrix $K_{0,\mu-1},$ we see, by Lemma \ref{Matrixlemma}(a), that there are precisely $q^{\frac{\ell_0}{2}\bigl(2N_2-2m_{\mu-3}(\boldsymbol{\ell})-\ell_0-1\bigr)}$ distinct choices for the matrix $H_{0,\mu-2}$ satisfying \eqref{eqn3.50}. To count the  choices for the matrix  $J_{0,\mu-1}$ satisfying \eqref{eqn3.5}, we see that  the congruence \eqref{K0mu-1} gives
\vspace{-1mm}$$\gamma^{\mu-2}\bigl(\mathcal{Q}_8+\mathtt{A}^{(\mu-2)}_{0,\mu-1}K_{0,\mu-1}^T+\mathtt{B}^{(\mu-2)}_{0,\mu-2}G_{0,\mu-2}^T+E_{0,\mu-1}(\mathtt{C}^{(\mu-2)}_{0,\mu-1})^T+M_{0,\mu-2}(\mathtt{D}^{(\mu-2)}_{0,\mu-2})^T\bigr)\equiv\gamma^{\mu-1}\mathcal{Q}_{10}\Mod{\gamma^{\mu}}\vspace{-1mm}$$
for some $\mathcal{Q}_{10}\in M_{k_0\times\ell_0}(\mathcal{T}_e).$ Now, on substituting this into the congruence \eqref{eqn3.5} and by applying Lemma \ref{Matrixlemma}(b) again, we see that the matrix $J_{0,\mu-1}$ has precisely $q^{\ell_0\bigl(N_1-m_{\mu-2}(\textbf{k})-k_0\bigr)}$ distinct choices. 


Finally, we will count the choices for the remaining matrices $H_{j,\mu-2},$ $J_{j,\mu-1}$ for $1\leq j\leq \mu-3$ and $L_{i,\mu-2},$ $F_{i,\mu-1}$ for $1\leq i\leq \mu-2$ satisfying the congruences \eqref{eqn3.6}-\eqref{eqn3.9} for given choices of the matrices $\mathtt{A}_{\mu-1,\mu-1},$ $\mathtt{D}_{\mu-2,\mu-2},$ $J_{\mu-2,\mu-1},$ $L_{0,\mu-2},$ $M_{0,\mu-2},$ $E_{0,\mu-1},$ $F_{0,\mu-1},$ $G_{0,\mu-2},$ $K_{0,\mu-1},$ $H_{0,\mu-2}$ and $J_{0,\mu-1}.$ To do this, we see, using the fact that the matrix $\overline{\mathtt{D}^{(\mu-2)}_{0,\mu-2}}$ is of full row-rank and by Lemma \ref{Matrixlemma}(b),   that there are precisely $q^{\ell_j\bigl(N_2-m_{\mu-3}(\boldsymbol{\ell})-\ell_0\bigr)}$ distinct choices for the matrix $H_{j,\mu-2}$ satisfying \eqref{eqn3.6} for $1\leq j\leq \mu-3.$  Further, for a given choice of the matrix $H_{j,\mu-2},$ we see, using the fact that the matrix $\overline{\mathtt{A}^{(\mu-2)}_{0,\mu-1}}$ is of full row-rank and by  Lemma \ref{Matrixlemma}(b) again, that there are precisely  $q^{\ell_j\bigl(N_1-m_{\mu-2}(\textbf{k})-k_0\bigr)}$ distinct choices for the matrix $J_{j,\mu-1}$ satisfying \eqref{eqn3.7} for $1\leq j\leq \mu-3.$
Further, for $1\leq i\leq \mu-2,$ we see, using the fact that the matrix $\overline{\mathtt{D}^{(\mu-2)}_{0,\mu-2}}$ is of full row-rank and   by Lemma \ref{Matrixlemma}(b), that the matrix $L_{i,\mu-2}$ satisfying \eqref{eqn3.8} has precisely $q^{k_i\bigl(N_2-m_{\mu-3}(\boldsymbol{\ell})-\ell_0\bigr)}$ distinct choices.   Now, for a given choice of the matrix $L_{i,\mu-2},$ we note, using the fact that  the matrix $\overline{\mathtt{A}^{(\mu-2)}_{0,\mu-1}}$ is of full row-rank and by Lemma \ref{Matrixlemma}(b) again, that the matrix $F_{i,\mu-1}$  satisfying \eqref{eqn3.9} has precisely $q^{k_i\bigl(N_1-m_{\mu-2}(\textbf{k})-k_0\bigr)}$ distinct choices.  

From the above discussion, we see that the unknown matrices $\mathtt{A}_{\mu-1,\mu-1},$ $\mathtt{D}_{\mu-2,\mu-2},$ $E_{0,\mu-1},$ $M_{0,\mu-2},$  $K_{0,\mu-1},$ $G_{0,\mu-2},$   $F_{i,\mu-1},$  $L_{i,\mu-2}$ and  $J_{i,\mu-1}$  for $0\leq i\leq \mu-2,$ and $H_{j,\mu-2}$ for $0\leq j\leq \mu-3,$ satisfying \eqref{eqn3.3}-\eqref{eqn3.12} have precisely   \vspace{-2mm}$$\qbin{N_1-m_{\mu-2}(\textbf{k})-k_0}{k_{\mu-1}}{q}\qbin{N_2-m_{\mu-3}(\boldsymbol{\ell})-\ell_0}{\ell_{\mu-2}}{q}q^{\Theta'(\textbf{k};\boldsymbol{\ell})}\vspace{-2mm}$$ distinct choices,
where
\vspace{-1mm}\small{\begin{align}\label{theta'}
\Theta'(\textbf{k};\boldsymbol{\ell})=&k_0\bigl(2N_1-2m_{\mu-2}(\textbf{k})-k_0-1\bigr)+\ell_0\bigl(2N_2-2m_{\mu-3}(\boldsymbol{\ell})-\ell_0-1\bigr)+\bigl(k_0+2\ell_0\bigr)\bigl(N_1-m_{\mu-2}(\textbf{k})-k_0\bigr)\nonumber\\&+\bigl(N_1+N_2-m_{\mu-2}(\textbf{k})-m_{\mu-3}(\boldsymbol{\ell})-k_0-\ell_0\bigr)\bigl(m_{\mu-2}(\textbf{k})+m_{\mu-3}(\boldsymbol{\ell})-k_0-\ell_0\bigr)+2k_0\bigl(N_2-m_{\mu-3}(\boldsymbol{\ell})\bigr).
\vspace{-1mm}\end{align}}\normalsize Further, one can easily observe that there are precisely $q^{\sum\limits_{j=0}^{\mu-3}\ell_j\ell_{\mu-2}+\sum\limits_{i=0}^{\mu-2}(k_i\ell_{\mu-2}+k_ik_{\mu-1}+\ell_ik_{\mu-1})}$ distinct choices of the aforementioned unknown matrices, which give rise to the same Euclidean self-orthogonal $\mathtt{R}_{\mu}\mathtt{R}_{\mu-1}$-linear code $\mathscr{C}_{\mu}$ of block-length $(N_1,N_2)$ and type $\{k_0, k_1, \ldots, k_{\mu-1}; \ell_0, \ell_1, \ldots, \ell_{\mu-2}\}$ satisfying $Tor_1(\mathscr{C}_\mu^{(X)})=P_1,$ $Tor_1(\mathscr{C}_\mu^{(Y)})=P_2,$  $Tor_{i+1}(\mathscr{C}_\mu^{(X)})=Tor_i(\mathscr{C}_{\mu-2}^{(X)})$ for $1\leq i\leq \mu-2$ and $Tor_{j+1}(\mathscr{C}_\mu^{(Y)})=Tor_j(\mathscr{C}_{\mu-2}^{(Y)})$ for $1\leq j\leq \mu-3.$  
Thus  each Euclidean self-orthogonal $\mathtt{R}_{\mu-2}\mathtt{R}_{\mu-3}$-linear code $\mathscr{C}_{\mu-2}$ with a generator matrix $\mathtt{G}_{\mu-2}$ of the form \eqref{Gmu-2} can be lifted to precisely
\vspace{-4mm}\begin{eqnarray*}
&\qbin{N_1-m_{\mu-2}(\textbf{k})-k_0}{k_{\mu-1}}{q}\qbin{N_2-m_{\mu-3}(\boldsymbol{\ell})-\ell_0}{\ell_{\mu-2}}{q}q^{\Theta'(\textbf{k},\boldsymbol{\ell})-\sum\limits_{j=0}^{\mu-3}\ell_j\ell_{\mu-2}-\sum\limits_{i=0}^{\mu-2}(k_i\ell_{\mu-2}+k_ik_{\mu-1}+\ell_ik_{\mu-1})}
\vspace{-4mm}\end{eqnarray*}
distinct Euclidean self-orthogonal $\mathtt{R}_{\mu}\mathtt{R}_{\mu-1}$-linear codes $\mathscr{C}_\mu$ satisfying $Tor_1(\mathscr{C}_\mu^{(X)})=P_1,$ $Tor_1(\mathscr{C}_\mu^{(Y)})=P_2,$  \eqref{eqn3.32} and \eqref{eqn3.33},  where  $\Theta'(\textbf{k},\boldsymbol{\ell})$ is given by \eqref{theta'}.
Further, we note that for given choices of $P_1$ and $P_2,$ the code $\mathscr{C}_{\mu-2}$ has precisely $q^{k_0(\sum\limits_{w=2}^{\mu-2}k_w+\sum\limits_{s=1}^{\mu-3}\ell_s)+\ell_0(\sum\limits_{i=2}^{\mu-2}k_i+\sum\limits_{j=2}^{\mu-3}\ell_j)}$ distinct generator matrices $\mathtt{G}_{\mu-2}$ of the form \eqref{Gmu-2} corresponding to the choices of the matrices $\mathtt{A}_{0,i}''^{(\mu-2)}\in M_{k_0\times k_i}(\mathcal{T}_e)$ and $\mathtt{C}_{0,i}''^{(\mu-2)}\in M_{\ell_0\times k_i}(\mathcal{T}_e)$ for $2\leq i\leq \mu-2,$ the matrices $\mathtt{B}_{0,j}''^{(\mu-2)}\in M_{k_0\times \ell_{j}}(\mathcal{T}_e)$ and $\mathtt{D}_{0,j}''^{(\mu-2)}\in M_{\ell_0\times \ell_j}(\mathcal{T}_e)$ for $2\leq j\leq \mu-3,$ and the matrix $\mathtt{B}_{0,1}^{(\mu-2)}\in M_{k_0\times \ell_{1}}(\mathcal{T}_e).$ From this, it follows that each Euclidean self-orthogonal $\mathtt{R}_{\mu-2}\mathtt{R}_{\mu-3}$-linear code $\mathscr{C}_{\mu-2}$ of block-length $(N_1,N_2)$ and type $\{k_0+k_1,k_2,\ldots,k_{\mu-2};\ell_0+\ell_1,\ell_2,\ldots,\ell_{\mu-3}\}$ gives rise to precisely 
\vspace{-1mm}\begin{align*}
&\qbin{N_1-m_{\mu-2}(\textbf{k})-k_0}{k_{\mu-1}}{q}\qbin{N_2-m_{\mu-3}(\boldsymbol{\ell})-\ell_0}{\ell_{\mu-2}}{q}\qbin{k_0+k_1}{k_{0}}{q}\qbin{\ell_0+\ell_1}{\ell_0}{q}
\vspace{-1mm}\\&\times q^{\Theta'(\mathbf{k};\boldsymbol{\ell})+k_0(\sum\limits_{w=2}^{\mu-2}k_w+\sum\limits_{s=1}^{\mu-3}\ell_s)+\ell_0(\sum\limits_{w'=2}^{\mu-2}k_{w'}+\sum\limits_{s'=2}^{\mu-3}\ell_{s'})-\sum\limits_{j=0}^{\mu-3}\ell_j\ell_{\mu-2}-\sum\limits_{i=0}^{\mu-2}(k_i\ell_{\mu-2}+k_ik_{\mu-1}+\ell_ik_{\mu-1})}\vspace{-1mm}\\
=&\qbin{N_1-m_{\mu-2}(\textbf{k})-k_0}{k_{\mu-1}}{q}\qbin{N_2-m_{\mu-3}(\boldsymbol{\ell})-\ell_0}{\ell_{\mu-2}}{q}\qbin{k_0+k_1}{k_{0}}{q}\qbin{\ell_0+\ell_1}{\ell_0}{q}q^{\Theta_{\mu}(\mathbf{k};\boldsymbol{\ell})}
\vspace{-1mm}\end{align*}
distinct Euclidean self-orthogonal $\mathtt{R}_\mu\mathtt{R}_{\mu-1}$-linear codes $\mathscr{C}_{\mu}$ of block-length $(N_1,N_2)$ and type $\{k_0,k_1,\ldots,k_{\mu-1};\ell_0,\\\ell_1,\ldots,\ell_{\mu-2}\}$ satisfying  $Tor_{i+1}(\mathscr{C}_\mu^{(X)})=Tor_i(\mathscr{C}_{\mu-2}^{(X)})$ for $1\leq i\leq \mu-2$ and $Tor_{j+1}(\mathscr{C}_\mu^{(Y)})=Tor_j(\mathscr{C}_{\mu-2}^{(Y)})$ for $1\leq j\leq \mu-3,$ where $\Theta_\mu(\mathbf{k};\boldsymbol{\ell})$ and $\Theta'(\mathbf{k};\boldsymbol{\ell})$ are given by \eqref{thetamu} and \eqref{theta'}, respectively. From this, the desired result follows. 
\end{proof}

Now, with the help of the recurrence relation derived in Proposition \ref{Thm6.1}(c), we will enumerate all Euclidean self-orthogonal $\mathtt{R}_e\mathtt{R}_{e-1}$-linear codes of block-length $(N_1,N_2).$ Towards this, let us first fix some notations. For integers $n\geq 1$ and $s\geq 4,$ and an $s$-tuple $\textbf{k}=(k_0,k_1,\ldots,k_{s-1})\in\mathscr{K}_{n,s},$ let us define the $(s-2i)$-tuple $\textbf{k}^{(i)}$ as
\vspace{-1mm}\begin{align}
\textbf{k}^{(i)}=\bigl(m_i(\textbf{k}),k_{i+1},k_{i+2},\ldots,k_{s-1-i}\bigr)~\text{ for }~0\leq i \leq \left\lfloor \tfrac{s-2}{2} \right\rfloor. 
\vspace{-1mm}\end{align}
Further, for an $e$-tuple $\textbf{k}=(k_0,k_1,\ldots,k_{e-1})\in\mathscr{K}_{N_1,e}$ and an $(e-1)$-tuple $\boldsymbol{\ell}=(\ell_0,\ell_1,\ldots,\ell_{e-2})\in\mathscr{K}_{N_2,e-1},$ we see that $\textbf{k}^{(i)}\in\mathscr{K}_{N_1,e-2i}$ for $0\leq i\leq \floor{\frac{e-2}{2}}$ and $\boldsymbol{\ell}^{(j)}\in\mathscr{K}_{N_2,e-1-2j}$ for $0\leq j\leq \floor{\frac{e-3}{2}}.$ Throughout this paper, for an $e$-tuple $\textbf{k}\in\mathscr{K}_{N_1,e}$ and an $(e-1)$-tuple $\boldsymbol{\ell}\in\mathscr{K}_{N_2,e-1},$ let us define 
\vspace{-2mm}\begin{align}\label{defdeltaE}
\Delta_e(\textbf{k},\boldsymbol{\ell})=\left\{\begin{array}{ll}
\sum\limits_{i=0}^{\frac{e-5}{2}}\Theta_{5+2i}(\textbf{k}^{(\frac{e-5}{2}-i)};\boldsymbol{\ell}^{(\frac{e-5}{2}-i)})&~\text{if}~e~\text{is odd};\\
\sum\limits_{i=0}^{\frac{e-4}{2}}\Theta_{4+2i}(\textbf{k}^{(\frac{e-4}{2}-i)};\boldsymbol{\ell}^{(\frac{e-4}{2}-i)})&~\text{if}~e~\text{is even}\end{array}\right.
\end{align}\vspace{-1mm} and
\vspace{-2mm}\begin{align}\label{defskl}
s_{e}(\textbf{k},\boldsymbol{\ell})=\left\{\begin{array}{ll} \Theta_3(\textbf{k}^{(\frac{e-3}{2})};\boldsymbol{\ell}^{(\frac{e-3}{2})})&~\text{if}~e~\text{is odd}; \vspace{3mm}\\
\Theta_2\bigl(\textbf{k}^{(\frac{e-2}{2})};m_{\frac{e-2}{2}}(\boldsymbol{\ell})\bigr) &~\text{if}~e~\text{is even,}
\vspace{-2mm}\end{array} \right.
\end{align}\normalsize
where $\Theta_\mu(\mathbf{k};\boldsymbol{\ell})$ is given by \eqref{thetamu} for  $ 4 \leq  \mu \leq e,$ and $\Theta_2\bigl(\textbf{k}^{(\frac{e-2}{2})};m_{\frac{e-2}{2}}(\boldsymbol{\ell})\bigr)$ and $\Theta_3(\textbf{k}^{(\frac{e-3}{2})};\boldsymbol{\ell}^{(\frac{e-3}{2})})$   are given by \eqref{theta2} and \eqref{theta3}, respectively. Also, for a prime power $Q_1,$ let us define the number
\vspace{-1mm}\small\begin{align}\label{Beab}
\mathscr{B}_{e,Q_1}^{(N_1,N_2)}(\textbf{k},\boldsymbol{\ell})=& \prod\limits_{i=1}^{\floor{\frac{e-1}{2}}}
\qbin{m_i(\textbf{k})}{k_i}{Q_1}
\prod\limits_{j=1}^{\floor{\frac{e-2}{2}}}
\qbin{m_j(\boldsymbol{\ell})}{\ell_j}{Q_1}
\prod\limits_{w=\floor{\frac{e-1}{2}}}^{e-2}
\qbin{N_1-m_w(\textbf{k})-m_{e-2-w}(\textbf{k})}{k_{w+1}}{Q_1}\nonumber \\& \times \prod\limits_{s=\floor{\frac{e-2}{2}}}^{e-3}
\qbin{N_2-m_s(\boldsymbol{\ell})-m_{e-3-s}(\boldsymbol{\ell})}{\ell_{s+1}}{Q_1}.
\vspace{-1mm}\end{align}\normalsize

In the following theorem, we obtain the number $\mathfrak{N}_e(N_1,N_2;\textbf{k};\boldsymbol{\ell}),$ where $\textbf{k}=(k_0,k_1,\ldots,k_{e-1})$ and $\boldsymbol{\ell}=(\ell_0,\ell_1,\ldots,\ell_{e-2}).$
\vspace{-1mm}\begin{thm}\label{Thm6.2}Let $e \geq 4$ be a fixed integer. 
There exists a Euclidean self-orthogonal $\mathtt{R}_e\mathtt{R}_{e-1}$-linear code of block length $(N_1,N_2)$ and type $\{k_0,k_1,\ldots,k_{e-1};\ell_0,\ell_1,\ldots,\ell_{e-2}\}$ if and only if the $e$-tuple $\textbf{k}=(k_0,k_1,\ldots,k_{e-1})\in\mathscr{K}_{N_1,e}$ and the $(e-1)$-tuple $\boldsymbol{\ell}=(\ell_0,\ell_1,\ldots,\ell_{e-2})\in\mathscr{K}_{N_2,e-1}.$ Furthermore, for $\textbf{k}\in\mathscr{K}_{N_1,e}$ and $\boldsymbol{\ell}\in\mathscr{K}_{N_2,e-1},$ we have \vspace{-1mm}$$\mathfrak{N}_e(N_1,N_2;\textbf{k};\boldsymbol{\ell})=\sigma_q\bigl(N_1,m_{\floor{\frac{e-1}{2}}}(\textbf{k})\bigr)\sigma_q\bigl(N_2,m_{\floor{\frac{e-2}{2}}}(\boldsymbol{\ell})\bigr)\mathscr{B}_{e,q}^{(N_1,N_2)}(\textbf{k},\boldsymbol{\ell}) q^{\Delta_{e}(\textbf{k},\boldsymbol{\ell})+s_{e}(\textbf{k},\boldsymbol{\ell})},\vspace{-1mm}$$ where the numbers $\sigma_q\bigl(N_1,m_{\floor{\frac{e-1}{2}}}(\textbf{k})\bigr)$\textquotesingle s  and $\sigma_q\bigl(N_2,m_{\floor{\frac{e-2}{2}}}(\boldsymbol{\ell})\bigr)$\textquotesingle s are given by \eqref{sigmaE}, and $\Delta_{e}(\textbf{k},\boldsymbol{\ell})$\textquotesingle s, $s_{e}(\textbf{k},\boldsymbol{\ell})$\textquotesingle s and $\mathscr{B}_{e,q}^{(N_1,N_2)}(\textbf{k},\boldsymbol{\ell})$\textquotesingle s are given by \eqref{defdeltaE}, \eqref{defskl} and \eqref{Beab}, respectively.
\end{thm}
\begin{proof}By repeatedly applying the recurrence relation derived in Proposition \ref{Thm6.1}(c)
   and using Remark \ref{Remark1.1} and Theorems \ref{Thm4.2} and \ref{Thm5.2}, we get the desired result.
\end{proof} 
In the following theorem, we obtain an enumeration formula for $\mathfrak{N}_e(N_1,N_2).$ 
\vspace{-1mm}\begin{thm}\label{Thm6.3}
For $e \geq 4,$ we have
\vspace{-1mm}\begin{align*}
\mathfrak{N}_{e}(N_1,N_2)=&\sum\limits_{\textbf{k}\in \mathscr{K}_{N_1,e}}~\sum\limits_{\boldsymbol{\ell}\in \mathscr{K}_{N_2,e-1}}\sigma_q\bigl(N_1,m_{\floor{\frac{e-1}{2}}}(\textbf{k})\bigr)\sigma_q\bigl(N_2,m_{\floor{\frac{e-2}{2}}}(\boldsymbol{\ell})\bigr)\mathscr{B}_{e,q}^{(N_1,N_2)}(\textbf{k},\boldsymbol{\ell}) ~q^{\Delta_{e}(\textbf{k},\boldsymbol{\ell})+s_{e}(\textbf{k},\boldsymbol{\ell})},
\vspace{-1mm}\end{align*}
where the numbers $\sigma_q\bigl(N_1,m_{\floor{\frac{e-1}{2}}}(\textbf{k})\bigr)$\textquotesingle s and $\sigma_q\bigl(N_2,m_{\floor{\frac{e-2}{2}}}(\boldsymbol{\ell})\bigr)$\textquotesingle s are given by \eqref{sigmaE}, and $\Delta_{e}(\textbf{k},\boldsymbol{\ell})$\textquotesingle s, $s_{e}(\textbf{k},\boldsymbol{\ell})$\textquotesingle s and $\mathscr{B}_{e,q}^{(N_1,N_2)}(\textbf{k},\boldsymbol{\ell})$\textquotesingle s are given by \eqref{defdeltaE}, \eqref{defskl} and \eqref{Beab}, respectively.
\end{thm}
\begin{proof}
It follows immediately from Theorem \ref{Thm6.2}.
\end{proof}
\vspace{-3mm}\begin{example}\label{ex4.5} Using Magma, we see that there are precisely $1065$ distinct non-zero Euclidean self-orthogonal $\mathbb{Z}_{3^4}\mathbb{Z}_{3^3}$-linear codes of block-length $(2,2).$ This agrees with Theorem \ref{Thm6.3}.
\end{example}
In the following corollary, we count all self-orthogonal additive codes of length $N$ over $\mathcal{R}_e$ when $e \geq 4.$
\vspace{-1mm}\begin{cor}\label{Cor3.1}
For $e\geq4,$ the number of self-orthogonal additive codes of length $N$ over $\mathcal{R}_e$ is given by
\vspace{-1mm}\begin{align*}
&\sum\limits_{\textbf{k}\in \mathscr{K}_{Nrt,e}}~\sum\limits_{\boldsymbol{\ell}\in \mathscr{K}_{Nr(k-t),e-1}}\sigma_p\bigl(Nrt,m_{\floor{\frac{e-1}{2}}}(\textbf{k})\bigr)\sigma_p\bigl(Nr(k-t),m_{\floor{\frac{e-2}{2}}}(\boldsymbol{\ell})\bigr)\mathscr{B}_{e,p}^{\bigl(Nrt,Nr(k-t)\bigr)}(\textbf{k},\boldsymbol{\ell})~p^{\Delta_{e}(\textbf{k},\boldsymbol{\ell})+s_{e}(\textbf{k},\boldsymbol{\ell})},
\vspace{-1mm}\end{align*}
where the numbers $\sigma_p\bigl(Nrt,m_{\floor{\frac{e-1}{2}}}(\textbf{k})\bigr)$\textquotesingle s and $\sigma_p\bigl(Nr(k-t),m_{\floor{\frac{e-2}{2}}}(\boldsymbol{\ell})\bigr)$\textquotesingle s are given by \eqref{sigmaE}, and $\Delta_{e}(\textbf{k},\boldsymbol{\ell})$\textquotesingle s,  $s_{e}(\textbf{k},\boldsymbol{\ell})$\textquotesingle s and $\mathscr{B}_{e,p}^{\bigl(Nrt,Nr(k-t)\bigr)}(\textbf{k},\boldsymbol{\ell})$\textquotesingle s are given by \eqref{defdeltaE}, \eqref{defskl} and \eqref{Beab}, respectively.
\end{cor}
\begin{proof}It follows by  Remark \ref{KEY} and taking $N_1=Nrt$ and $N_2=Nr(k-t)$ in Theorem  \ref{Thm6.3}.\end{proof}
\vspace{-3mm}\begin{example}By Corollary \ref{Cor3.1}, we see that 
there are precisely $1065$ distinct non-zero self-orthogonal additive codes of length $2$ over $\mathbb{Z}_{3^4}[y]/\langle y^2-3,3^3y \rangle.$ This agrees with Example \ref{ex4.5} and Magma computations.
\end{example}
We will next count all Euclidean self-dual $\mathtt{R}_e\mathtt{R}_{e-1}$-linear codes of block-length $(N_1,N_2).$ To do this, we assume, throughout this paper, that for positive integers $n$ and $s$,  the set $\mathscr{L}_{n,s}$ consists  of all $s$-tuples $(k_0,k_1,\ldots,k_{s-1})$ of non-negative integers satisfying $k_i=k_{s-i}$ for $1\leq i\leq s-1$ and $2(k_0+k_1+\cdots+k_{\floor{\frac{s-1}{2}}})+\delta^{(s)}k_{\floor{\frac{s}{2}}}= n,$ where $\delta^{(s)}=0$ if $s$ is odd, while $\delta^{(s)}=1$ if $s$ is even.
One can easily observe that $\mathscr{L}_{n,s}\subseteq\mathscr{K}_{n,s}.$ Further, for an $e$-tuple $\textbf{k}=(k_0,k_1,\ldots,k_{e-1})\in\mathscr{L}_{N_1,e},$ an  $(e-1)$-tuple $\boldsymbol{\ell}=(\ell_0,\ell_1,\ldots,\ell_{e-2})\in\mathscr{L}_{N_2,e-1}$ and a prime power $Q_1,$ let us define the number
\vspace{-2mm}\small{\begin{align}\label{Eeab}
\mathscr{F}_{e,Q_1}^{(N_1,N_2)}(\textbf{k},\boldsymbol{\ell})=\prod\limits_{i=1}^{\floor{\frac{e-1}{2}}}
\qbin{m_i(\textbf{k})}{k_i}{Q_1}
\prod\limits_{j=1}^{\floor{\frac{e-2}{2}}}
\qbin{m_j(\boldsymbol{\ell})}{\ell_j}{Q_1}.
\vspace{-2mm}\end{align}}
\normalsize

In the following theorem, we consider the case $e \geq 4$ and derive a necessary and sufficient condition for the existence of a Euclidean self-dual $\mathtt{R}_e\mathtt{R}_{e-1}$-linear code of block-length $(N_1,N_2).$ We also  obtain an enumeration formula for $\mathfrak{M}_{e}(N_1,N_2).$
\vspace{-1mm}\begin{thm}\label{Thm6.4}
Let $e \geq 4$ be a fixed integer. There exists a Euclidean self-dual $\mathtt{R}_e\mathtt{R}_{e-1}$-linear code of block-length $(N_1,N_2)$ if and only if either $(i)$ $e$ is odd, $N_1$ is even and $(-1)^\frac{N_1}{2}$ is a square in $\mathcal{T}_e,$ or $(ii)$ both $e$ and $N_2$ are even and $(-1)^\frac{N_2}{2}$ is a square in $\mathcal{T}_e.$ Furthermore, when either $(i)$ $e$ is odd, $N_1$ is even and $(-1)^\frac{N_1}{2}$ is a square in $\mathcal{T}_e,$ or $(ii)$ both $e$ and $N_2$ are even and $(-1)^\frac{N_2}{2}$ is a square in $\mathcal{T}_e,$ we have
\vspace{-1mm}\begin{align*}
\mathfrak{M}_{e}(N_1,N_2)=\sum\limits_{\textbf{k}\in\mathscr{L}_{N_1,e}}~\sum\limits_{\boldsymbol{\ell}\in\mathscr{L}_{N_2,e-1}}\sigma_q\bigl(N_1,m_{\floor{\frac{e-1}{2}}}(\textbf{k})\bigr)\sigma_q\bigl(N_2,m_{\floor{\frac{e-2}{2}}}(\boldsymbol{\ell})\bigr)\mathscr{F}_{e,q}^{(N_1,N_2)}(\textbf{k},\boldsymbol{\ell})~
q^{\Delta_{e}(\textbf{k},\boldsymbol{\ell})+s_e(\textbf{k},\boldsymbol{\ell})},
\vspace{-1mm}\end{align*}
where the numbers $\sigma_q\bigl(N_1,m_{\floor{\frac{e-1}{2}}}(\textbf{k})\bigr)$\textquotesingle s and $\sigma_q\bigl(N_2,m_{\floor{\frac{e-2}{2}}}(\boldsymbol{\ell})\bigr)$\textquotesingle s are given by \eqref{sigmaE}, and $\Delta_{e}(\textbf{k},\boldsymbol{\ell})$\textquotesingle s, $s_e(\textbf{k},\boldsymbol{\ell})$\textquotesingle s and $\mathscr{F}_{e,q}^{(N_1,N_2)}(\textbf{k},\boldsymbol{\ell})$\textquotesingle s are given by \eqref{defdeltaE}, \eqref{defskl} and \eqref{Eeab}, respectively.
\end{thm}
\begin{proof}
It follows from Theorem \ref{Thm6.3}, and Remarks \ref{Remark1.1} and \ref{Remark2.2}.
\end{proof}
In the following corollary, we consider the case $e \geq 4$ and derive a necessary and sufficient condition for the existence of a self-dual additive code of length $N$ over $\mathcal{R}_e.$ We also count all  self-dual additive codes of length $N$ over $\mathcal{R}_e.$
\vspace{-1mm}\begin{cor}\label{Cor3.2}
Let $e \geq 4$ be a fixed integer. There exists a self-dual additive code of length $N$ over $\mathcal{R}_e$ if and only if either $(i)$ $e$ is odd, $Nrt$ is even and $(-1)^\frac{Nrt}{2}$ is a square in $\mathbb{Z}_p,$ or $(ii)$ both $e$ and $Nr(k-t)$ are even and $(-1)^\frac{Nr(k-t)}{2}$ is a square in $\mathbb{Z}_p.$ Furthermore, when either $(i)$ $e$ is odd, $Nrt$ is even and $(-1)^\frac{Nrt}{2}$ is a square in $\mathbb{Z}_p,$ or $(ii)$ both $e$ and $Nr(k-t)$ are even and $(-1)^\frac{Nr(k-t)}{2}$ is a square in $\mathbb{Z}_p,$ the number of self-dual additive codes of length $N$ over $\mathcal{R}_e$ is given by
\vspace{-1mm}\begin{align*}
&\sum\limits_{\textbf{k}\in\mathscr{L}_{Nrt,e}}\sum\limits_{\boldsymbol{\ell}\in\mathscr{L}_{Nr(k-t),e-1}}\sigma_p\bigl(Nrt,m_{\floor{\frac{e-1}{2}}}(\textbf{k})\bigr)\sigma_p\bigl(Nr(k-t),m_{\floor{\frac{e-2}{2}}}(\boldsymbol{\ell})\bigr)\mathscr{F}_{e,p}^{\bigl(Nrt,Nr(k-t)\bigr)}(\textbf{k},\boldsymbol{\ell})
p^{\Delta_{e}(\textbf{k},\boldsymbol{\ell})+s_e(\textbf{k},\boldsymbol{\ell})},
\vspace{-1mm}\end{align*}
where the numbers $\sigma_p\big(Nrt,m_{\floor{\frac{e-1}{2}}}(\textbf{k})\big)$\textquotesingle s and $\sigma_p\big(Nr(k-t),m_{\floor{\frac{e-2}{2}}}(\boldsymbol{\ell})\big)$\textquotesingle s are given by \eqref{sigmaE}, and $\Delta_{e}(\textbf{k},\boldsymbol{\ell})$\textquotesingle s, $s_e(\textbf{k},\boldsymbol{\ell})$\textquotesingle s and  $\mathscr{F}_{e,p}^{\bigl(Nrt,Nr(k-t)\bigr)}(\textbf{k},\boldsymbol{\ell})$\textquotesingle s are given by \eqref{defdeltaE}, \eqref{defskl} and \eqref{Eeab}, respectively.
\end{cor}
\begin{proof} It follows by Remark \ref{KEY} and taking $N_1=Nrt$ and $N_2=Nr(k-t)$ in Theorem  \ref{Thm6.4}.\end{proof}

\vspace{-5mm}\section{Enumeration of  ACD codes of length $N$ over $\mathcal{R}_e$}\label{ACD}
\vspace{-2mm}In this section, we will count all ACD codes of length $N$ over $\mathcal{R}_e.$ For this, we will apply Theorem 3.5 of Jose and Sharma \cite{Jose2} to obtain an explicit enumeration formula for all Euclidean $\mathbb{Z}_{p^e}\mathbb{Z}_{p^{e-1}}$-LCD codes of an arbitrary block-length.  To state this result, let $L_p(n,s)$ be the number of distinct Euclidean LCD codes of  length $n$ and dimension $s$ over the finite field $\mathbb{Z}_p$ for $0\leq s\leq n,$ where $n$ is a positive integer. It is clear that $L_p(n,0)=L_p(n,n)=1.$ Now for $1\leq s\leq n-1,$ we see, by Theorems 3.2 and 3.3 of Yadav and Sharma \cite{Yadav2}, that the number $L_p(n,s)$ is given by the following:
\vspace{-2mm}\begin{itemize}
    \item[(a)] When $p=2,$ we have 
\begin{equation}\hspace{-3.8cm}\label{lqne}L_2(n,s)=\left\{\begin{array}{llll}  2^{\frac{(n-s)(s+1)}{2}}\qbin{(n-1)/2}{(s-1)/2}{4}  & \text{if both } n \text{ and } s \text{ are odd}; \\
 2^{\frac{ns-s^2+n-1}{2}}\qbin{(n-2)/2}{(s-1)/2}{4} & \text{if } n \text{ is even and } s \text{ is odd;}\\
2^{\frac{s(n-s+1)}{2}}\; \qbin{(n-1)/2}{s/2}{4}   & \text{if } n \text{ is odd and } s \text{ is even};\\
2^{\frac{ns-s^2-2}{2}}\Big( ( 2^{s}+1)\qbin{(n-2)/2}{s/2}{4}\\+( 2^{n-s+1}- 2^{n-s}+1)\qbin{(n-2)/2}{(s-2)/2}{4}\Big) & \text{if both } n \text{ and } s \text{ are even}.
  \end{array} \right.
\vspace{-2mm}\end{equation}
 \item[(b)] When $p$ is an odd prime, we have 
   \vspace{-2mm} \begin{equation}\label{lqno}L_p(n,s)=\left\{\begin{array}{llll}   p^{\frac{(n-s)(s+1)}{2}}\qbin{(n-1)/2}{(s-1)/2}{p^{2}}   & \text{if both } s \text{ and } n \text{ are odd}; \\
 p^{\frac{ns-s^2-1}{2}} ( p^{\frac{n}{2}}-1)\qbin{(n-2)/2}{(s-1)/2}{p^{2}} & \hspace{-2mm}\begin{array}{l}\text{if } s \text{ is odd and } n \text{ is even } \text{with either } p \equiv 1 \Mod{4} \\ \text{ or } n \equiv 0 \Mod{4} \text{ and } p \equiv 3 \Mod{4};\end{array}
\\
 p^{\frac{ns-s^2-1}{2}}( p^{\frac{n}{2}}+1)\qbin{(n-2)/2}{(s-1)/2}{p^{2}}  & \hspace{-2mm}\begin{array}{l} \text{if } s \text{ is odd, } p \equiv 3 \Mod{4}\text{ and }
 n \equiv 2 \Mod{4};\end{array}
\\
 p^{\frac{s(n-s+1)}{2}} \qbin{(n-1)/2}{s/2}{p^{2}}&\text{if } s \text{ is even and } n \text{ is odd;}
\\
   p^{\frac{s(n-s)}{2}}\qbin{n/2}{s/2}{p^{2}} & \text{if both } s \text{ and } n \text{ are even}.
  \end{array} \right.
\vspace{-2mm}\end{equation}
\end{itemize}

Now, the following theorem provides an explicit enumeration formula for all Euclidean  $\mathbb{Z}_{p^e}\mathbb{Z}_{p^{e-1}}$-LCD codes of block-length $(N_1,N_2).$ 
\vspace{-1mm}\begin{thm}[\cite{Jose2}]\label{Thm5.100}
For $e \geq 2,$ the number of Euclidean $\mathbb{Z}_{p^e}\mathbb{Z}_{p^{e-1}}$-LCD codes of block-length $(N_1,N_2)$ is given by
\vspace{-5mm}\begin{align*}
\sum\limits_{i=0}^{N_1}\sum\limits_{j=0}^{N_2}L_p(N_1,i)L_p(N_2,j)p^{(N_1-i)(e-1)(i+j)+(N_2-j)\bigl((e-1)i+(e-2)j\bigr)},
\end{align*}
where the numbers $L_p(N_1,i)$\textquotesingle s and $L_p(N_2,j)$\textquotesingle s are given by \eqref{lqne} and \eqref{lqno}.
\end{thm}
\begin{proof}It follows from Theorem 3.5 of Jose and Sharma \cite{Jose2}.
\end{proof}
As a consequence of the above theorem, we obtain an  explicit enumeration formula for all ACD codes of length $N$ over $\mathcal{R}_e$ in the following corollary.
\vspace{-1mm}\begin{cor}\label{Cor8.1}
For $e \geq 2,$ the number of ACD codes of length $N$ over $\mathcal{R}_e$ is given by
\vspace{-1mm}\begin{align*}
\sum\limits_{i=0}^{Nrt}\sum\limits_{j=0}^{Nr(k-t)}L_p(Nrt,i)L_p\bigl(Nr(k-t),j\bigr)p^{(Nrt-i)(e-1)(i+j)+\bigl(Nr(k-t)-j\bigr)\bigl((e-1)i+(e-2)j\bigr)},\
\end{align*}
where the numbers $L_p(Nrt,i)$\textquotesingle s and $L_p\bigl(Nr(k-t),j\bigr)$\textquotesingle s are given by \eqref{lqne} and \eqref{lqno}.
\end{cor}
\begin{proof}It follows by Remark \ref{KEY} and taking $N_1=Nrt$ and $N_2=Nr(k-t)$ in Theorem  \ref{Thm5.100}.\end{proof}
\vspace{-4mm}\begin{example}
 Using Magma,  we see that
there are precisely $113$ and $883$ distinct  ACD codes of length $2$ over $\mathbb{Z}_4[y]/\langle y^2-2,2y \rangle$ and $\mathbb{Z}_9[y]/\langle y^2-3,3y \rangle,$ respectively.  This agrees with Corollary \ref{Cor8.1}.
\end{example}
\vspace{-4mm}\section{Classification of self-orthogonal, self-dual and ACD codes over $\mathcal{R}_e$}\label{Classification}
\vspace{-2mm}
Two additive codes of length $N$ over $\mathcal{R}_e$ are monomially equivalent if  one code can be obtained from the other by a combination of operations of the following two types: (A) Permutation of the $N$ coordinate positions of the code, and (B) Multiplication of the code symbols appearing in a given coordinate position by a unit in $\mathbb{Z}_{p^e}.$ A generator matrix of an additive code  of length $N$ over $\mathcal{R}_e$ is a matrix whose rows form its minimal generating set over $\mathbb{Z}_{p^e}$.
Thus two additive codes $\mathcal{C}_1$ and $\mathcal{C}_2$ of length $N$ over $\mathcal{R}_e$  with their respective generator matrices $\mathcal{G}_1$ and $\mathcal{G}_2$ are monomially equivalent if and only if $\mathcal{G}_1 =\mathcal{G}_2 \mathcal{U},$  where $\mathcal{U}$ is an $N\times N$ monomial matrix over $\mathbb{Z}_{p^e},$ (note that a monomial matrix over $\mathbb{Z}_{p^e}$ is a square matrix whose non-zero entries are units in $\mathbb{Z}_{p^e}$ and each of its rows and columns has exactly one non-zero entry).  We further see, by Theorem \ref{thm0.1}, that the map $\Psi$ induces a duality-preserving 1-1 correspondence between additive codes of length $N$ over $\mathcal{R}_e$ and $\mathbb{Z}_{p^e}\mathbb{Z}_{p^{e-1}}$-linear codes of block-length $\bigl(Nrt,Nr(k-t)\bigr).$ Accordingly, we will now translate the notion of monomial equivalence between additive codes of length $N$ over $\mathcal{R}_e$ to a suitable notion of equivalence between $\mathbb{Z}_{p^e}\mathbb{Z}_{p^{e-1}}$-linear codes of block-length $\bigl(Nrt,Nr(k-t)\bigr).$

To do this, let $\mathcal{C}_1$ and $\mathcal{C}_2$ be two additive codes of length $N$ over $\mathcal{R}_e$ with generator matrices $\mathcal{G}_1$ and $\mathcal{G}_2,$ respectively. Then the $\mathbb{Z}_{p^e}\mathbb{Z}_{p^{e-1}}$-linear codes $\Psi(\mathcal{C}_1)$ and $\Psi(\mathcal{C}_2)$ of block-length $\bigl(Nrt,Nr(k-t)\bigr)$ have generator matrices $\Psi(\mathcal{G}_1)=[A_1~|~B_1]$ and  $\Psi(\mathcal{G}_2)=[A_2~|~B_2],$ respectively, where $A_1$ and $A_2$ are matrices over $\mathbb{Z}_{p^e}$ with $Nrt$ columns, and $B_1$ and $B_2$ are matrices over $\mathbb{Z}_{p^{e-1}}$ with $Nr(k-t)$ columns. Now, the codes $\mathcal{C}_1$ and $\mathcal{C}_2$ are monomially equivalent if and only if there exists an $N \times N$ monomial matrix $\mathcal{U}$ over $\mathbb{Z}_{p^e}$ satisfying \vspace{-2mm}\begin{equation}\label{mono}\mathcal{G}_1=\mathcal{G}_2 \mathcal{U}.\vspace{-2mm}\end{equation}
From this and using the fact that the map $\Psi$ is a $\mathbb{Z}_{p^e}$-module isomorphism, we see that the matrix equation \eqref{mono} is equivalent to  \vspace{-2mm}\begin{equation}\label{mono2}\Psi(\mathcal{G}_1)=\Psi(\mathcal{G}_2) \Psi(\mathcal{U}).\vspace{-2mm}\end{equation} Further, if $u_{i,j}$ denotes the $(i,j)$-th entry of $\mathcal{U}$ and $u'_{i,j}\equiv u_{i,j} ~(\text{mod }p^{e-1})$ for $1 \leq i, j \leq N,$  then the matrix $\Psi(\mathcal{U}) $ is given by 
\vspace{-1mm}$$\Psi(\mathcal{U})=\left[ \begin{array}{cccc|cccc}u_{1,1}I_{rt} & u_{1,2}I_{rt} & \cdots & u_{1,N}I_{rt}& \textbf{0} & \textbf{0} & \cdots & \textbf{0}\\
u_{2,1}I_{rt} & u_{2,2}I_{rt} & \cdots & u_{2,N}I_{rt}& \textbf{0} & \textbf{0} & \cdots & \textbf{0}\\
\vdots & \vdots & \ddots & \vdots& \vdots & \vdots & \ddots & \vdots\\u_{N,1}I_{rt} & u_{N,2}I_{rt} & \cdots & u_{N,N}I_{rt}& \textbf{0} & \textbf{0} & \cdots & \textbf{0}\\
\textbf{0} & \textbf{0} & \cdots & \textbf{0}& u'_{1,1}I_{r(k-t)} & u'_{1,2}I_{r(k-t)} & \cdots & u'_{1,N}I_{r(k-t)}\\
\textbf{0} & \textbf{0} & \cdots & \textbf{0} & u'_{2,1}I_{r(k-t)} & u'_{2,2}I_{r(k-t)} & \cdots & u'_{2,N}I_{r(k-t)}\\
\vdots & \vdots & \ddots & \vdots& \vdots & \vdots & \ddots & \vdots\\\textbf{0} & \textbf{0} & \cdots & \textbf{0} & u'_{N,1}I_{r(k-t)} & u'_{N,2}I_{r(k-t)} & \cdots & u'_{N,N}I_{r(k-t)}
\end{array}\right]=\left[\begin{array}{c|c}\mathcal{U}_1 & \textbf{0}\\ \textbf{0} & \mathcal{U}_2\end{array}\right],\vspace{-1mm}$$ and we  say that the matrix  $\left[\begin{array}{c|c}\mathcal{U}_1 & \textbf{0}\\ \textbf{0} & \mathcal{U}_2\end{array}\right]$ is  a 
monomial-type matrix, where $\textbf{0}$ denotes the zero matrix, and $I_{rt}$ and $I_{r(k-t)}$ denote $rt\times rt $ and $r(k-t) \times r(k-t)$ identity matrices,  respectively.    From this and by \eqref{mono2}, we see that the codes $\mathcal{C}_1$ and $\mathcal{C}_2$  are monomially equivalent if and only if there exists a monomial-type matrix $\left[\begin{array}{c|c}\mathcal{U}_1 & \textbf{0}\\ \textbf{0} & \mathcal{U}_2\end{array}\right]$ satisfying $[A_1~|~B_1]=[A_2~|~B_2]\left[\begin{array}{c|c}\mathcal{U}_1 & \textbf{0}\\ \textbf{0} & \mathcal{U}_2\end{array}\right]=[A_2\mathcal{U}_1~|~B_2 \mathcal{U}_2].$  Accordingly, we say that two $\mathbb{Z}_{p^e}\mathbb{Z}_{p^{e-1}}$-linear codes $\mathtt{C}_1$ and $\mathtt{C}_2$ of block-length $\bigl(Nrt,Nr(k-t)\bigr)$ with generator matrices $[A_1~|~B_1]$ and $[A_2~|~B_2]$ are $\ast$-equivalent if there exists a monomial-type matrix $\left[\begin{array}{c|c}\mathcal{U}_1 & \textbf{0}\\ \textbf{0} & \mathcal{U}_2\end{array}\right] $ satisfying $[A_1~|~B_1]=[A_2~|~B_2]\left[\begin{array}{c|c}\mathcal{U}_1 & \textbf{0}\\ \textbf{0} & \mathcal{U}_2\end{array}\right].$ From this, we deduce the following:
\vspace{-1mm}\begin{lemma}\label{lemeq} Two additive codes $\mathcal{C}_1$ and $\mathcal{C}_2$ of length $N$ over $\mathcal{R}_e$ are monomially equivalent if and only if $\Psi(\mathcal{C}_1)$ and $\Psi(\mathcal{C}_2)$ are $\ast$-equivalent.
\vspace{-1mm}\end{lemma}
From the above lemma, we see that classifying all self-orthogonal and self-dual additive codes and ACD codes of length $N$ over $\mathcal{R}_e$ up to monomial equivalence is equivalent to classifying all Euclidean self-orthogonal and self-dual $\mathbb{Z}_{p^e}\mathbb{Z}_{p^{e-1}}$-linear codes and Euclidean  $\mathbb{Z}_{p^e}\mathbb{Z}_{p^{e-1}}$-LCD codes of block-length $\bigl(Nrt,Nr(k-t)\bigr)$ up to $\ast$-equivalence, respectively. Towards this, the enumeration formulae  obtained in Sections \ref{SectionEnumeration} and \ref{ACD} are useful in classifying the aforementioned classes of codes up  equivalence. We will now illustrate this in certain specific cases   by carrying out computations in Magma.

Throughout this section, for a prime $p,$ let $\mathscr{R}_p$ be the finite commutative chain ring $\mathbb{Z}_{p^2}[y]/\langle y^2-p,py\rangle.$ By Theorem \ref{thm0.1} and Lemma \ref{lemeq}, we see that  classifying all self-orthogonal and self-dual additive codes of length $N$ over $\mathscr{R}_p$ up to monomial equivalence is equivalent to classifying all Euclidean self-orthogonal and  self-dual $\mathbb{Z}_{p^2}\mathbb{Z}_{p}$-linear codes of block-length $(N,N)$  up to $\ast$-equivalence. For this, we first note that each element $\alpha \in \mathscr{R}_p$ can be uniquely expressed as $\alpha=\alpha_0+\alpha_1y,$ where $\alpha_0 \in \mathbb{Z}_{p^2}$ and $\alpha_1 \in \mathbb{Z}_p.$ Further, we see, by Greferath and Schmidt \cite[p. 1]{Greferath1999}, that the homogeneous weight on $\mathscr{R}_p$  is defined as \vspace{-2mm}$$w_{hom}(\alpha)=w_{hom}(\alpha_0+\alpha_1y)=\left\{\begin{array}{cl} 0 & \text{if } \alpha=0;\\
p^2 & \text{if } \alpha_0\in\{p,2p,\ldots, (p-1)p\}~\text{ and }~\alpha_1=0;\\
(p-1)p & \text{otherwise,}\end{array}\right.\vspace{-2mm}$$
where $\alpha=\alpha_0+\alpha_1y~$ with $~\alpha_0 \in \mathbb{Z}_{p^2}$ and $\alpha_1 \in \mathbb{Z}_p.$
Further, the homogeneous weight of a word  $\textbf{a}=(a_0,a_1,\ldots,a_{N-1})\in \mathscr{R}_p^N$, denoted by $\textbf{w}_{hom}(\textbf{a}),$ is defined as $\textbf{w}_{hom}(\textbf{a})=\sum\limits_{i=0}^{N-1}w_{hom}(a_i).$ The homogeneous distance between the words $\textbf{a}, \textbf{b} \in \mathscr{R}_p^N$ is defined as $d_{hom}(\textbf{a}, \textbf{b})=\textbf{w}_{hom}(\textbf{a}- \textbf{b}).$ The homogeneous distance of an additive code $\mathcal{C}$ of length $N$ over $\mathscr{R}_p$ is defined as $d_{hom}(\mathcal{C})=\min\{ d_{hom}(\textbf{a}, \textbf{b}): \textbf{a}, \textbf{b} \in \mathcal{C} \text{ and }\textbf{a}\neq \textbf{b}\}= \min\{\textbf{w}_{hom}(\textbf{c}): \textbf{c}(\neq \textbf{0}) \in \mathcal{C}\}.$ Accordingly, we define the homogenous weight on the mixed alphabet $\mathbb{Z}_{p^2}\oplus \mathbb{Z}_p$ as $w_{Hom}(\alpha_0| \alpha_1):=w_{hom}(\alpha_0+\alpha_1 y)$ for all $(\alpha_0| \alpha_1) \in \mathbb{Z}_{p^2}\oplus \mathbb{Z}_p.$ This can be further extended to  the mixed alphabet $\mathbb{Z}_{p^2}^N\oplus \mathbb{Z}_p^N$ as 
\vspace{-5mm}$$\textbf{w}_{Hom}(\alpha_0,\alpha_1,\ldots,\alpha_{N-1}|\beta_0,\beta_1,\ldots,\beta_{N-1})=\sum\limits_{i=0}^{N-1}w_{Hom}(\alpha_i|\beta_i)=\sum\limits_{i=0}^{N-1}w_{hom}(\alpha_i+\beta_iy)\vspace{-2mm}$$ for all $(\alpha_0,\alpha_1,\ldots,\alpha_{N-1}|\beta_0,\beta_1,\ldots,\beta_{N-1}) \in \mathbb{Z}_{p^2}^N\oplus \mathbb{Z}_p^N.$ The homogeneous distance between the words $(\textbf{a}|\textbf{b}), (\textbf{c}|\textbf{d}) \in \mathbb{Z}_{p^2}^N\oplus \mathbb{Z}_p^N$ is defined as $d_{Hom}\big((\textbf{a}|\textbf{b}), (\textbf{c}|\textbf{d})\big)= \textbf{w}_{Hom}(\textbf{a}-\textbf{c}|\textbf{b}-\textbf{d}).$ Now the homogenous distance of a $\mathbb{Z}_{p^2}\mathbb{Z}_p$-linear code $\mathscr{C}$ of block-length $(N,N)$ is defined as 
$d_{Hom}(\mathscr{C})=\min\{d_{Hom}\big((\textbf{a}|\textbf{b}), (\textbf{c}|\textbf{d})\big): (\textbf{a}|\textbf{b}), (\textbf{c}|\textbf{d}) \in \mathscr{C} \text{ and } (\textbf{a}|\textbf{b})\neq  (\textbf{c}|\textbf{d})\}=\min\{ \textbf{w}_{Hom}(\textbf{a}|\textbf{b}): (\textbf{a}|\textbf{b}) (\neq \textbf{0}) \in \mathscr{C}\}.$ 
From the above discussion and by Theorem \ref{thm0.1}, we see that for an additive code $\mathcal{C}$ of length $N$ over $\mathscr{R}_p,$ we have $d_{hom}(\mathcal{C})=d_{Hom}(\Psi(\mathcal{C})).$ In view of this, we see that the $\mathbb{Z}_{p^e}$-modules $(\mathscr{R}_p^{N},d_{hom})$ and $(\mathbb{Z}_{p^2}^N \oplus \mathbb{Z}_{p}^N,d_{Hom})$ are isometric.

First of all, we will consider the case $p=3$ and classify all self-orthogonal and  self-dual additive codes of lengths $2$ and $3$ over $\mathscr{R}_3$ up to monomial equivalence by classifying all Euclidean self-orthogonal and self-dual  $\mathbb{Z}_{9}\mathbb{Z}_{3}$-linear codes of block-lengths $(2,2)$ and $(3,3),$ respectively, up to $\ast$-equivalence, by carrying out computations in Magma and applying Theorems \ref{Thm4.3} and \ref{Lemma3.3}.  Here, generator matrices of all $\mathbb{Z}_{9} \mathbb{Z}_{3}$-linear codes for which the corresponding additive code over $\mathscr{R}_3$ achieves the Plotkin bound for homogeneous weights (see Greferath and O\textquotesingle Sullivan \cite[Th. 2.2]{Greferath}) are marked $\dagger.$ 
\vspace{-1mm}\begin{itemize}
\vspace{-1mm}\item[I.] There are $3$ $\ast$-inequivalent non-zero Euclidean self-orthogonal $\mathbb{Z}_9\mathbb{Z}_3$-linear codes of block-length $(2,2).$ Among these codes, there are
\begin{enumerate}
\vspace{-1mm}\item[$\odot$] $2$ Euclidean self-orthogonal $\mathbb{Z}_9\mathbb{Z}_3$-linear codes of homogenous weight $9$ with generator matrices $[3~0~|~0~0]$ and $\left[\begin{array}{cc|cc}
     3&0&0&0  \\  
     0&3&0&0 
\end{array}\right];$ and
\vspace{-1mm}\item[$\odot$] $1$ Euclidean self-orthogonal $\mathbb{Z}_9\mathbb{Z}_3$-linear code of homogenous weight $18$ with a generator matrices ${}^{\dagger}[3~6~|~0~0].$
\end{enumerate}
\vspace{-1mm}\item[II.] By Theorem \ref{Lemma3.3}, we see that there does not exist any Euclidean self-dual $\mathbb{Z}_9\mathbb{Z}_3$-linear code of block-length $(2,2).$
\vspace{-1mm}\item[III.] There are $117$ $\ast$-inequivalent non-zero Euclidean self-orthogonal $\mathbb{Z}_9\mathbb{Z}_3$-linear codes of block-length $(3,3).$ Among these codes, there are
\begin{enumerate}
\vspace{-1mm}\item[$\odot$] $19$ Euclidean self-orthogonal $\mathbb{Z}_9\mathbb{Z}_3$-linear codes of homogenous weight $9$ with generator matrices as listed below:
\\$\bullet$ $[3~0~0~|~0~0~0]$ and $\left[\begin{array}{ccc|ccc}
     3&0&0&0&0&0\\
     0&3&0&0&0&0\\
     0&0&3&0&0&0\\
     0&0&0&1&1&1
\end{array}\right];$
\\$\bullet$ $\left[\begin{array}{ccc|ccc}
     3&0&x&0&0&0\\
     0&y&z&u&v&w
\end{array}\right]$ with $(x,y,z,u,v,w)\in\{(0,3,0,0,0,0),(3,3,0,0,0,0),(0,3,6,2,2,1),(0,3,0,\\2,2,2),(0,3,6,2,2,2),(0,0,0,1,1,2),(0,3,0,1,1,2),(0,3,6,1,1,2)\};$
\\$\bullet$ $\left[\begin{array}{ccc|ccc}
     3&0&0&0&0&0\\
     0&3&x&y&z&u\\
     0&0&v&w&t&s
\end{array}\right]$ with $(x,y,z,u,v,w,t,s)\in\{(0,0,0,0,3,0,0,0),(3,0,0,0,0,1,2,2),(0,1,2,\\1,3,1,2,1),(0,0,0,0,0,1,2,1),(0,0,0,0,3,2,1,2),(0,1,1,1,3,1,1,1),(0,0,0,0,3,1,2,1),(0,1,2,2,3,1,\\2,2),(3,0,0,0,0,1,2,1)\};$
\vspace{-1mm}\item[$\odot$] $97$ Euclidean self-orthogonal $\mathbb{Z}_9\mathbb{Z}_3$-linear codes of homogenous weight $18$ with generator matrices as listed below:
\\$\bullet$ $[x~y~z~|~u~v~w]$ with $(x,y,z,u,v,w)\in\{(3,0,6,0,0,0),(3,3,6,2,1,2),(3,3,0,1,1,2),(0,0,0,1,1,1),\\(3,3,0,2,2,2),(0,3,0,1,2,1),(3,3,6,1,1,2),(3,3,6,2,2,2),(0,3,0,1,1,1),(3,3,0,2,1,2),(3,3,6,2,2,\\1),
(1,4,1,1,1,2),(1,7,8,0,2,0),(1,2,4,2,0,2),(1,4,1,1,1,1),(1,2,4,2,0,1),(1,4,1,0,0,0),(1,4,1,1,\\2,2),(1,7,8,0,1,0),(1,4,1,2,2,2),(1,2,4,1,2,0)\};$ 
\\$\bullet$ $\left[\begin{array}{ccc|ccc}
     3&0&x&y&z&u\\
     0&v&w&x_1&y_1&z_1
\end{array}\right]$ with $(x,y,z,u,v,w,x_1,y_1,z_1)\in\{(6,0,0,0,3,0,2,2,2),(6,0,0,0,3,6,2,2,\\1),(6,0,0,0,3,0,2,2,1),(6,0,0,0,3,6,2,1,2),(3,2,2,1,3,0,1,1,2),(3,1,1,1,3,0,2,2,2),(6,0,0,0,3,6,\\2,2,2),(3,1,1,2,3,0,2,2,1),(3,0,0,0,3,6,0,0,0),(6,1,1,2,3,6,2,2,1),(3,1,2,1,3,0,2,1,2),(0,1,2,1,\\3,0,2,1,2),(0,2,1,1,3,0,1,2,2),(3,2,1,2,3,0,1,2,1),(6,0,0,0,3,0,1,1,2),(6,0,0,0,3,0,1,1,1),(6,1,\\1,1,3,6,2,2,2),(6,0,0,0,3,6,1,1,2),(6,0,0,0,3,6,1,2,1),(0,2,2,1,3,0,1,1,2),(6,0,0,0,3,6,1,1,1),\\(3,2,2,2,3,0,1,1,1)\};$
\\$\bullet$ $\left[\begin{array}{ccc|ccc}
     1&x&y&z&u&v\\
     0&w&x_1&y_1&z_1&u_1
\end{array}\right]$ with $(x,y,z,u,v,w,x_1,y_1,z_1,u_1)\in\{(5,2,1,2,0,0,3,2,1,2),(5,2,0,1,2,\\0,3,2,2,1),
(1,1,0,2,2,0,0,1,2,1),
(2,7,0,0,0,0,0,1,1,2),
(2,7,0,0,0,0,0,1,2,1),
(1,1,0,1,2,0,0,1,1,\\1),
(5,2,0,2,1,0,3,1,1,2),
(5,2,2,0,2,0,3,1,2,1),
(1,1,0,1,2,3,6,2,2,2),
(5,2,2,2,0,0,3,1,1,2),
(5,2,\\1,2,0,0,3,2,1,1),
(5,2,2,1,0,0,3,1,2,1),
(5,2,3,0,1,0,3,1,1,2),
(1,1,1,1,0,3,6,2,1,2),
(5,2,0,2,2,0,\\3,2,1,1),
(5,2,2,1,0,0,3,1,2,2),
(5,2,1,0,1,0,3,2,1,2),
(5,2,1,0,2,0,3,2,2,1),
(1,1,0,2,2,3,6,1,2,1),\\
(5,2,2,2,0,0,3,1,1,1),
(2,7,0,0,0,3,3,2,2,1),
(2,7,0,0,0,3,3,1,2,1),
(5,2,2,0,1,0,3,1,2,2),
(1,1,0,2,\\1,3,6,2,2,2),(1,7,1,0,0,3,6,0,0,0),(1,5,2,2,1,3,3,0,0,0),(1,5,2,2,2,3,3,0,0,0),(1,1,1,0,2,3,6,0,0,\\0),(1,1,2,2,0,3,6,0,0,0),(1,5,0,0,0,3,3,0,0,0),(1,5,1,1,2,3,3,0,0,0),(1,1,1,1,0,3,6,0,0,0),(1,7,0,\\2,0,3,6,0,0,0),(1,5,1,2,2,3,3,0,0,0)
\};$
\\$\bullet$ $\left[\begin{array}{ccc|ccc}
     3&6&x&0&0&0\\
     0&0&0&1&y&z
\end{array}\right]$ with $(x,y,z)\in\{(0,1,2),(0,2,1),(6,1,2),(6,2,2)\};$
\\$\bullet$ $\left[\begin{array}{ccc|ccc}
     1&x&y&z&u&v\\
     0&3&3&0&0&0\\
     0&0&w&x_1&y_1&z_1
\end{array}\right]$ with $(x,y,z,u,v,w,x_1,y_1,z_1,u_1)\in\{(2,7,0,0,0,0,1,2,2),(1,2,0,2,0,3,\\1,2,2),(1,2,0,1,0,3,1,1,2),(1,2,2,0,0,3,2,2,1),(1,8,0,1,2,0,1,2,2),(1,2,1,0,0,3,1,2,1),(2,7,0,0,\\0,0,1,2,1),(1,2,0,2,0,3,1,2,1),(1,8,0,1,1,0,1,1,2),(1,2,0,1,0,3,2,1,2)\};$
\\$\bullet$ $\left[\begin{array}{ccc|ccc}
     3&0&x&y&z&u\\
     0&3&v&w&x_1&y_1\\
     0&0&z_1&u_1&v_1&w_1
\end{array}\right]$ with $(x,y,z,u,v,w,x_1,y_1,z_1,u_1,v_1,w_1)\in\{(0,1,2,2,0,2,1,1,3,2,1,1),\\(0,2,1,2,0,1,2,1,3,1,2,1),(6,0,0,0,6,0,0,0,0,1,1,1),(0,2,1,1,0,1,2,2,3,1,2,2),(0,1,2,1,0,2,1,2,\\3,2,1,2),(6,0,0,0,6,0,0,0,0,1,2,1)
\};$ and
\vspace{-1mm}\item[$\odot$] $1$ Euclidean self-orthogonal $\mathbb{Z}_9\mathbb{Z}_3$-linear code of homogenous weight $27$ with a generator matrix \\${}^{\dagger}[3~6~6~|~0~0~0].$
\end{enumerate}
\vspace{-1mm}\item[IV.] By Theorem \ref{Lemma3.3}, we see that there does not exist any Euclidean self-dual $\mathbb{Z}_9\mathbb{Z}_3$-linear code of block-length $(3,3).$ 
\end{itemize}We will now   classify all ACD codes of length  $2$ over $\mathscr{R}_2$ and  $\mathscr{R}_3$  up to monomial equivalence by classifying all Euclidean   $\mathbb{Z}_{4}\mathbb{Z}_{4}$-LCD codes and $\mathbb{Z}_{9}\mathbb{Z}_{3}$-LCD codes of block-lengths $(2,2),$  respectively, up to $\ast$-equivalence, with the help of Magma and applying Theorem \ref{Thm5.100}.
\begin{itemize}
\vspace{-1mm}\item[V.] There are $203$ $\ast$-inequivalent non-zero Euclidean $\mathbb{Z}_9\mathbb{Z}_3$-LCD codes of block-length $(2,2).$ Among these codes, there are
\begin{enumerate}
\vspace{-1mm}\item[$\odot$] $114$ Euclidean $\mathbb{Z}_9\mathbb{Z}_3$-LCD codes of homogeneous weight $6$ with generator matrices as listed below:
\\$\bullet$ $[x~y~|~z~w]$ with $(x,y,z,w)\in\{(3,0,2,0),(3,0,1,0),(0,0,1,0),(1,0,0,0),(1,0,1,0),(1,0,2,0)\};$
\\$\bullet$ $\left[\begin{array}{cc|cc}
     x&y&z&u\\
     v&w&x_1&y_1
\end{array}\right]$ with $(x,y,z,u,v,w,x_1,y_1)\in\{(1,0,1,0,0,1,0,2),(1,0,0,0,0,1,1,0),(1,0,1,0,0,\\1,0,1),(1,0,2,0,0,1,0,0),(1,0,0,2,0,1,0,1),(1,0,0,0,0,1,0,0),(1,0,0,0,0,1,0,1),(1,0,1,0,0,1,2,1),\\(1,0,2,0,0,1,1,2),(1,0,2,2,0,1,0,1),(1,0,2,0,0,1,1,1),(1,0,2,0,0,1,2,0),(1,0,2,0,0,1,0,2),(1,0,1,\\2,0,1,0,0),(1,0,0,0,0,1,1,2),(1,2,0,0,0,3,0,2),(1,0,0,2,0,0,1,1),(1,0,2,0,0,3,1,0),(1,0,0,0,0,0,\\1,1),(1,6,1,0,0,0,0,1),(1,0,0,0,0,3,1,0),(1,0,2,0,0,3,0,2),(1,0,1,1,0,3,0,2),(1,0,0,0,0,3,0,2),(1,\\2,1,0,0,3,0,1),(1,8,0,1,0,0,1,0),(1,0,1,0,0,3,2,0),(1,0,1,2,0,3,0,1),(1,2,0,1,0,3,1,0),(1,0,1,0,0,\\3,0,1),(1,0,2,0,0,0,0,1),(1,0,2,0,0,3,1,1),(1,0,0,0,0,3,2,2),(1,0,0,0,0,0,0,1),(1,6,2,0,0,0,0,1),\\(1,2,0,2,0,3,2,0),(1,6,0,0,0,0,0,1),(1,0,1,0,0,3,2,1),(1,0,0,1,0,3,0,2),(1,0,0,2,0,0,1,0),(1,0,2,\\2,0,3,0,1),(1,6,0,2,0,0,1,0),(1,0,2,0,0,3,0,1),(1,0,0,2,0,3,0,1),(1,6,0,0,0,0,1,0),(1,0,0,0,0,3,0,\\1),(1,2,0,1,0,3,2,0),(1,0,0,0,0,3,2,1),(1,0,0,2,0,0,1,2),(1,6,0,1,0,0,1,0),(1,0,2,1,0,3,0,2),(1,0,\\0,0,0,0,1,0),(1,0,1,0,0,3,0,2),(1,8,0,2,0,0,1,0),(1,0,2,0,0,3,1,2),(1,8,0,0,0,0,1,0),(1,2,0,0,0,3,\\1,0),(1,0,1,0,0,3,2,2),(1,0,1,0,0,0,0,1),(3,0,1,0,0,3,2,1),(3,0,2,0,0,3,1,2),(3,0,2,0,0,3,1,1),(3,\\0,1,0,0,0,0,1),(3,0,1,0,0,3,0,1),(3,6,0,1,0,0,1,0),(3,0,0,2,0,0,1,0),(3,6,1,0,0,0,0,1),(3,0,2,0,0,\\3,0,2),(3,0,2,0,0,0,0,1),(3,0,2,0,0,3,0,1),(0,3,0,1,0,0,1,2),(0,0,1,0,0,0,0,1),(3,0,1,0,0,3,2,2),\\(0,3,0,2,0,0,1,1)\};$
\\$\bullet$ $\left[\begin{array}{c|c}
     I_2&0\\
     0&I_2
\end{array}\right];$
\\$\bullet$ $\left[\begin{array}{cc|cc}
     1&x&y&z\\
     0&u&v&w\\
     0&x_1&y_1&z_1
\end{array}\right]$ with $(x,y,z,u,v,w,x_1,y_1,z_1)\in\{(0,2,0,1,0,0,0,0,1),(0,0,2,1,0,2,0,1,0),(0,0,\\2,1,0,0,0,1,1),(0,0,2,1,0,1,0,1,0),(0,0,2,1,0,1,0,1,2),(0,0,0,1,2,0,0,0,1),(0,0,0,1,0,0,0,0,1),(0,\\0,0,1,0,1,0,1,0),(0,0,1,1,0,2,0,1,1),(0,0,0,1,0,2,0,1,1),(0,0,2,1,0,1,0,1,1),(0,0,0,1,0,0,0,1,2),\\(1,0,0,3,0,1,0,1,1),(0,0,0,0,1,0,0,0,1),(0,0,2,3,0,2,0,1,0),(0,0,0,3,0,1,0,1,0),(6,0,0,0,1,0,0,0,\\1),(0,0,2,3,0,2,0,1,2),(1,0,0,3,0,2,0,1,1),(0,0,1,3,0,1,0,1,0),(0,0,0,3,0,2,0,1,0),(7,0,0,0,1,0,0,\\0,1),(0,0,0,3,2,0,0,0,1),(1,0,0,3,0,2,0,1,2),(0,0,1,3,0,1,0,1,1),(0,1,0,3,1,0,0,0,1),(1,0,0,3,2,0,\\0,0,1),(0,0,0,3,0,1,0,1,2),(1,0,0,3,1,0,0,0,1),(0,2,0,3,2,0,0,0,1),(0,0,2,3,0,2,0,1,1),(0,0,1,3,0,\\1,0,1,2),(0,0,0,3,0,2,0,1,2)\};$
\vspace{-1mm}\item[$\odot$] $48$ Euclidean $\mathbb{Z}_9\mathbb{Z}_3$-LCD codes of homogeneous weight $9$ with generator matrices as listed below:
\\$\bullet$ $[1~x~|~y~z]$ with $(x,y,z)\in\{(6,0,2),(6,0,0),(6,0,1),(0,0,2),(6,1,2),(6,1,0),(6,1,1),(6,2,2),(0,1,2),\\(6,2,0),(6,2,1),(0,2,2)\};$
\\$\bullet$ $\left[\begin{array}{cc|cc}
     1&x&y&z\\
     0&u&v&w
\end{array}\right]$ with $(x,y,z,u,v,w)\in\{(0,1,2,1,1,1),(0,1,2,1,2,2),(0,2,2,1,1,2),(0,1,2,1,2,1),\\(0,1,1,1,1,0),(0,2,1,1,2,1),(0,2,2,1,1,0),(0,2,2,1,2,2),(0,1,1,1,2,0),(0,2,2,1,2,0),(0,0,2,1,1,\\0),(0,0,2,1,2,0),(0,2,2,3,1,0),(6,0,0,0,1,1),(0,1,2,3,2,0),(0,2,1,3,1,2),(0,0,1,3,1,2),(0,2,2,3,\\1,1),(0,1,1,3,2,2),(0,0,2,3,1,1),(0,2,1,3,1,0),(6,0,2,0,1,2),(0,1,2,3,2,1),(6,0,1,0,1,1),(0,1,1,\\3,2,0),(0,0,2,3,2,0),(0,2,1,3,2,2),(0,2,1,3,1,1),(0,0,1,3,2,2),(0,1,1,3,2,1),(0,0,2,3,2,1),(0,0,\\1,3,2,0),(6,0,2,0,1,1),(6,0,0,0,1,2),(0,1,1,3,1,2),(6,0,1,0,1,2)\};$
\vspace{-1mm}\item[$\odot$] $38$ Euclidean $\mathbb{Z}_9\mathbb{Z}_3$-LCD codes of homogeneous weight $12$ with generator matrices as listed below:
\\$\bullet$ $[x~y~|~z~u]$ with $(x,y,z,u)\in\{(1,4,0,2),(1,4,0,0),(1,4,1,2),(1,4,1,0),(1,4,1,1),(1,4,2,2),(3,6,2,2),\\(3,0,2,1),(3,0,1,2),(0,0,1,2),(3,6,2,1),(3,6,1,2)\};$
\\$\bullet$ $\left[\begin{array}{cc|cc}
     1&x&y&z\\
     0&u&v&w
\end{array}\right]$ with $(x,y,z,u,v,w)\in\{(8,0,2,0,1,1),(8,0,0,0,1,1),(2,0,2,3,2,2),(2,0,0,3,2,2),\\(2,2,0,3,1,1),(2,0,0,3,1,1),(8,0,1,0,1,1),(2,0,1,3,2,2),(2,0,0,3,2,1),(2,1,0,3,1,2),(8,0,2,0,1,2),\\(8,0,0,0,1,2),(2,2,0,3,1,2),(2,1,0,3,1,1)\};$
\\$\bullet$ $\left[\begin{array}{cc|cc}
     3&x&y&z\\
     0&u&v&w
\end{array}\right]$ with $(x,y,z,u,v,w)\in\{(0,1,2,3,1,1),(0,0,2,3,1,1),(6,0,2,0,1,1),(0,0,1,3,1,2),\\(0,0,2,3,2,1),(6,0,1,0,1,2),(0,0,1,3,2,2),(6,0,1,0,1,1),(0,0,2,3,1,0),(0,0,2,3,2,0),(0,1,2,3,2,2),\\(0,2,2,3,1,2)\};$ and
\vspace{-1mm}\item[$\odot$] $3$ Euclidean $\mathbb{Z}_9\mathbb{Z}_3$-LCD codes of homogeneous weight $15$ with generator matrices $[3~6~|~2~0],$ $[3~0~|~0~2]$ and $[3~6~|~1~0].$ 
\end{enumerate}
\vspace{-1mm}\item[VI.] There are $61$ $\ast$-inequivalent non-zero Euclidean  $\mathbb{Z}_4\mathbb{Z}_2$-LCD codes of block-length $(2,2).$ Among these codes, there are
\begin{enumerate}
\vspace{-1mm}\item[$\odot$] $45$ Euclidean $\mathbb{Z}_4\mathbb{Z}_2$-LCD codes of homogeneous weight $2$ with generator matrices as listed below:
\\$\bullet$ $[x~y~|~z~u]$ with $(x,y,z,u)\in\{(0,0,1,0),(0,2,0,1),(1,0,0,0),(1,0,1,0)\};$
\\$\bullet$ $\left[\begin{array}{cc|cc}
     x&y&z&u\\
     v&w&x_1&y_1
\end{array}\right]$ with $(x,y,z,u,v,w,x_1,y_1)\in\{(1,0,1,0,0,1,1,1),(1,0,0,1,0,1,0,1),(1,0,1,0,0,\\1,0,0),(1,0,0,0,0,1,1,1),(1,0,0,0,0,1,0,0),(1,0,0,0,0,1,1,0),(1,0,1,0,0,1,0,1),(1,2,1,0,0,0,0,1),\\(1,2,0,0,0,0,0,1),(1,0,1,0,0,2,0,1),(1,0,0,0,0,0,1,0),(1,0,0,0,0,1,0,1),(1,2,0,0,0,0,1,0),(1,0,1,\\1,0,2,0,1),(1,0,0,1,0,0,1,0),(1,0,1,0,0,2,1,0),(1,0,1,0,0,0,0,1),(1,0,0,1,0,2,0,1),(1,2,0,1,0,0,\\1,0),(1,0,0,0,0,2,1,0),(1,0,0,0,0,0,0,1),(2,0,0,1,0,0,1,1),(0,2,0,1,0,0,1,0),(0,0,1,0,0,0,0,1),(2,\\2,1,0,0,0,0,1),(2,0,1,1,0,2,0,1),(0,2,1,0,0,0,0,1),(2,0,1,0,0,2,0,1)\};$
\\$\bullet$ $\left[\begin{array}{c|c}
     I_2&0\\
     0&I_2
\end{array}\right];$
\\$\bullet$ $\left[\begin{array}{cc|cc}
     1&0&x&y\\
     0&z&u&v\\
     0&0&w&x_1
\end{array}\right]$ with $(x,y,z,u,v,w,x_1)\in\{(1,0,1,0,0,0,1),(0,0,1,0,0,1,0),(1,0,1,1,0,0,1),(0,0,\\1,1,0,0,1),(0,0,2,1,0,0,1),(0,0,2,0,1,1,1),(0,0,0,1,0,0,1),(0,0,2,0,1,1,0),(0,1,2,0,1,1,1),(0,0,0,\\1,0,0,1),(1,0,2,1,0,0,1),(0,1,2,0,1,1,0)\};$
\vspace{-1mm}\item[$\odot$] $14$ Euclidean $\mathbb{Z}_4\mathbb{Z}_2$-LCD codes of homogeneous weight $4$ with generator matrices as listed below:
\\$\bullet$ $[1~x~|~y~z]$ with $(x,y,z)\in\{(2,1,1),(2,0,0),(2,1,0),(0,0,1),(0,1,1),(2,0,1)\};$
\\$\bullet$ $\left[\begin{array}{cc|cc}
     x&y&z&u\\
     v&w&x_1&y_1
\end{array}\right]$ with $(x,y,z,u,v,w,x_1,y_1)\in\{(1,0,1,1,0,1,1,1),(1,0,1,1,0,1,1,0),(1,0,0,1,0,1,\\1,0),(1,0,1,1,0,2,1,0),(1,0,0,1,0,2,1,0),(2,0,0,1,0,2,1,0),(2,2,0,1,0,0,1,1),(2,0,1,1,0,2,1,0)\};$ and
\vspace{-1mm}\item[$\odot$] $2$ Euclidean $\mathbb{Z}_4\mathbb{Z}_2$-LCD codes of homogeneous weight $6$ with generator matrices $[2~2~|~1~0]$ and $[0~2~|~1~0].$
\end{enumerate}
\end{itemize}

\end{document}